\documentclass[11pt]{article}
\usepackage{amsmath,amssymb,amsthm,mathtools}
\usepackage[dvipsnames]{xcolor}
\usepackage{algorithm,Levi}
\usepackage{tocloft}

\usepackage{url}

\usepackage{lineno}

\usepackage[margin = 2.54cm]{geometry}

\def\FullBox{\hbox{\vrule width 8pt height 8pt depth 0pt}}
\newcommand{\QED}{\;\;\;\FullBox}
\renewenvironment{proof}{\noindent{\bf Proof:~}}{\hfill\QED}
\newenvironment{proofof}[1]{\noindent{\bf Proof of {#1}:~}}{\hfill\(\QED\)}

\makeatletter
\newtheorem*{rep@theorem}{\rep@title}
\newcommand{\newreptheorem}[2]{%
	\newenvironment{rep#1}[1]{%
		\def\rep@title{#2 \ref{##1}}%
		\begin{rep@theorem}}%
		{\end{rep@theorem}}}
\makeatother
\newreptheorem{theorem}{Theorem}

\newcommand{\indi}{\boldsymbol{1}}

\newcommand{\ix}{\mathrm{Ind}}

\newcommand{\Simulate}{\texttt{Simulate}}

\date{}

\begin{document}

\author{Sourav Chakraborty\thanks{Indian Statistical Institute, Kolkata, India. Research partially supported by the SERB project SQUID-1982-SC-4837. Email: \href{mailto:chakraborty.sourav@gmail.com}{chakraborty.sourav@gmail.com}.} \and Eldar Fischer\thanks{Technion - Israel Institute of Technology, Israel. Email: \href{mailto:eldar@cs.technion.ac.il}{eldar@cs.technion.ac.il}. Research supported by an Israel Science Foundation grant number 879/22.} \and Arijit Ghosh\thanks{Indian Statistical Institute, Kolkata, India. Research partially supported by the Science \& Engineering Research Board of the DST, India, through the MATRICS grant MTR/2023/001527. Email: \href{mailto: arijitiitkgpster@gmail.com}{ arijitiitkgpster@gmail.com}.}
\and Amit Levi\thanks{University of Haifa, Israel. Email: \href{mailto: alevi@cs.haifa.ac.il}{ alevi@cs.haifa.ac.il}.} \and Gopinath Mishra\thanks{National University of Singapore, Singapore. Email: \href{mailto:gopinath@nus.edu.sg}{gopinath@nus.edu.sg}.}\and Sayantan Sen\thanks{Centre for Quantum Technologies, National University of Singapore, Singapore. Research supported by the National Research Foundation, Singapore and $\mbox{A}^*$STAR under its Quantum Engineering Programme NRF2021-QEP2-02-P05. The author would like to thank Eldar Fischer for hosting him at Technion for an academic visit where this work was initiated. Email: \href{mailto:sayantan789@gmail.com}{sayantan789@gmail.com}.}}

\title{Testing vs Estimation for Index-Invariant Properties in the Huge Object Model} 
\maketitle

\thispagestyle{empty}
\begin{abstract}
The Huge Object model of property testing [Goldreich and Ron, TheoretiCS 23] concerns properties of distributions supported on $\{0,1\}^n$, where $n$ is so large that even reading a single sampled string is unrealistic. Instead, query access is provided to the samples, and the efficiency of the algorithm is measured by the total number of queries that were made to them.

Index-invariant properties under this model were defined in [Chakraborty et al., COLT 23], as a compromise between enduring the full intricacies of string testing when considering unconstrained properties, and giving up completely on the string structure when considering label-invariant properties. Index-invariant properties are those that are invariant through a consistent reordering of the bits of the involved strings.

Here we provide an adaptation of Szemer\'edi's regularity method for this setting, and in particular show that if an index-invariant property admits an $\epsilon$-test with a number of queries depending only on the proximity parameter $\epsilon$, then it also admits a distance estimation algorithm whose number of queries depends only on the approximation parameter.

\end{abstract}
\thispagestyle{empty}

\newpage
\thispagestyle{empty}
\setcounter{tocdepth}{2}
\setlength{\cftparskip}{5pt}
\tableofcontents
\thispagestyle{empty}
\newpage
\setcounter{page}{1}

\section{Introduction}
Distribution testing is a sub-field of property testing~\cite{RS96,GGR98} which has received a significant amount of attention in the past few decades (See e.g.,~\cite{Fis04a,R08,R10,CS10,RS11,G17,canonne2020survey,Can22,BY22}). In the classical setting of distribution testing, an algorithm can draw independent samples from an unknown input distribution $\mu$, and needs to either accept or reject it based on the drawn samples. An $\eps$-testing algorithm for some distribution property $\calP$ is required to accept every distribution satisfying $\calP$ with high probability, and reject every distribution that has distance of at least $\eps$ from any distribution satisfying $\calP$ with high probability, while drawing as few samples as possible.

One particular setting of interest, mainly due to its resemblance to real-world data, is the setting where the distribution is supported on a high dimensional product space. In the classical setting of distribution testing, it is implicitly assumed that the algorithm has complete access to the sampled elements. However, for product spaces with extremely large dimension, such access is impractical: Even obtaining the entire label of a single drawn element might be too expensive. 

To handle this setting, Goldreich and Ron~\cite{GR22} defined and initiated the \emph{Huge Object model}. In this model the algorithm can draw samples from the distribution, but since the samples are not fully accessible, the algorithm is provided with query access to the sampled strings. The distance notion is changed accordingly from the total variation distance to the \emph{Earth Mover's distance}. Various properties and algorithmic behaviors have been studied under the huge object model (See~\cite{GR22,CFGMS23,AF23,AFL23,canonnegrainedness25}).

When considering properties of distributions, it is commonly the case that the property at hand has some ``symmetry" in its structure. One such symmetry is \emph{label-invariance}. A property is label-invariant if applying any permutation to the labels does not effect membership in the property. Label-invariant properties have been extensively studied in the literature~\cite{BDKR05,P08,GR11,V11,DKN14,CDVV14,ADK15,VV17,BC17,DKS18, CFGMS22}.

Recently, Chakraborty et al.~\cite{CFGMS23} considered the general task of testing \emph{index-invariant} properties under the huge object model. A property of distributions supported on $\zo^n$ is called index-invariant if any permutation $\pi:[n]\to[n]$ of the indices of the strings does not effect membership in the property. Note that the class of index-invariant properties extends the class of label-invariant properties, since the invariance condition is weaker compared to label-invariance: It only considers relabeling strings by rearranging their symbols through a common permutation of the indices. In particular, index-invariance still allows for properties to have some string-related structure. For example, monotonicity of the probability mass function with respect to the natural partial order on $\zo^n$ is an index-invariant property but not a label-invariant one.

 The main result of \cite{CFGMS23} showed that the number of queries is controlled through upper and lower bounds depending on the VC-dimension of the support of the distributions in the property (along with the proximity parameter).
In addition, they showed that for properties that are not index-invariant such bounds are not guaranteed, and that there exists an almost tight quadratic gap between adaptive testers (i.e., the setting where queries might depend on answers to previous queries) and non-adaptive testers (i.e., the setting where queries are determined in advanced) of index-invariant properties. Note that such a distinction is non-existent in the classical model of distribution testing, as there is no notion of adaptivity when the algorithm is only allowed to draw independent samples from the distribution. 

In this work, we complement the result of \cite{CFGMS23} and show that for any index-invariant property that is $\eps$-testable using a constant (depending only on $\eps$) number of queries, there exists an estimation algorithm that approximates the distance of the input distribution from the property, by performing only a constant (depending only on $\eps$) number of queries. 

\begin{theorem}[Informal statement of Theorem~\ref{thm:main-estimation}]\label{thm:main_informal} If\; $\mathcal{P}$ is an index-invariant property of distributions over $\{0,1\}^n$ which is $\epsilon'$-testable with a constant number of queries for any fixed $\epsilon'>0$, then there is an estimation algorithm with a constant number of queries for any fixed $\epsilon>0$, that for a distribution $\mu$ of distance $d$ from $\calP$, outputs a distance approximation $\boldeta$ such that with high probability $|\boldeta-d|\le \epsilon$.
\end{theorem}

This brings index-invariant properties under the umbrella of properties for which constant-query testability implies constant-query estimability, similar to dense graph properties, and in contrast to general string properties and also general Huge Object model properties. Along the way, we provide a ``regularity-like'' framework to index-invariant properties, to treat them using a ``meta-testing'' scheme. More on this will be discussed in Section \ref{subs:tech}. It is important to note that the dependencies here are not a tower function of $1/\epsilon$, a distinction from most regularity-like frameworks in the literature, which usually provide a dependency which is a tower in the other parameters, or even worse. In particular, the query complexity in our result has double-exponential dependency on $1/\eps$.
Section \ref{subs:prior} contains more information on this and relations to prior works.

\subsection{Technical overview}
\label{subs:tech}

The proof of our main result is motivated by Szemer\'{e}di's regularity lemma for graphs~\cite{szemeredi1978regular}. Intuitively, Szemer\'edi’s regularity lemma states that a graph can be partitioned into a bounded number of parts such that the subgraphs between the parts are “random-like”. In addition, we take inspiration from~\cite{FN05} that proves an analogous estimation result for the dense graph testing model, and in particular we define and use a notion inspired by the notion of robust partitions from \cite{FN05}.

\subsubsection{Detailings} 
Our starting point is to define a notion that resembles that of a graph partition for the setting of distributions. For technical reasons, we define a notion that is slightly more generalized. A \emph{detailing} $\xi$ of a distribution $\mu$ over $\zo^n$ with respect to some set $A$ is a distribution on $\zo^n\times A$ such that the distribution on $\{0,1\}^n$ that is projected from $\xi$ is equal to $\mu$. Loosely speaking, a detailing of $\mu$ with respect to $A$ is in particular a representation of $\mu$ as a weighted sum of distributions, where the weights are given by the distribution over $A$ projected from $\xi$, and for every $a\in A$ the corresponding ``component distribution'' $\mu_a$ results from conditioning $\xi$ to the event $\{0,1\}^n\times\{a\}$.

One particular type of a detailing which we consider extensively in this work is a \emph{detailing by a set of variables}. For a set $U\subseteq\{1,\ldots,n\}$ the detailing of $\mu$ with respect to $U$ is a distribution $\mu^U$ supported over $\{0,1\}\times A$ for $A=\{0,1\}^U$, where for $a\in A$ the component $\mu_a$ is exactly the conditioning of $\mu$ to the event $\{x\in\{0,1\}^n:x_U=a\}$. Note that this detailing is actually a partition of the probability space to weighted subcubes indexed by assignments to variables in $U$. A detailing by variables can be better handled by a testing algorithm as compared to a general detailing, because a draw  $(\bx,\ba)\sim \mu^U$ can be easily obtained by taking a sample $\bx\in\{0,1\}^n$ from $\mu$ and augmenting it with $\ba=\bx_U$. In comparison, a draw from a general detailing $\xi$ is not always easily obtainable given only samples from the input distribution $\mu$. 

The notion of a detailing allows us to extract some useful \emph{statistics} of the underlying distribution. The statistic of a detailing $\xi$ with respect to $A$ is essentially the distribution $\eta$ over $A$ that is the projection of $\xi$ to the second argument, along with a multiset of \emph{type} vectors $\{t_1,\ldots, t_n\}$ over $[0,1]^A$ where $t_i(a)=\Prx_{\bx\sim\mu_a}[\bx_i=1]$. Namely, $t_i(a)$ is the probability that the variable $\bx_i$ is set to one conditioned on the event that a draw from $\xi$ lies in $\{0,1\}^n\times\{a\}$. Equivalently, the multiset of type vectors can be identified with a distribution $\Lambda$ supported on vectors in $[0,1]^{A}$. We refer to $\eta$ as the \emph{weight distribution}, and to $\Lambda$ as the \emph{type distribution} of the detailing $\xi$. Intuitively, one can consider the statistic as a sort of a ``lossy compression scheme'' for the underlying distribution. 

\subsubsection{Goodness and predictability}
One important property that a detailing can have is being \emph{$(\eps,q)$-good}. For $\eps\in(0,1)$ and $q\in \N$, we say that a single distribution $\nu$ is $(\eps,q)$-good if at least a $1-\eps$ fraction of the $q$-tuples of variables under this distribution are $\eps$-close to being independent (this is a weak form of being ``somewhat $q$-wise independent''). We say that a detailing $\xi$ of $\mu$ is $(\eps,q)$-good if at least a $1-\eps$ fraction of the component distributions $\mu_a$ (counted by their weights in the detailing) are $(\eps,q)$-good.

This property is extremely beneficial for finding a distribution that approximates the distribution of samples and queries made by a \emph{canonical algorithm} using the statistic of $\xi$. A canonical algorithm for a distribution property $\calP$ is a randomized procedure that picks a set of $s$ independent samples drawn from the distribution $\mu$, and a set of $q$ variables (that is ``coordinates'' from the space $\{0,1\}^n$, identified by their indices from $\{1,\ldots,n\}$), chosen uniformly at random, and makes its decision (either accept or reject) based on the samples restricted to the chosen variables (and possibly on its internal random coins). In particular, the acceptance probability of the tester is determined by a function $\alpha:\zo^{s \times q}\to [0,1]$. Considering canonical testers suffices for our needs, since for any index-invariant property having a testing algorithm making $s$ samples and $q$ queries, there exists a corresponding canonical tester which makes at most $sq$ queries, by \cite{CFGMS23}.

The goodness property allows us to simulate the behaviour of a canonical algorithm by using only the detailing statistic (this is essentially the ``meta-testing'' concept that was featured in~\cite{FN05,AFNS06}). In particular, to simulate the behaviour of a canonical algorithm with $s$ samples and $q$ queries we first draw $\ba_1\ldots,\ba_s$ independently from the weight distribution $\eta$ of the detailing, and then draw $q$ vectors $\bt_1,\ldots,\bt_q$ from $\Lambda$ independently. Then, we draw a random matrix $\bM\in\{0,1\}^{s\times q}$ such that for every $(i,j)\in \{1,\ldots,s\}\times\{1,\ldots,q\}$, the value $\bM_{i,j}$ is drawn independently according to the Bernoulli distribution with parameter $\bt_j(\ba_i)$. 

Essentially, for every $i\in \{1,\ldots,s\}$ the draw $\ba_i$ simulates the choice of the respective component distribution $\mu_{\ba_i}$ of $\xi$, and the values $\bt_1(\ba_i),\ldots,\bt_q(\ba_i)$  simulate the queries to the sample $i$. Note that since $\bt_j(\ba_i)$ corresponds to the $j$-th marginal of a draw from the component distribution $\mu_{\ba_i}$, and the drawn component distribution with high probability satisfies the goodness condition, the distribution of the string $\bM_{i,1},\ldots,\bM_{i,q}$ is close to the distribution of a sample drawn from $\mu$ and restricted to the set of $q$ uniformly random indices. The matrix $\bM$ is then used along with $\alpha$ to determine the acceptance probability of the algorithm.
This implies a \emph{predictability} property of the detailing: One can approximate the acceptance probability of a canonical algorithm by using only its statistic and without actually performing the queries to the input.

\subsubsection{Robustness and inheritance}
The \emph{predictability} notion is tied to a more general notion of \emph{robustness}. Roughly speaking, a detailing is robust if whenever it is extended to a more refined detailing (but still with a limited degree of refinement), the perceived probabilities of the outcomes in individual indices will not change by much. In other words, for the most parts, the values appearing in the type distribution of the extended detailing will be \emph{inherited} from their counterparts in the type distribution of the original detailing. Such an inheritance feature will be crucial for designing the distance estimation algorithm.
To facilitate such a feature, instead of tracking all the individual probabilities we go the route of Szemer\'edi's regularity lemma, and define a single ``detailing-index'' measure for a detailing, with respect to which robustness is defined. The above robustness notion is then proved to follow from satisfying the detailing-index-related measure.

However, due to the access we have to the input we are unable to find such a robust detailing. We can only consider extensions which are obtained by fixing a subset of the bits in the string. As a result, we define a weaker notion of robustness. A detailing $\xi$ is said to be \emph{weakly robust} if extending it to a more refined detailing obtained by fixing a randomly chosen and not too large subset of the bits (as opposed to arbitrary refinements in the general definition of robustness) will not change the perceived probabilities of the outcomes in individual indices by much. We later prove that being a weakly robust detailing implies the general notion of robustness, and implies the above mentioned goodness notion as well.

\subsubsection{Distance estimation algorithm}

A main engine of our work is an efficient query algorithm that provides an approximation of the statistic of a robust detailing. We set the robustness requirement to go beyond what is needed just for predictability, as the above-mentioned inheritance feature will be required for distance estimation later on. Given the detailing statistic, we then consider hypothetical distributions with hypothetical detailings featuring predictability.

For such a hypothetical detailing we first require that its statistic must predict acceptance (or at least predict a high enough probability of not rejecting) by a testing algorithm. For a statistic satisfying this requirement, we check its hypothetical distance from the distribution (which we actually measure using a distance measure between the corresponding type statistics). The smallest such hypothetical distance gives us our estimate.

On the one hand, if there exists a relatively close distribution that satisfies the property, it would imply a corresponding hypothetical detailing statistic:  We consider a predictive detailing of this target distribution, and show how to ``combine'' it with the detailing we obtained from the input to demonstrate a distance close to the hypothetical one. Here the inheritance property is crucial, since the combined detailing is related to a refinement of the detailing of the input, and we need it to maintain its type statistic.

On the other hand, if there exists a good hypothetical detailing statistic with a small hypothetical distance, it would imply the existence of a ``fake'' distribution that has a similar distance bound from the input distribution, and additionally is not rejected by the test (and hence there must be a ``real'' distribution satisfying the property that is not much further away). Building such a distribution is mostly a matter of considering a draw from a detailing, and then altering each coordinate independently with some probability that is tailored to eventually match the hypothetical acceptable statistic.

\subsection{Relation to prior works}
\label{subs:prior}

The study of {tolerant} property testing and distance approximation (also known as estimation) was initiated by Parnas, Ron and Rubinfeld \cite{PRR06}. In the setting of tolerant testing, the algorithm is required to accept with high probability if the input is close to the property, and reject with high
probability if the input is far from any other input satisfying the property. The relation between ``non-tolerant'' property testing and distance approximation depends much on whether invariance restrictions are imposed on a property. While testable graph properties (which must be invariant with respect to graph isomorphism) were proved to admit full distance estimation in \cite{FN05}, in the general string testing model (where the properties are not required to satisfy any invariance condition) the existence of non-tolerantly testable properties which are not estimable, and in fact not even tolerantly testable, was proved in \cite{FF06}. Later, it was proved in \cite{BFLR20} that there are testable properties that require a near-linear number of queries for a tolerant test.

The results here exhibit the role that invariance requirements can play in the setting of the Huge Object Model. In the most general Huge Object Model setting, features of traditional string testing factor in, and in particular one can use the reductions from the original \cite{GR22} to convert the properties from \cite{FF06} or \cite{BFLR20} to create non-tolerantly testable properties which are not tolerantly testable in the Huge Object model (as mandated by Theorem \ref{thm:main_informal}, these properties are indeed not index-invariant).

In \cite{GR22} two invariance notions were considered. These are label-invariance which discards most of the ``string-ness'' of this model and focuses on distribution properties (such as uniformity and bounded support size) that just happen to be represented by distributions on strings, and the even more restrictive notion of mapping-invariance. In \cite{CFGMS23} a study of a milder notion of invariance was initiated, namely the notion of \emph{index-invariance}. This work proves that index-invariant properties already belong to the domain where testability implies estimability, as is the case with the dense graph model.

The influence of Szemer\'{e}di's regularity-like constructions in property testing also has a long history, starting with \cite{AFKS2000} (in fact a mathematical result in this direction has already been presented in \cite{RD85}). Such constructions were usually required to go beyond the original regularity lemma, at times in terms of strength (as in \cite{AFKS2000}) and at other times in terms of the objects involved outside standard graphs, such as product posets \cite{FN07}, graphs in the sparse model \cite{BS11,BCC+19,BCC+19b,CFS21} or vertex-ordered graphs \cite{ABF17}.

To avoid ad-hoc constructions, there are several mathematical approaches to systematically strengthening the notion of regularity and defining it for other objects, and two of them gained particular prominence. The most influential one is the analytic approach, developed in \cite{LS06,BCL+06a,BCL+06b,LS07,BCL+08,LS10,BCL+12}, which has many applications to other research domains (see e.g., \cite{HKM+13,CR20,GHHS20,Sim23}). For an extensive introduction to the analytic approach see \cite{Lov12}.
This approach works particularly well in settings where the exact dependency functions do not matter. The other approach is that of robust partitions (or objects) which was first used in \cite{FN05}. Many times both approaches work, compare e.g.\ the robustness approach of \cite{ABF17} with its analytic counterpart \cite{2021Limits,GHHS20}. It is the robustness approach that we use here.

A disconcerting feature of regularity-related approaches is that they almost always provide parameter dependencies that are a tower function or even worse (e.g.\ tower of towers). The analytic approach tends to do away with explicit dependencies altogether. However, there are exceptions, and much research went into reducing the dependencies in cases that allow it, starting with \cite{FK99}. Of particular interest is the work of \cite{GKS23}, which provided a version of the estimability result from \cite{FN05} whose parameter dependency is not a tower. Also in this work the resulting dependency is not a tower, although its (constant) number of exponentiations leaves something to be desired. Considering the lower bounds in \cite{CFGMS23} related to the VC dimension, it is unlikely that this dependency can be reduced.

Another interesting insight from the comparison with the work of \cite{GKS23} is gained by noting that in our work the regularity framework in itself has a non-tower dependency, while for the dense graph setting this is not possible by the lower bound of \cite{gowers1997lower} (see also \cite{MS16}). In \cite{GKS23} they manage to move the specific estimability proof to use the weak regularity framework of \cite{FK99} instead.

\subsection{Organization of the paper}

We start with somewhat longer than usual preliminaries, Section~\ref{sec:prelim}, that provide the basic groundwork for our handling of distributions. Then, in Section~\ref{sec:detail} we develop our definitions of a detailing of a distribution, and the definition for its goodness and robustness. In Section~\ref{sec:predict} we show how from only knowing the statistic of a good detailing we can predict the behavior of a canonical test over a given input distribution, while in Section~\ref{sec:find} we show how to find a robust (and good) detailing by variables, and estimate its statistic. The final Section~\ref{sec:estimation} is where all of this comes together in our distance estimation algorithm.

\section{Preliminaries}\label{sec:prelim}
\subsection{Basic handling and manipulation of distributions}
For an integer $n$, we will denote the set $\{1, \ldots, n\}$ as $[n]$. We use boldface letters (such as $\bx$) to denote random variables.
Given two vectors ${x}$ and ${y}$ in $\{0,1\}^{n}$, we denote by $d_{H}({x}, {y})$ the normalized Hamming distance between ${x}$ and ${y}$, that is, $
d_{H}(x ,y )\eqdef \frac{1}{n}\cdot\left|\{i \in [n]\,:\, x_{i} \neq {y}_{i}\}\right|.
$

Given a discrete distribution $\mu$ over $\Omega$, for $x\in\Omega$ we use $\Prx_{\mu}[x]$ and $\mu(x)$ interchangeably (while for an event $S\subseteq\Omega$ we only use $\Pr_{\mu}[S]=\sum_{x\in S}\mu(x)$ for its probability). For $p\in [0,1]$, we let $\Ber(p)$ denote the {Bernoulli distribution} with parameter $p$.

To streamline the arguments and analysis in the following, we will define and use a somewhat uncommon notation for some very common probabilistic notions.

\begin{definition}[Restriction of a distribution]
Let $\mu$ be a distribution over $\Omega$ and $S \subseteq \Omega$ be an event such that $\Pr_\mu[S] \neq 0$. Then $\mu | ^S$ is the distribution $\mu$ restricted to $S$, which is defined over $S$ as follows: For every $x \in S$, $\Pr_{\mu|^{S}}[x]=\frac{\Pr_{\mu}[x]}{\Pr_\mu[S]}$. We also sometimes pad the conditional distribution to $\Omega$ by defining $\mu|^S(x)=0$ for $x\in\Omega\setminus S$.
\end{definition}

\begin{definition}[Projection of a distribution]
	Let $\mu$ be a distribution over $\Omega=\prod_{\ell=1}^nA_{i_\ell}$ and $i_1, \ldots, i_k$ be a set of integers from $[n]$ such that $i_1<\cdots<i_k$. The \emph{projection} $\mu|_{\{i_1,\ldots,i_k\}}$  of $\mu$ to a set of coordinates $\{i_1,\ldots,i_k\}$ is the distribution supported on $\prod_{\ell\in[k]}A_{i_\ell}$ obtained by first drawing $\bx\sim\mu $ and returning $\bx|_{\{i_1,\ldots, i_k\}}$.  Moreover, for an event $S\subseteq \Omega$, we denote by $\mu|^S_{\{{{i_1}}  ,\ldots ,{{i_k}}\}}$
	the conditioning of $\mu$ over $S$ followed by the projection on $\{i_1,\ldots,i_k\}$.
\end{definition}

\begin{definition}[Conditioning shorthand]\label{def:shorthand}
    When we have a distribution $\mu$ on a set $\Omega$ which is a product, we use a shorthand notation when we condition it by the coordinate values. For example, for $\Omega=A\times B$ and $a\in A$, we use $\mu|^{1:a}$ to denote the conditioning $\mu|^{\{a\}\times B}$, and similarly for $A'\subseteq A$ use $\mu|^{1:A'}$ to denote $\mu|^{A'\times B}$. We extend this notation to multiple coordinates as well. For example, for $\Omega=A\times B\times C$ and $a\in A$ and $c\in C$, we use $\mu|^{1,3:(a,c)}$ or $\mu|^{1:a,3:c}$ to denote $\mu|^{\{a\}\times B\times \{c\}}$.
\end{definition}

For two distributions that ``agree on a common variable'' we define a way to unify them into a single distribution, named after the database operation resembling it.

\begin{definition}[Join of distributions]\label{def:join}
    Given two distributions $\mu$ over $A\times B$ and $\nu$ over $B\times C$ that satisfy $\mu|_2=\nu|_1$, we define their \emph{join} $\mu\bowtie\nu$, as the following distribution over $A\times B\times C$:
    $$(\mu\bowtie\nu)(a,b,c)=\mu(a,b)\cdot\nu|_2^{1:b}(c)=\mu|_1^{2:b}(a)\cdot\nu(b,c).$$
    Whenever $\mu|_2(b)=\nu|_1(b)=0$ we (as expected) define $(\mu\bowtie\nu)(a,b,c)=0$. In cases where it is not clear which coordinate we unify (for example
when $A=B$) we state it explicitly. 
\end{definition}

We also define a way to ``adjust'' a distribution over a product set to have a desired projections to one of its coordinates.

\begin{definition}[Adjustment]\label{def:adjust}
	Given a distribution $\eta$ over $A\times B$ and a distribution $\nu$ over $A$, such that $\nu(a)=0$ whenever $\eta|_1(a)=0$, the \emph{adjustment of $\eta$ to $\nu$}, denoted $\nu\rhd\eta$, is defined as the following distribution over $A\times B$:
	$$(\nu\rhd\eta)(a,b)=\nu(a)\cdot\eta|^{1:a}_2(b).$$
	Whenever $\eta|_1(a)=0$ (and then also $\nu(a)=0$) we as expected define $(\nu\rhd\eta)(a,b)=0$.
\end{definition}

\subsection{Distribution distances}

The following is the most basic distance notion between distributions.

\begin{definition}[{Variation distance}]
	Let $\mu$ and $\tau$ be two probability distributions over a discrete set $\Omega$. 
	The \emph{variation distance} between $\mu$ and $\tau$ is defined as:
	$$\dtv(\mu,\tau)\eqdef \frac{1}{2} \sum_{w\in \Omega}|\mu(w)-\tau(w)|.$$
\end{definition}

Central to our work will be the notion of Earth Mover distance, with some variants. Before defining it, let us consider the notion of transfer distribution.

\begin{definition}[Transfer distribution]
	Let $\mu$ and $\tau$ be two distributions defined over two sets $A$ and $B$, respectively. A distribution $T$ over $A \times B$ is said to be a \emph{transfer distribution} between $\mu$ and $\tau$ if for every $a \in A$, $\Prx_{(\bx,\by) \sim T}[\bx = a]= \mu(a)$, and for every $b \in B$, $\Pr_{(\bx,\by) \sim T}[\by = b]= \tau(b)$ (in other words, $T|_1=\mu$ and $T|_2=\tau$). The set of all transfer distributions between $\mu$ and $\tau$ is denoted by $\calT(\mu,\tau)$.
\end{definition}

\begin{definition}[EMD with respect to a distance function]
	Let $\mu$ and $\tau$ be two distributions defined over a set $\Omega$ and $d_{\Omega}$ be a distance function defined over {$\Omega$}. Then the \emph{Earth Mover distance} (EMD) between $\mu$ and $\tau$ with respect to $d_{\Omega}$ is defined as follows:
	$$\dem(\mu,\tau) \eqdef \inf_{T \in \calT(\mu,\tau)} \Ex_{(\bx,\by) \sim T} \left[d_{\Omega}(\bx,\by)\right]$$
	where $\calT(\mu, \tau)$ denotes the set of all possible transfer distributions between $\mu$ and $\tau$.
	
	Unless stated otherwise, whenever $\Omega=\zo^n$, we use $\dem$ to refer to the EMD over the normalized Hamming distance.
\end{definition}

We will not really have to take care of limits in our arguments due to the following observation.

\begin{observation}\label{obs:achievedist}
If $\mu$ and $\tau$ are distributions over a finite set, then $\mathcal T$ is compact and in particular there exists a transfer distribution $T\in\mathcal T$ achieving the respective EMD distance.
\end{observation}

The following observation, that variation distance is essentially also a special case of the Earth Mover distance, is well-known.
\begin{observation}\label{obs:TV_is_EMD}
	The Earth Mover distance over the Kronecker distance function is identical to the total variation distance. That is, \[\dtv(\mu,\tau)\eqdef\inf_{T \in \calT(\mu,\tau)}\Ex_{(\bx,\by) \sim T}\left[\indi_{\{\bx\neq\by\}}\right],\]
	where $\indi_{\{\bx\neq\by\}}$ is the indicator function for the event that $\bx\neq\by$.
\end{observation}

Keeping with tradition, in the variation distance setting we will also refer to a transfer distribution as a \emph{coupling}. The following lemma (which will be useful to us in the sequel) is a good example of an application of the above observation.

\begin{lemma}\label{lem:adjust-dist}
	If $\nu$ is a distribution over $A$ and $\eta$ is a distribution over $A\times B$, then the adjustment of $\eta$ to $\nu$ (see Definition~\ref{def:adjust}) satisfies $\dtv(\eta,\nu\rhd\eta)\leq \dtv(\eta|_1,\nu)$.
\end{lemma}

\begin{proof}
	Assume that $T$ is an optimal coupling between $\eta|_1$ and $\nu$. That is, $\Ex_{(\ba,\ba') \sim T}\left[\indi_{\{\ba\neq\ba'\}}\right]=\dtv(\eta,\nu)$. We use it to define a (not necessarily optimal) coupling $T'$ between $\eta$ and $\nu\rhd\eta$ for which $\dtv(\eta,\nu\rhd\eta)\leq\Ex_{((\ba,\bb),(\ba',\bb')) \sim T'}\left[\indi_{\{(\ba,\bb)\neq(\ba',\bb')\}}\right]\leq\Ex_{(\ba,\ba') \sim T}\left[\indi_{\{\ba\neq\ba'\}}\right]$, which completes the proof.
	
	A draw $((\ba,\bb),(\ba',\bb')) \sim T'$ is taken as follows: We first draw $(\ba,\ba')\sim T$. If $\ba=\ba'$, then we let $\bb=\bb'$ be the result of a single draw from $\eta|^{1:\ba}_2$. If $\ba\neq\ba'$, then we let $\bb$ be the result of a draw from $\eta|^{1:\ba}_2$, and let $\bb'$ be the result of an independent draw from $\eta|^{1:\ba'}_2$.
	
	It is not hard to see that $T'$ is indeed a transfer distribution between between $\eta$ and $\nu\rhd\eta$. For example, $T'|_1(a,b)=T|_1(a)\cdot\eta|^{1:a}_2(b)=\eta|_1(a)\cdot\eta|^{1:a}_2(b)=\eta(a,b)$, and the proof for $T'|_2=\nu\rhd\eta$ is similar. It remains to bound the probability for the event $(\ba,\bb)\neq(\ba',\bb')$, and this follows from noting that the definition of $T'$ explicitly states that $\ba=\ba'$ implies $\bb=\bb'$.
\end{proof}

In the sequel, we will also use the Earth Mover distance over the weighted $\ell_1$ norm.

\begin{definition}[EMD with respect to weighted $\ell_1$-distance]\label{definition:weightedl1dist}
	Let $\eta$ be a distribution over $A$.  Also, let $x$ and $y$ be two vectors in $[0,1]^A$. The $\eta$-weighted $\ell_1$-distance between $x$ and $y$ is defined as 
	$$d_{\ell_1}^\eta(x,y)\eqdef\Ex_{\ba\sim \eta}[|{x_{\ba}-y_{\ba}}|].$$
	Considering two distributions $\Lambda$ and $\Upsilon$ defined over $[0,1]^{A}$, the EMD between them with respect to the $\eta$-weighted $\ell_1$-distance is denoted by $\dem^{\eta}(\Lambda,\Upsilon)$.
\end{definition}

In our setting, the distributions $\Lambda$ and $\Upsilon$ will always be finitely-supported (even that  $\Omega$ itself is infinite), and in particular, Observation~\ref{obs:achievedist} will still apply.

\subsection{The testing model}

The Huge Object model uses the following oracle access to the unknown input distribution.
\begin{definition}[Huge Object oracle] 
In the Huge Object model, the algorithm can access the input distribution $\mu$ in the following manner: At every stage, the algorithm may ask for a new sample $\bx\sim \mu$, independently of all previous samples, or it may ask to query a coordinate $j\in\{1,\ldots,n\}$ of a previously obtained sample $\bx'$. When such a query is made, the output of the oracle is $\bx'_j\in\{0,1\}$. 
\end{definition}

    \begin{definition}[Distribution Property]
        A \emph{distribution property} $\mathcal{P}$ is a sequence $\mathcal{P}_1,\mathcal{P}_2,\ldots$ such that for every $n \ge 1$, $\mathcal{P}_n$ is a compact subset of the set of all distributions over $\{0,1\}^n$.
    \end{definition}

    \begin{definition}[Distance of a distribution from a property]
        Let $\mathcal{P} = (\mathcal{P}_1,\mathcal{P}_2,\ldots)$ be a property and $\mu$ be a distribution over $\{0,1\}^n$ for some $n$. The \emph{distance of $\mu$ from $\mathcal{P}$} is defined as $\dem(\mu,\mathcal{P}) = \min_{\tau\in \mathcal{P}_n}\{\dem(\mu,\tau)\}$.
    \end{definition}

Testing and tolerant testing are defined as follows.

\begin{definition}[$\eps$-tester] Fix $0<\eps\le 1$ and let $\calP$ be a property of distributions supported on $\zo^n$. An \emph{$\eps$-tester} for $\calP$ is a randomized procedure that has Huge Object oracle access to an input distribution $\mu$ and satisfies the following with probability at least $2/3$:
\begin{enumerate}
    \item (Completeness) If $\mu\in\calP$, then the algorithm outputs \textbf{Accept}.
    \item (Soundness) If $\dem(\mu,\calP)>\eps$, then the algorithm outputs \textbf{Reject}.
\end{enumerate}
\end{definition}

\begin{definition}[$(\eps_1,\eps_2)$-tolerant tester] Fix $0\le\eps_1<\eps_2\le 1$ and let $\calP$ be a property of distributions supported on $\zo^n$. An \emph{$(\eps_1,\eps_2)$-tolerant tester} for $\calP$ is a randomized procedure that has Huge Object oracle access to an input distribution $\mu$ and satisfies the following with probability at least $2/3$:
\begin{enumerate}
    \item (Completeness) If $\dem(\mu,\calP)\le \eps_1$, then the algorithm outputs \textbf{Accept}.
    \item (Soundness) If $\dem(\mu,\calP)>\eps_2$, then the algorithm outputs \textbf{Reject}.
\end{enumerate}
Note that an (non-tolerant) $\eps$-tester for a property $\calP$ is the same as a ``$(0,\eps)$-tolerant'' tester for $\calP$.
\end{definition}

In this work we will be interested in the following class of properties.

\begin{definition}[Index-invariant property] Let $\mu$ be a distribution over $\zo^n$. For any permutation $\pi:[n]\to[n]$, let $\mu_\pi$ be the distribution such that $\mu(x_1,\ldots,x_n)=\mu_\pi(x_{\pi(1)},\ldots,x_{\pi(n)})$ for every $x\in\zo^n$. A distribution property $\calP$ is called \emph{index-invariant} if $\{\mu_\pi:\mu\in\mathcal \calP\}=\calP$ for every permutation $\pi$.
\end{definition}

\subsection{Canonical testers}

The following is the definition of the information obtained by a specific sampling and querying pattern from an input distribution $\mu$.

\begin{definition}[Canonical distribution] Fix $s,q\in\N$ and let $\mu$ be a distribution over $\zo^n$. The \emph{$(s,q)$-canonical distribution} for $\mu$ is a distribution $\calD^{s,q}_{\mathsf{test}}$ over $\zo^{s\times q}$ obtained by the following process. Draw $s$ independent samples $(\bx^1,\ldots,\bx^s)$ from $\mu$, then pick a uniformly random $q$-tuple $(\bj_1,\ldots,\bj_q)\in[n]^q$, and finally return the following matrix:
	\[ \bM=\begin{bmatrix}
		\bx^1_{\bj_1}&\cdots &\bx ^{1}_{\bj_q}\\
		\vdots& &\vdots \\
		\bx^s_{\bj_1}&\cdots& \bx ^{s}_{\bj_q}
	\end{bmatrix} .\]
When $s$ and $q$ are clear from the context, we omit the superscript and use $\calD_{\mathsf{test}}$ to denote the above.
\end{definition} 

Canonical testers are defined as those testers that sample and query according to the above pattern.

\begin{definition}[$(s,q)$-canonical tester]\label{def:cantest}
Fix $\epsilon\in(0,1)$. An $(s,q)$-canonical tester with proximity parameter $\eps$ for some distribution property $\calP$, is a randomized procedure that acts by obtaining a matrix $\bM$ from the $(s,q)$-canonical distribution $\calD_{\mathsf{test}}$ and then accepting or rejecting based on $\bM$ and possibly some internal coin tosses.  	We let $\alpha_\eps:\zo^{s\times q}\to [0,1]$ so that $\alpha_\eps(\bM)$ denotes the probability that the $(s,q)$-canonical tester (with proximity parameter $\epsilon$) accepts $\bM$. The acceptance probability of the tester is denoted by $\textsf{acc}_\epsilon(\mu)\eqdef \Ex_{\bM\sim\calD_{\mathsf{test}}}[\alpha_\eps(\bM)]$.
	
In particular, an $(s,q)$-canonical $\eps$-test must satisfy the following. If $\mu\in\calP$ then the tester accepts with probability at least $2/3$ (i.e., $\textsf{acc}_\epsilon(\mu)\ge 2/3$) and if $d(\mu,\calP)>\eps$, then the tester rejects with probability at least $2/3$ (i.e., $\textsf{acc}_\epsilon(\mu)< 1/3$). The total number of queries of such a test is $sq$.
\end{definition}

The following lemma states that for index-invariant properties, general (even adaptive) testers can be converted to canonical testers at a cost that is at most quadratic in their number of queries.
\begin{lemma}[Theorem 1.7 in~\cite{CFGMS23}] Fix $\epsilon\in(0,1)$. Let $\calP$ be an index-invariant property such that there exists a (possibly adaptive) $\epsilon$-test for $\mathcal P$ with $s$ samples and $q$ queries. Then there exists an $(s,q)$-canonical $\eps$-test for $\calP$ performing at most $sq\le q^2$ queries.
\end{lemma}

\subsection{Quantized distributions and distance bounds}

We now define the notion of quantized distributions, which will feature heavily in our estimation algorithm. Quantizing will allow us to apply ``finite'' algorithms over objects which are defined over infinite sets.

\begin{definition}[$\rho$-quantized distribution] For $r\in \N$ and $\rho=1/r\in(0,1)$, we say that a distribution $\mu$ over $\Omega$ is \emph{$\rho$-quantized} if for any $w\in \Omega$ it holds that $\mu(w)$ is an integer multiple of $\rho$. That is, the measure takes the form $\mu:\Omega\to \{0,\rho,2\rho,\ldots,1\}$.
\end{definition}
The following lemma is folklore. For completeness we provide it with a proof sketch.

\begin{lemma}\label{lem:quatization}
Any distribution $\mu$ over a finite set $\Omega$ can be transformed to a $\rho$-quantized distribution $\mu'$ so that $\mu'(x)\in\{\rho\lceil\mu(x)/\rho\rceil,\rho\lfloor\mu(x)/\rho\rfloor\}$ for every $x\in\Omega$. In particular $\mu'(x)=\mu(x)$ whenever $\mu(x)\in\zo$, $|\mu(x)-\mu'(x)|<\rho$ for all $x\in\Omega$, and $\dtv(\mu,\mu')<\frac{\rho}{2}\cdot |\Omega|$.
\end{lemma}

\begin{proof}
We construct $\mu'$ by finding for every $x\in\Omega$ a choice of $\mu'(x)\in\{\rho\lceil\mu(x)/\rho\rceil,\rho\lfloor\mu(x)/\rho\rfloor\}$ so that $\mu'$ will be a distribution, i.e.\ $\sum_{x\in\Omega}\mu'(x)=1$. To produce this choice, we set an order $\Omega=\{x_1,\ldots,x_r\}$ where $r=|\Omega|$, and then inductively choose $\mu'(x_i)=\rho\lceil\mu(x_i)/\rho\rceil$ if $\sum_{j\in [i-1]}\mu'(x_j)\leq\sum_{j\in [i-1]}\mu(x_j)$ and $\mu'(x_i)=\rho\lfloor\mu(x_i)/\rho\rfloor$ otherwise. It is not hard to see that this way we maintain the invariant $\sum_{j\in [i]}\mu(x_j)-\rho<\sum_{j\in [i]}\mu'(x_j)<\sum_{j\in [i]}\mu(x_j)+\rho$, and in particular $\sum_{j\in [r]}\mu'(x)=1$ since it is an integer multiple of $\rho$ that lies strictly between $1-\rho$ and $1+\rho$.
\end{proof}

We continue with more definitions and results regarding approximate distributions, mainly centered around the quantization notion.

\begin{definition}[$\rho$-legitimate approximation] \label{def:rho-legit}Fix $t\in[0,1]^A$ and $\rho\in(0,1)$. A vector $\tilde{t}\in[0,1]^A$ is a \emph{$\rho$-legitimate approximation} of $t$ if $\|\tilde{t}-t\|_\infty\le \rho$.
\end{definition}
\begin{observation} \label{obs:small_l1}Let $\eta$ be a distribution over a set $A$. For any vector $t\in[0,1]^A$ we have a $\rho$-legitimate approximation $\tilde t$ satisfies $d^\eta_{\ell_1}(\tilde{t},t)\le\rho $.
\end{observation}
\begin{proof}By definition of the $\eta$-weighted $\ell_1$ distance
	$d^{\eta}_{\ell_1}(\tilde{t},t)=\sum_{a\in A}\eta(a)\cdot|\tilde{t}(a)-t(a)|\le \rho$.
\end{proof}

We will need to quantize not only the probabilities, but also the probability space itself (if it were not finite to begin with).

\begin{definition}For $\rho\in (0,1)$ and a distribution $\Lambda$ over $[0,1]^A$, we define $\Lambda_\rho$ to be a distribution over $\{0,\rho,2\rho,\ldots,1\}^A$ obtained in the following way. Sample $t$ according to $\Lambda$ and round each entry in $t$ to the nearest multiple of $\rho$ to obtain a $\rho$-legitimate approximation $\tilde{t}$ of $t$.
\end{definition}

Now we have the following lemma which connects these two definitions.

\begin{lemma} \label{lem:domain_quat}Let $\rho\in(0,1)$ and fix a distribution $\Lambda$ over $[0,1]^A$.  Then for any distribution $\eta$ over $A$, $\dem^{\eta}(\Lambda,\Lambda_\rho)\le \rho/2$.
\end{lemma}
\begin{proof} We consider the following transfer function $T:[0,1]^A\to[0,1]^A$ from $\Lambda$ to $\Lambda_\rho$:
	\[T(x,y)\eqdef\begin{cases}
		\Lambda(x)& \text{if } y=\tilde{x}\\
		0 &\text{otherwise}, 
	\end{cases}\]
	where $\tilde{x}$ denotes the vector obtained from $x$ by rounding the entries to the nearest multiple of $\rho$, which is at most $\rho/2$ away from the unrounded value. From Observation~\ref{obs:small_l1} we have that whenever $T(x,y)\neq0$ we have $d^{\eta}_{\ell_1}(x,y)\le \rho/2$. Therefore,
	\begin{align*}
		\dem^{\eta}(\Lambda,\Lambda_\rho)\le \sum_{x,y\in[0,1]^A}T(x,y)\cdot d^{\eta}_{\ell_1}(x,y)\le \frac{\rho}{2}\sum_{x,y\in[0,1]^A}T(x,y)\le \rho/2.
\end{align*}
\end{proof}

As a sort of a summary of the above, we now formulate a ``quantize everything'' lemma.

\begin{lemma} \label{lem:Type_Quants}
Fix $\rho=1/r$ for some $r\in\mathbb N$, and let $\Lambda$ be a distribution over $[0,1]^A$ for some finite set $A$ of size $\ell$. There exists a $\rho/(r+1)^\ell$-quantized distribution $\Lambda'$ over $\{0,\rho,2\rho,\ldots,1\}^A$ such that $\dem^\eta(\Lambda',\Lambda)<\rho$.
\end{lemma}
\begin{proof}
We use Lemma~\ref{lem:domain_quat} to obtain a distribution $\Lambda_\rho$ over $\Omega=\{0,\rho,2\rho,\ldots,1\}^A$ such that $\dem^{\eta}(\Lambda,\Lambda_\rho)\le \rho/2$, and apply Lemma~\ref{lem:quatization} using $\rho'=\rho/(r+1)^\ell=\rho/|\Omega|$ to obtain a $\rho'$-quantized distribution $\Lambda'$ over $\{0,\rho,2\rho,\ldots,1\}^A$,  such that $\dtv(\Lambda_\rho,\Lambda')\le \rho/2$. Overall, using the triangle inequality, we have:
	\[\dem^{\eta}(\Lambda',\Lambda)\le  \dem^{\eta}(\Lambda',\Lambda_\rho)+ \dem^{\eta}(\Lambda_\rho,\Lambda)\le \rho. \]
\end{proof}

We conclude with a lemma that allows us to move to guarantee a small change in the weighted Earth Mover distance when we need to use an \emph{approximation} of the weight distribution $\eta$.

\begin{lemma}\label{lem:small_dtv} Let $\eta$ and $\eta'$ be two distributions over $A$ satisfying $\dtv(\eta,\eta')\le \rho$, and let $\Lambda$ and $\Upsilon$ be two discrete distributions over $[0,1]^A$. Then, $\dem^{\eta'}(\Lambda,\Upsilon)\le \dem^{\eta}(\Lambda,\Upsilon)+2\rho$.
\end{lemma}
\begin{proof} Let $T$ denote an optimal transfer function between $\Lambda$ and $\Upsilon$ with respect to $\dem^{\eta}(\Lambda,\Upsilon)$ (which might be sub-optimal with respect for $\dem^{\eta'}(\Lambda,\Upsilon)$). Using also the fact that $|t(a)-t'(a)|\le 1$ for any $t,t'\in[0,1]^A$ and $a\in A$ we obtain:
	\begin{align*}
		\dem^{\eta'}(\Lambda,\Upsilon)-\dem^{\eta}(\Lambda,\Upsilon)&\leq\Ex_{(\bt,\bt')\sim T}\left[\Ex_{\ba\sim\eta'}|\bt(\ba)-\bt'(\ba)|\right]-\Ex_{(\bt,\bt')\sim T}\left[\Ex_{\ba\sim\eta}|\bt(\ba)-\bt'(\ba)|\right]\\
		&=\Ex_{(\bt,\bt')\sim T}\left[\sum_{a\in A}(\eta'(a)-\eta(a))\cdot |\bt(a)-\bt'(a)|\right]\\
		&\le\Ex_{(\bt,\bt ')\sim T}\left[\sum_{a\in A}|\eta'(a)-\eta(a)|\cdot |\bt(a)-\bt'(a)|\right]\\
		&\le \sum_{a\in A}|\eta'(a)-\eta(a)|\le 2\dtv(\eta,\eta')\le 2\rho.
\end{align*}\end{proof}

\subsection{Some useful probabilistic inequalities} %

We will use some well-known large deviation inequalities.

\begin{lemma}[Multiplicative Chernoff bound, see~\cite{DubhashiP09}]
\label{lem:cher_bound1}
Let $\bX_1, \ldots, \bX_n$ be independent random variables such that $\bX_i \in [0,1]$  for every $i\in[n]$. For $\bX=\sum\limits_{i=1}^n \bX_i$, the following holds for any $0\leq \delta \leq 1$.
$$ \Prx[\abs{\bX-\Ex[\bX]} \geq \delta\Ex[\bX]] \leq 2\exp{\left(-{\Ex[\bX] \delta^2}/{3}\right)}.$$
\end{lemma}

\begin{lemma}[Additive Chernoff bound, see~\cite{DubhashiP09}]
\label{lem:cher_bound2}
Let $\bX_1, \ldots, \bX_n$ be independent random variables such that $\bX_i \in [0,1]$  for every $i\in[n]$. For $\bX=\sum\limits_{i=1}^n \bX_i$ and $a \leq \Ex[\bX] \leq b$, the following hold for any $\delta >0$.
\begin{itemize}
\item[(i)] $\Prx \left[ \bX \geq b + \delta \right] \leq \exp{\left(-{2\delta^2}/{n}\right)}$.
\item[(ii)] $\Prx \left[ \bX \leq a - \delta \right] \leq \exp{\left(-{2\delta^2} /{ n}\right)}$.
\end{itemize}
\end{lemma}

\begin{lemma}[Hoeffding's Inequality, see~\cite{DubhashiP09}] \label{lem:hoeffdingineq}

Let $\bX_1,\ldots,\bX_n$ be independent random variables such that $a_i \leq \bX_i \leq b_i$ and $\bX=\sum\limits_{i=1}^n \bX_i$. Then, for all $\delta >0$, 
$$ \Prx\left[\abs{\bX-\Ex[\bX]} \geq \delta\right] \leq  2\exp\left({-2\delta^2}/ {\sum\limits_{i=1}^{n}(b_i-a_i)^2}\right).$$

\end{lemma}

\begin{lemma}[Hoeffding's Inequality for sampling without replacement~\cite{hoeffding1994probability}] \label{lem:hoeffdingineq_without_replacement}

Let $n$ and $m$ be two integers such that $1 \leq n \leq m$, and  $x_1, \ldots, x_m$ be real numbers,  with $a \leq x_i \leq b$ for every $i \in [m]$. Suppose that $\bI$ is a set that is drawn uniformly from all subsets of $[m]$ of size $n$, and let $\bX = \sum\limits_{i \in \bI} x_i$.
Then, for all $\delta >0$, 
$$ \Prx\left[\abs{\bX-\Ex[\bX]} \geq \delta\right] \leq  2\exp\left({-2\delta^2}/ {n \cdot (b-a)^2}\right).$$
\end{lemma}

We will also use some second moment lower bounds.

\begin{lemma}[Cauchy-Schwartz inequality]\label{lem:CS}
    Given any distribution with two random variables $\bX$ and $\bZ$,
    $$\Ex_{\bZ}\left[\Ex[\bX^2\mid \bZ]\right]\ge \Ex_{\bZ}\left[\Ex[\bX\mid\bZ]^2\right].$$ This is a direct implication of $\Ex[\bX^2]-(\Ex[\bX])^2=\Ex[(\bX-\Ex[\bX])^2]\geq 0$ when considered over any possible conditioning of the probability space to a value of $\bZ$.
\end{lemma}

\begin{observation}[Defect form Cauchy-Schwartz inequality]\label{obs:square_bump}
	If $\bX$ is a random variable satisfying $\Prx[|\bX-\Ex[\bX]|\geq\alpha]\geq\beta$, then $$\Ex[\bX^2]-(\Ex[\bX])^2=\Ex[(\bX-\Ex[\bX])^2]\geq\alpha^2\beta.$$
\end{observation}

In addition, we will use the following well-known bound.
\begin{lemma}[Reverse Markov's inequality] \label{lem:reverseMarkov} For any $\gamma\in(0,1)$, let $\bX$ be a random variable in $[-1,1]$ with $\Ex[\bX]\ge \gamma$. Then $\Prx[\bX\ge \gamma/2]\ge \gamma/2$.  
\end{lemma}

\section{Detailing, variable types and robustness}\label{sec:detail}

In this section we define the main objects that we use throughout this work.

\subsection{Detailing and refinements}

Our most basic definition is that of a detailing of a distribution $\mu$. Essentially it is a possible ``deconstruction'' providing additional information to aid with the analysis of $\mu$, in analogy to the notion of a vertex partition of a graph.

\begin{definition}[Detailing] \label{definition:detailing}
	Let $\mu$ be a distribution over $\Omega$ and $A$ be a set. A \emph{detailing} $\xi$ of $\mu$ with respect to $A$ is a distribution over $\Omega\times A$ that satisfies $\xi|_1=\mu$. 
 $|A|$ is referred to as the \emph{length} of the detailing. When $|A|=1$, we call $\xi$ the \emph{trivial detailing} and identify it with $\mu$ itself.
\end{definition}

The following immediate observation is in fact the motivation for Definition~\ref{def:adjust}.

\begin{observation}
	If $\nu$ is a distribution over $A$ and $\eta$ is a distribution over $A\times B$, then the adjustment $\nu\rhd\eta$ of $\eta$ to $\nu$ is in particular a detailing of $\nu$ with respect to $B$.
\end{observation}

We will also define $\rho$-quantized detailings, even when the original distribution is not quantized. Such detailings are not quantized by themselves, only their ``detailing portion'' is quantized.

\begin{definition}[Quantized detailing]\label{def:quantized-detailing}
	Let $\mu$ be a distribution over $\Omega$, $A$ be a set and $\xi$ be a detailing of $\mu$ with respect to $A$. For $\rho\in(0,1)$ for which $1/\rho\in\mathbb{N}$, we say that $\xi$ is \emph{$\rho$-quantized} if for every $x\in\Omega$ for which $\xi|_1(x)>0$, the conditional distribution $\xi|^{1:x}$ is $\rho$-quantized.
\end{definition}

The following is an immediate consequence of Lemma~\ref{lem:quatization}.

\begin{observation}\label{obs:detailing-quatization}
	For a distribution $\mu$ on a finite set $\Omega$, any detailing $\xi$ of $\mu$ with respect to a finite $A$ can be transformed to a $\rho$-quantized detailing $\xi'$, so that $\xi'|^{1:x}(a)\in\{\rho\lceil\xi|^{1:x}(a)/\rho\rceil,\rho\lfloor\xi|^{1:x}(a)/\rho\rfloor\}$ for every $x\in\Omega$ and $a\in A$. In particular $\xi'(x,a)=\xi(x,a)$ whenever $\xi'|^{1:x}(a)\in\zo$ and $\dtv(\xi,\xi')<\frac{\rho}{2}\cdot |A|$.
\end{observation}

\begin{proof}
We will use Lemma~\ref{lem:quatization} over $\xi|^{1:x}$ for every $x\in\Omega$ with $\xi|_1(x)>0$.
\end{proof}

Any given detailing of a distribution can be ``detailed further'', resulting in a refinement.

\begin{definition}[Refinement of a detailing]\label{definition:refinementdetailing}
	Let $\mu$ be a distribution defined over $\Omega$ and $\xi$ be a detailing of $\mu$ with respect to $A$. A \emph{refinement} of $\xi$ by a set $B$ is a detailing $\xi'$ of $\xi$ with respect to $B$. As there is a natural bijection between $\Omega \times (A \times B)$ and $(\Omega \times A) \times B$, $\xi'$ can and will be considered also as a detailing of $\mu$ with respect to $A \times B$.
\end{definition}

Specific to our investigation, we also need a definition for a ``do-nothing'' detailing (or refinement), to which we refer as flat.

\begin{definition}[Flat detailing] 
	Let $\mu$ be a distribution over $\Omega$ and $A$ be a set. A {detailing} $\xi$ of $\mu$ with respect to $A$ is said to be a \emph{flat detailing} if it is a product distribution over $\Omega \times A$, i.e. $\xi(x,a)=\mu(x)\cdot\xi|_2(a)$ for every $x \in \Omega, a \in A$.
\end{definition}

\begin{definition}[Flat refinement of a detailing]
	Let $\mu$ be a distribution over $\Omega=\{0,1\}^n$ and $\xi$ be a detailing of $\mu$ with respect to $A$. Consider a refinement $\xi'$ of $\xi$ with respect to  $B$, i.e.,  a detailing of $\mu$ with respect to $A \times B$ for which $\xi'|_{1,2}=\xi$. Such  a
	$\xi'$ is said to be a \emph{flat refinement} if it is a flat detailing of $\xi$.
\end{definition}

A detailing will usually have a continuum of possible refinements, even with respect to a set of size $2$. Moreover, there is no good way to sample from a detailing or a refinement in its general form when we only have access to samples from $\mu$, and even then we only have query access to the samples from $\mu$. A specific way to define a detailing which is tangible for a testing algorithm is to ``partition'' $\mu$ by the restriction of the samples to a specific set of variables.

\begin{definition}[A refinement by a set of variables $U$]\label{definition:part-sub}
	Let $\mu$ be a distribution over $\Omega=\{0,1\}^n$ and $\xi$ be a detailing of $\mu$ with respect to a set $A$.  For a subset $U \subseteq [n]$, the \emph{refinement $\xi^{U}$ of $\xi$ by the set $U$} is the distribution over $\Omega \times (A\times \{0,1\}^{|U|})$ that satisfies the following:
	\begin{enumerate}
		\item $\xi^U(x,a,v)=0$ if $x_U \neq v$, and
		
		\item $\xi^U(x,a,v)=\xi(x,a)$ if $x _U = v$. 
	\end{enumerate}

	Note that the length of the detailing $\xi^U$ is $|A|\cdot 2^{|U|}$.  When $\xi$ is the trivial detailing (i.e. $\xi=\mu$) we denote the refinement by $U$ as $\mu^U$, and also call it the \emph{detailing of $\mu$ defined by the variable set $U$}.
\end{definition}

\subsection{Type distributions}

When trying to analyze a distribution by its detailing with respect to a set $A$, a major role is given to its ``weights'' as given by its projection over $A$.

\begin{definition}[Weight distribution of a detailing]
	Let $\mu$ be a distribution defined over $\Omega$, and let $\xi$ be a detailing of $\mu$ with respect to a set $A$. The \emph{weight distribution of $\xi$} is defined as $\xi|_2$.
\end{definition}

We sometimes go the other way and define a flat refinement of a detailing $\xi$ by its weight distribution (which should ``conform'' to the weight distribution of $\xi$).

\begin{definition}[Flat refinement by weights]
    Given a detailing $\xi$ of $\mu$ over $A$, and a detailing $\eta$ of $\xi|_2$ over $B$ (that is, a distribution over $A\times B$ for which $\eta|_1=\xi|_2$, thus ``extending'' the weight distribution of $\xi$), we denote by $\xi_{\langle\eta\rangle}$ the specific flat refinement of $\xi$ defined by $\xi_{\langle\eta\rangle}(x,a,b)=\xi(x,a)\cdot\eta|^{1:a}_2(b)=\xi|^{2:a}_1(x)\cdot\eta(a,b)$. In other words, $\xi_{\langle\eta\rangle}=\xi\bowtie\eta$.
\end{definition}

We now define the other important aspect of a detailing of a distribution over $\zo^n$, the ``type'' of a variable $i\in [n]$ as defined by its distribution in the component $\xi|^{2:a}_1$ for every $a\in A$.

\begin{definition}[Variable types of a detailing]
	Let $\mu$ be a distribution of $\{0,1\}^n$. Let us consider a detailing $\xi$ of $\mu$ with respect to $A$. The \emph{type}  $t_i$ of a variable $i \in [n]$ with respect to $\xi$, is a vector in $[0,1]^{A}$ defined by
	$$t_i \eqdef \left\langle\Prx_{(\bx,a) \sim \xi | ^{\mu \times \{a\} }} [\bx_i=1]\right\rangle_{a\in A}.$$
    For any $a\in A$ for which $\xi|_2(a)=0$, the value $t_i(a)$ is set arbitrarily (and will not affect any computation in the following).
\end{definition}

\begin{definition}[Type distribution]\label{definition:typedist} Let $\xi$ be a detailing of $\mu$ with respect to $A$.
    The \emph{type distribution} of $\xi$ is denoted by $\Lambda$ and defined over $[0,1]^{A}$, where a sample from $\Lambda$ is obtained as follows: Choose $i \in [n]$ uniformly at random and report its type $t_i$.
\end{definition}

We also look at the types that we would expect from a flat refinement of a detailing.

\begin{definition}[Flat extension]
    Given a vector $t\in [0,1]^A$ and a set $B$, its \emph{flat extension} with respect to $B$ is the vector $t'\in [0,1]^{A\times B}$ defined by $t'(a,b)=t(a)$ for all $a\in A$ and $b\in B$. Given a finitely supported distribution $\Lambda$ over $[0,1]^A$, its \emph{flat extension} with respect to $B$, denoted by $\Lambda_{\langle B\rangle}$, is correspondingly defined as the distribution for which $\Lambda_{\langle B\rangle}(t')=\Lambda(t)$ for any $t\in [0,1]^A$ and $t'\in [0,1]^{A\times B}$ which is the flat extension of $t$.
\end{definition}

\begin{observation}
    If $\xi$ is a detailing of $A$, and $\xi'$ is any flat refinement of $\xi$ with respect to $B$, then for any $i$, the type $t'_i$ of any variable $i$ with respect to $\xi'$ is exactly the flat extension of the type $t_i$ of $i$ by $\xi$ with respect to $B$ (not counting pairs $(a,b)$ for which $\xi'|_{2,3}(a,b)=0$, for which $t'_i(a,b)$ is anyway arbitrarily defined).

    Respectively, the type distribution of such an $\xi'$ is the flat extension with respect to $B$ of the type distribution of $\xi$.
\end{observation}
While usually converting the multi-set of variables' types to a distribution as per Definition~\ref{definition:typedist}, we sometimes consider the other direction of starting with a distribution and ``implementing'' it as an ordered sequence of elements. We formulate the following with respect to a general space $\Omega$, and then provide some lemmas for a corresponding general Earth Mover distance, but here we will almost always use the space $[0,1]^A$ for some $A$ and an $\eta$-weighted $\ell_1$ metric for some $\eta$, as befits distributions over variable types.

\begin{definition}[Implementation] Let $\Lambda$ be a discrete distribution over $\Omega$. An \emph{implementation} of $\Lambda$ is a function $h:[n]\to\Omega$, so that $|\{i\in[n]:h(i)=t\}|=n\cdot\Lambda(t)$ for all $t\in\Omega$.
	
Given a detailing $\xi$ of some $\mu$ over $\{0,1\}^n$ admitting type distribution $\Lambda$, the \emph{implementation of $\Lambda$ demonstrated by the detailing $\xi$} is the map $h$ defined by $h(i)=t_i$, where $t_i$ is the type of $i$ under the detailing $\xi$.
\end{definition}

Note that an implementation of $\Lambda$ is equivalent to a uniform distribution $\pi$ over $n$ tuples $(i,w)\in [n]\times\Omega$, where every $i\in [n]$ appears in exactly one of the pairs and $\Lambda(w)=\Prx_\pi [[n]\times \{w\}]=\pi|_2(w)$.

The proof of the following lemma is obtained using linear programming.

\begin{lemma}[Follows from the proof of Claim 6.5 in \cite{CFGMS23}] \label{lem:quantized-transfer}
	Fix a space $\Omega$ with a distance measure $d$.
	If two finitely supported distributions $\Lambda$ and $\Lambda'$ over $\Omega$ are $\frac{1}{n}$-quantized, then there exists a $\frac{1}{n}$-quantized transfer function $T$ from $\Lambda$ to $\Lambda'$ which is optimal with respect to $\dem(\Lambda,\Lambda')$, that is, \[\dem(\Lambda,\Lambda')=\Ex_{(\bt,\bt')\sim T}\left[d(\bt,\bt')\right]\]
\end{lemma}

The following ``satisfiable pigeon-hole principal'' is trivial.

\begin{observation}\label{obs:pigeonhole}If $L$ is a multiset of members from $X\times Y$ and $N$ is a multiset of members from $Y\times Z$, and the multisets $\langle b:(a,b)\in L\rangle$ and $\langle b:(b,c)\in N\rangle$ are identical, then there is a multiset $J$ of members from $X\times Y\times Z$ for which $\langle (a,b):(a,b,c)\in J\rangle=L$ and $\langle (b,c):(a,b,c)\in J\rangle=N$.
\end{observation}

Using Lemma~\ref{lem:quantized-transfer} and Observation~\ref{obs:pigeonhole} the following lemma is immediate, by identifying the $\frac1n$-quantized distributions with the corresponding $n$ element multisets, and using the fact that an implementation is a $\frac1n$-quantized distribution over $[n]\times\Omega$, whose restriction to the first coordinate is uniform over $[n]$.

\begin{lemma}\label{lem:1/n-quantized-optimal}
	Fix a space $\Omega$ with a distance mesure $d$. If $\Lambda$ and $\Lambda'$ are $\frac{1}{n}$-quantized, than any implementation $h$ of $\Lambda$ can be extended to an implementation $H$ of a transfer distribution $\kappa$ between $\Lambda$ and $\Lambda'$ such that \[\dem(\Lambda,\Lambda')=\Ex_{(\bt,\bt')\sim \kappa}\left[d(\bt,\bt')\right]=\Ex_{\bi\sim [n]}\left[d(((H(\bi))_1,(H(\bi))_2)\right]\]
\end{lemma}

We end by noting that the above lemma has an easy ``converse''.

\begin{observation}\label{obs:weightedEMD} Let $\Lambda$ and $\Lambda'$ be distributions over some $\Omega$ with distance measure $d$, let $h:[n]\to\Omega$ be an implementation of $\Lambda$ and $h':[n]\to\Omega$ be an implementation of $\Lambda'$. Then, 
	\[\dem(\Lambda,\Lambda')\le\Ex_{\bi\sim[n]}\left[d(h(\bi),h'(\bi))\right]. \]
\end{observation}

\begin{proof}
By the definition of $\dem(\Lambda,\Lambda')$, it is enough to construct a transfer distribution $T$ between $\Lambda$ and $\Lambda'$ for which $\Ex_{(\bt,\bt')\sim T}\left[d(\bt,\bt')\right]\leq\Ex_{\bi\sim[n]}\left[d(h(\bi),h'(\bi))\right]$. To construct such a $T$, we define a draw $(\bt,\bt')\sim T$ as the result of uniformly drawing $\bi\sim[n]$ and setting $(\bt,\bt')=(h(\bi),h'(\bi))$.
\end{proof}

\subsection{Notions of robustness}

\begin{definition}[Index of a detailing]
	Let $\xi$ be a detailing of $\mu$ with respect to $A$.
	The \emph{index} of $\xi$ is defined as 
	$$\ix(\xi) \eqdef \Ex_{\bi\sim[n]}\left[\Ex_{\ba\sim\xi|_2}\left[\Pr_{\bx\sim\xi|_1^{2:\ba}}[\bx_i=1]^2\right]\right]$$    
\end{definition}

Note that for any detailing $\xi$, $0 \leq \ix(\xi) \leq 1$.

We first define the robustness of a detailing in its general form.
\begin{definition}[Robust detailing]
	Let $\delta \in (0,1)$ and $\ell\in\mathbb{N}$ be constants. A detailing $\xi$ of $\mu$ with respect to $A$ is called \emph{$(\delta,\ell)$-robust} if there exists no refinement $\xi'$ of $\xi$ with respect to a set $B$ such that $|B|\leq \ell$ and $\ix(\xi') \geq \ix(\xi) + \delta$.
\end{definition}

The second type of robustness concerns only refinements obtained by a set of variables, and even then puts its requirement over most (rather than all) candidate variable sets.

\begin{definition}[Weakly robust detailing]\label{def:weaklyrobust}
	Let $\delta \in (0,1)$ and $k\in\mathbb{N}$ be constants. A detailing $\xi$  is said to be \emph{$(\delta,k)$-weakly robust} if for at least a $1-\delta$ fraction of the sets $U\subseteq [n]$ of size at most $k$, it holds that $\ix(\xi^U)<\ix(\xi)+\delta$.
	
\end{definition}

A useful feature of the way robustness is defined, is that the existence of a robust refinement of a given detailing follows almost immediately from the definition. We state here the variant of such a lemma that will be useful to us later on.

\begin{lemma}\label{lem:weakly_robust_refinement}
	For any $k$ and $\delta$, if $\xi$ is any detailing of $\mu$ with respect to a set $A$, then there exists a set $U$ of size at most $k/\delta$ such that $\xi^U$ is $(\delta,k)$-robust.
\end{lemma}

\begin{proof}
	We consider the following sequence of sets $U_0,U_1,\ldots,U_r$ defined by the following inductive process: We start with the basis $U_0=\emptyset$, and given $U_i$ we check whether there exists a set $V_i\subset [n]$ of size at most $k$ so that $\ix(\xi^{U_i\cup V_i})\geq\ix(\xi^{U_i})+\delta$. If there exists such a set, we set $U_{i+1}=U_i\cup V_i$ and continue this process for $U_{i+1}$. Otherwise, we set $r=i$, $U=U_i$ and terminate the process.
	
	By definition, once this process is terminated, we obtain a $(\delta,k)$-weakly robust detailing $\xi^U$ (in fact the termination condition is slightly stronger than the condition of Definition~\ref{def:weaklyrobust}). Additionally, noting that for every $i\leq r$ we have $\ix(\xi^{U_i})\geq\ix(\xi)+i\cdot\delta$ and $|U_i|\leq ik$ by induction, and that $\ix(\xi^W)\in [0,1]$ by definition, for any set $W\subseteq [n]$, we conclude that $r\leq 1/\delta$ and so $|U|\leq rk\leq k/\delta$, as required.
\end{proof}

We now define what it means for a detailing to be $(\epsilon,k)$-good, a notion that will provide us a degree of predictability with respect to testing algorithms.

\begin{definition}[$\eps$-independent tuple]
	Fix $\eps\in (0,1)$ and let $\mu$ be a distribution over $\{0,1\}^n$, and $i_1, \ldots, i_k \in [n]$ be a set of integers. We say that the $k$-tuple $( i_1, \ldots, i_k )$ an \emph{$\eps$-independent tuple} with respect to $\mu$ if $\dtv\left(\mu|_{\{{{i_1}, \ldots,{i_k}}\}},\prod_{\ell=1}^k\mu|_{{i_\ell}}\right) \leq \eps$.
\end{definition}

\begin{definition}[$(\epsilon,k)$-good detailing] \label{def:eps_good_partition} Let $\mu$ be a distribution supported on $\zo^n$ and let $\xi$ be a detailing of $\mu$ over some $A$. Let $J\subseteq A$ be the set of elements where $a\in J$ if at least $(1-\epsilon)n^k$ of the $k$-tuples in $[n]$ are $\epsilon$-independent with respect to $\xi|_1^{2:a}$. We call the detailing {\em $(\epsilon,k)$-good} if $\Prx_{\ba\sim \xi|_2}[\ba\in J]\ge 1-\epsilon$, and refer to $J$ as the \emph{good set}.
\end{definition}

\subsection{From variable weak robustness to goodness}

We prove here that variable weak robustness implies goodness, as stated in the following lemma.

\begin{lemma}\label{lem:weakly_robust_is_good}
	If a detailing $\xi$ is $(\frac{\eps^6}{64k^3},k-1)$-weakly robust, then it is also $(\eps,k)$-good.
\end{lemma}

Before we prove the lemma, we set the stage with some technical tools and definitions. In the following, for a distribution $\mu$ supported on $\zo^n$, given a tuple $\alpha=(i_1,\ldots,i_k)\in[n]^k$ and an assignment $v=(v_{i_1},\ldots, v_{i_k})\in\zo^k$ we let $\mu|^{i_1:v_{i_1},\ldots,i_k:v_{i_k}}$  denote the conditioned distribution $\mu|^{\bx_\alpha=v}$. Note that this is not a new definition but just a new use for Definition~\ref{def:shorthand}.

\begin{lemma}\label{lem:dist_from_prod_v1}
	If the tuple $\{i_1,\ldots,i_k\}$ is not $\epsilon$-independent,
	then there exists an index $1\leq \ell\leq k$ and a set
	$V\subseteq\{0,1\}^{k-1}$, such that $\Prx_{\mu|_{\{{i_1},\ldots,{i_{\ell-1}},{i_{\ell+1}},\ldots,{i_k}\}}}[V]\geq \eps/2k$, and for every member $v=(v_{i_1},\ldots,v_{i_{\ell-1}},v_{i_{\ell+1}},\ldots,v_{i_k})\in V$ we have
	$\left|\mu|_{{i_\ell}}(1)-\mu|^{i_1:v_1,\ldots,i_{\ell-1}:v_{\ell-1},i_{\ell+1}:v_{\ell+1},\ldots,i_k:v_k}_{{i_\ell}}(1)\right|\geq\eps/2k$.
\end{lemma}

The proof of Lemma~\ref{lem:dist_from_prod_v1} follows from the following one.

\begin{lemma}\label{lem:diff_prod}
	Let $\mu$ be a distribution over $\{0,1\}^n$. Then $$\dtv\left(\mu|_{\{i_1,\ldots,i_k\}},\prod_{j=1}^k\mu|_{i_j}\right)\leq\sum_{j=1}^\ell\Ex_{\bx\sim\mu}\left[\left|\mu|^{i_1:\bx_{i_1},\ldots,i_{\ell-1}:\bx_{i_{\ell-1}}}_{i_\ell}(1)-\mu|_{i_\ell}(1)\right|\right]$$
\end{lemma}

\begin{proof}
	We use Observation~\ref{obs:TV_is_EMD}, and prove the statement by constructing an appropriate coupling (transfer distribution) $T$ which demonstrates this distance bound (with respect to the Kronecker norm) between $\nu\eqdef\mu|_{\{i_1,\ldots,i_k\}}$ and $\tau\eqdef\prod_{j=1}^k\mu|_{i_j}$. To draw $(\bx,\by)\sim T$ (where $(\bx,\by)\in(\{0,1\}^k)^2$), for $\ell\in [k]$ we draw $\bx_\ell$ and $\by_\ell$ inductively considering the values of $\bx_1,\ldots,\bx_{\ell-1}$ and $\by_1,\ldots,\by_{\ell-1}$. The base case (before we draw any bits) is trivial.

 To draw $\bx_\ell$ and $\by_\ell$ we first calculate $$\alpha_\ell=\Prx_{\bz\sim\mu|^{i_1:\bx_{1},\ldots,i_{\ell-1}:\bx_{{\ell-1}}}}[\bz_{j_\ell}=1]=\mu|^{i_1:\bx_{1},\ldots,i_{\ell-1}:\bx_{{\ell-1}}}_{i_\ell}(1)$$
 
 and $\beta_\ell=\mu|_{i_\ell}(1)$. Then if $\alpha_\ell\geq\beta_\ell$ we set $(\bx_\ell,\by_\ell)=(1,1)$ with probability $\beta_\ell$, $(\bx_\ell,\by_\ell)=(1,0)$ with probability $\alpha_\ell-\beta_\ell$, and $(\bx_\ell,\by_\ell)=(0,0)$ with probability $1-\alpha_\ell$. If $\alpha_\ell\leq\beta_\ell$ then we analogously set $(\bx_\ell,\by_\ell)=(1,1)$ with probability $\alpha_\ell$, $(\bx_\ell,\by_\ell)=(0,1)$ with probability $\beta_\ell-\alpha_\ell$, and $(\bx_\ell,\by_\ell)=(0,0)$ with probability $1-\beta_\ell$.
	
	It is not hard to see that this is indeed a coupling between $\nu$ and $\tau$. Also, for every $\ell$ clearly 
 $$\Prx_{(\bx,\by)\sim T}[\bx_\ell\neq\by_\ell]=\Ex_{\bx\sim \mu}\left[|\alpha_\ell-\beta_\ell|\right]=\Ex_{\bx\sim \mu}\left[\left|\mu|^{i_1:\bx_{i_1},\ldots,i_{\ell-1}:\bx_{j_{\ell-1}}}_{i_\ell}(1)-\mu|_{i_\ell}(1)\right|\right]$$ and hence by a union bound $$\Prx_{(\bx,\by)\sim T}[\bx\neq\by]\leq \sum_{j=1}^\ell\Ex_{\bx\sim\mu}\left[\left|\mu|^{i_1:\bx_{i_1},\ldots,i_{\ell-1}:\bx_{j_{\ell-1}}}_{i_\ell}(1)-\mu|_{i_\ell}(1)\right|\right]$$ completing the proof.
\end{proof}

\begin{proofof}{Lemma~\ref{lem:dist_from_prod_v1}}
	Assume that the conclusion of the lemma does not hold. We define for all $\ell\in [k]$ the set $V_\ell\subseteq\{0,1\}^{k-1}$ of all $v$ for which $\left|\mu|_{{i_\ell}}(1)-\mu|^{i_1:v_1,\ldots,i_{\ell-1}:v_{\ell-1},i_{\ell+1}:v_{\ell+1},\ldots,i_k:v_k}_{{i_\ell}}(1)\right|\geq\eps/2k$. We then note that $\delta_\ell=\Ex_{\bx\sim\mu}\left[\left|\mu|^{i_1:\bx_{i_1},\ldots,i_{\ell-1}:\bx_{j_{\ell-1}}}_{i_\ell}(1)-\mu|_{i_\ell}(1)\right|\right]$ satisfies
	\begin{align*}
		 \delta_\ell & = \Ex_{\bx\sim\mu}\left[\left|\Ex_{\bz\sim\mu}\left[\mu|^{i_1:\bx_{i_1},\ldots,i_{\ell-1}:\bx_{j_{\ell-1}},i_{\ell+1}:\bz_{i_{\ell+1}},\ldots,i_k:\bz_k}_{i_\ell}(1)-\mu|_{i_\ell}(1)\right]\right|\right] \\
		& \leq \Ex_{\bx\sim\mu}\left[\left|\mu|^{i_1:\bx_{i_1},\ldots,i_{\ell-1}:\bx_{j_{\ell-1}},i_{\ell+1}:\bx_{i_{\ell+1}},\ldots,i_k:\bx_k}_{i_\ell}(1)-\mu|_{i_\ell}(1)\right|\right] \\
		& \leq \Prx_{\mu|_{\{{i_1},\ldots,{i_{\ell-1}},{i_{\ell+1}},\ldots,{i_k}\}}}[V_\ell]\cdot 1+\Prx_{\mu|_{\{{i_1},\ldots,{i_{\ell-1}},{i_{\ell+1}},\ldots,{i_k}\}}}[\{0,1\}^{k-1}\setminus V_\ell]\cdot\frac{\eps}{2k}=\frac{\eps}{k}
	\end{align*}
	Then Lemma~\ref{lem:diff_prod} would have implied that $\dtv(\mu|_{\{{i_1},\ldots,{i_k}\}},\prod_{\ell=1}^k\mu|_{{i_\ell}})\leq\sum_{\ell\in [k]}\delta_\ell\leq\eps$, in contradiction to our premise.
\end{proofof}

\begin{definition}[Pivot tuple]\label{defi:pivot} Fix $k\in [n]$,
	let $\mu$ be a distribution over $\{0,1\}^n$ and $\eps \in (0,1)$. We say that a $(k-1)$-tuple $({i_1},\ldots,{i_{\ell-1}},{i_{\ell+1}},\ldots,{i_k})\in[n]^{k-1}$ is a \emph{pivot tuple for ${i_\ell}$ with respect to $\mu$} if the following holds.  There exists a set $V\subseteq\{0,1\}^{k-1}$ such that $\Prx_{\mu|_{\{{i_1},\ldots,{i_{\ell-1}},{i_{\ell+1}},\ldots,{i_k}\}}}[V]\geq \eps/2k$, and for every $v\in V$ we have
	$\left|\mu|_{{i_\ell}}(1)-\mu|^{i_1:v_1,\ldots,i_{\ell-1}:v_{\ell-1},i_{\ell+1}:v_{\ell+1},\ldots,i_k:v_k}_{{i_\ell}}(1)\right|\geq\eps/2k$.
\end{definition}

We now prove, in two stages, that if a detailing is not $(\eps,k)$-good, then there exist many tuples that are pivot over much of the weight for many variables. 

\begin{observation}\label{obs:pivot_many}
	Fix $k\in [n]$, let $\mu$ be a distribution over $\{0,1\}^n$ and $\eps \in (0,1)$. If there exists a set $\calS\in [n]^k$ of at least $\eps n^k$ tuples that are not $\eps$-independent, then there is a set $\calS'\subseteq [n]^{k-1}$ of at least $\frac{\eps}{2k}\cdot n^{k-1}$ tuples such that each of them is a pivot tuple for at least $\frac{\eps}{2k}\cdot n$ variables outside of it.
\end{observation}

\begin{proof}
	Note that Lemma \ref{lem:dist_from_prod_v1} immediately implies that for every $U\in\calS$ there exists $i_U\in U$ so that $U\setminus\{i_U\}$ is a pivot tuple with respect to $i_U$. By the premise of the observation, if we uniformly pick a tuple $U'\in [n]^{k-1}$ and $i\in [n]$, then with probability at least $\eps$ we will obtain that $U'\cup\{i\}\in\calS$ (and also in particular $|U'\cup\{i\}|=k$), and so with probability at least $\eps/k$ we obtain that $U'$ is a pivot tuple for $i$ (since the probability for $i=i_{U'\cup\{i\}}$ conditioned on $U'\cup\{i\}=U\in\calS$ is $1/k$).
	
	A reverse Markov inequality then implies that there exists a set $\calS'$ of at least $\frac{\eps}{2k}\cdot n^{k-1}$ tuples so that for every $U'\in\calS'$ there exist at least $\frac{\eps}{2k}\cdot n$ many variables for which it is a pivot tuple.
\end{proof}

\begin{lemma}\label{lem:pivot_much_many}
	If a detailing $\xi$ of $\mu$ over $A$ is not $(\eps,k)$-good, then there exists a set $\calS\subseteq [n]^{k-1}$ of at least $\frac{\eps}{4k}\cdot n^{k-1}$ tuples, so that for every $U\in\calS$ there exists a set $K_U\subseteq A$ for which $\Prx_{\xi|_2}[K_U]\geq\frac{\eps^2}{4k}$, satisfying that for every $a\in K_U$ the tuple $U$ is a pivot tuple for at least $\frac{\eps}{2k}\cdot n$ many variables with respect to $\xi|_1^{2:a}$.
\end{lemma}

\begin{proof}
	We let $K$ be the set of all $a\in A$ for which $\xi|^{2:a}$ admits a set $\calS_a\subseteq [n]^k$ of at least $\eps n^k$ tuples which are not $\eps$-independent. By the premise of $\xi$ not being $(\eps,k)$-good we have $\Prx_{\xi|_2}[K]\geq\epsilon$. By invoking Observation \ref{obs:pivot_many} for every $a\in K$, we obtain $\calS'_a\subseteq [n]^{k-1}$ of size at least $\frac{\eps}{2k}\cdot n^{k-1}$ so that every $U\in\calS'_a$ is a pivot tuple for at least $\frac{\eps}{2k}\cdot n$ many variables with respect to $\xi|^{2:a}$.
	
	In particular, if we draw uniformly a tuple $U\in [n]^{k-1}$ and according to $\xi|^{2:K}_2$ a value $a\in K$, with probability at least $\frac{\eps}{2k}$ the tuple $U$ is pivot tuple for at least $\frac{\eps}{2k}\cdot n$ many variables with respect to $\xi|^{2:a}$. To finalize, we use a reverse Markov inequality, and find a set $\calS\subseteq [n]^{k-1}$ of at least $\frac{\eps}{4k}\cdot n^{k-1}$ many tuples, so that every $U\in\calS$ satisfies the above assertion with respect to a set $K_U\subseteq K$ satisfying $\Prx_{\xi|^{2:K}_2}[K_U]\geq\frac{\eps}{4k}$ and hence $\Prx_{\xi|_2}[K_U]\geq\frac{\eps^2}{4k}$.
\end{proof}

We are now ready to prove that a sufficient (weak) robustness of a detailing implies its goodness.

\begin{proofof}{Lemma~\ref{lem:weakly_robust_is_good}}
	As $\xi$ is not $(\eps,k)$-good, we use Lemma \ref{lem:pivot_much_many} to find the set $\calS\subseteq [n]^{k-1}$ satisfying its conclusion, and also note the set $K_U\subseteq A$ promised by the lemma for every $U\in\calS$. Additionally, for every $U\in\calS$ and $a\in K_U$ we note the set $I_{U,a}\subseteq n$ of (at least $\frac{\eps}{2k}\cdot n$ many) variables for which $U$ is a pivot tuple with respect to $\xi|_1^{2:a}$.
	
	To conclude, we bound from below $\mathrm{Ind}(\xi^U)-\mathrm{Ind}(\xi)$ for every $U\in\mathcal S$. We set $\Delta\eqdef\mathrm{Ind}(\xi^U)-\mathrm{Ind}(\xi)$ and analyze it:
	\begin{align*}
		\Delta & = \frac1n\sum_{i\in [n]}\left(\Ex_{(\ba,\bv)\sim\xi^U|_{2,3}}\left[\Prx_{\bx\sim\xi^U|_1^{2,3:(\ba,\bv)}}[\bx_i=1]^2\right]-\Ex_{\ba\sim\xi|_2}\left[\Prx_{\bx\sim\xi|_1^{2:\ba}}[\bx_i=1]^2\right]\right)\\
		& = \frac1n\sum_{i\in [n]}\left(\Ex_{\ba\sim\xi|_2}\left[\Ex_{\bv\sim\xi^U|^{2:\ba}_3}\left[\Prx_{\bx\sim\xi^U|_1^{2,3:(\ba,\bv)}}[\bx_i=1]^2-\Prx_{\bx\sim\xi|_1^{2:\ba}}[\bx_i=1]^2\right]\right]\right)\\
		& = \Ex_{\ba\sim\xi|_2}\left[\frac1n\sum_{i\in [n]}\left(\Ex_{\bv\sim\xi^U|^{2:\ba}_3}\left[\Prx_{\bx\sim\xi^U|_1^{2,3:(\ba,\bv)}}[\bx_i=1]^2\right]-\Ex_{\bv\sim\xi^U|^{2:\ba}_3}\left[\Prx_{\bx\sim\xi|_1^{2,3:(\ba,\bv)}}[\bx_i=1]\right]^2\right)\right]\\
		& \geq \Prx_{\ba\sim\xi|_2}[K_U]\Ex_{\ba\sim\xi|^{2:K_U}_2}\left[\frac1n\sum_{i\in I_{U,\ba}}\left(\Ex_{\bv\sim\xi^U|^{2:\ba}_3}\left[\Prx_{\bx\sim\xi^U|_1^{2,3:(\ba,\bv)}}[\bx_i=1]^2\right]-\Ex_{\bv\sim\xi^U|^{2:\ba}_3}\left[\Prx_{\bx\sim\xi|_1^{2,3:(\ba,\bv)}}[\bx_i=1]\right]^2\right)\right]\\
		& \geq \frac{\eps^2}{4k}\cdot\Ex_{\ba\sim\xi|^{2:K_U}_2}\left[\frac{|I_{U,\ba}|}n\cdot \left(\frac{\eps}{2k}\right)^2\frac{\eps}{2k}\right] \geq \frac{\eps^6}{64k^5},
	\end{align*}
	where in the last line we used Observation~\ref{obs:square_bump} for the variable $\bX_{a,i}=\Prx_{\bx\sim\xi|_1^{2,3:(a,v)}}[\bx_i=1]$ and the probability space $\bv\sim\xi^U|_3^{2:a}$, for every $a\in K_U$ and $i\in I_{U,a}$, with $\alpha=\beta=\frac{\eps}{2k}$.
\end{proofof}

\subsection{From variable weak robustness to general robustness}

In this section we prove that weak variable robustness implies even the seemingly much stronger notion of general robustness, as stated in the following lemma.

\begin{lemma} \label{lem:Translate}Fix $\gamma\in (0,1)$, $\ell\in \N$ and let $\xi$ be a detailing of $\mu$ with respect to $A$. If $\xi$ is $(\gamma/2,C\cdot\frac{\ell^8}{\gamma^{38}})$-weakly robust for some absolute constant $C>1$, then it is $(\gamma,\ell)$-robust.
\end{lemma}

The above lemma is a direct consequence of the following lemma (which will be proved here), which is a slightly strengthened restatement of Lemma \ref{lem:Translate} in a ``contrary form''.

\begin{lemma} \label{lem:conTranslate}Fix $\gamma\in(0,1)$, let $\xi$ be a detailing of $\mu$ with respect to $A$, and let $\zeta$ be a refinement of $\xi$ with respect to $A\times B$ such that $\ix(\zeta)-\ix(\xi)\ge \gamma$. Then, there exists  $k=O\left(\frac{|B|^8}{\gamma^{38}}\right)$ such that the following holds. A uniformly random $k$-tuple $\balpha\in [n]^k$ satisfies $\ix(\zeta)-\ix(\xi^{\balpha})<\gamma/2$, with probability at least $1/2$.
\end{lemma}

To prove Lemma~\ref{lem:conTranslate}, we need the following \emph{progress lemma} to be proved in Section~\ref{ssec:Progress}.
\begin{lemma}\label{lem:AdditionalVar}
	If $\xi$ is a detailing of $\mu$ with respect to $ A$, and $\zeta$ is
	a refinement of $\xi$ with respect to $A\times B$ satisfying $\mathrm{Ind}(\zeta)-\mathrm{Ind}(\xi)\geq\gamma$, then for at least 
	$\Omega\left(\frac{\gamma^{19}n}{|B|^4}\right)$ many $i\in [n]$ we have
	$\mathrm{Ind}(\xi^{\{i\}})-\mathrm{Ind}(\xi)\geq\Omega\left(\frac{\gamma^{19}}{|B|^4}\right)$.
\end{lemma}

We now define some notations. Let $\balpha=(\bx_1,\ldots,\bx_k)\in [n]^k$ be a uniformly random $k$-tuple and define $\bX_1,\dots,\bX_k$ to be the following random variables for every $t \in [k]$ (which depend on $\balpha$),
\[\bX_t=
\begin{cases}
	0, &\text{if   }\;\ix(\xi^{ \{\bx_1,\ldots,\bx_t\}})-\ix(\xi^{ \{\bx_1,\ldots,\bx_{t-1}\}})<\rho\\
	1, &\text{otherwise}
\end{cases},
\]
where $\rho=\Omega\left(\frac{\gamma^{19}}{|B|^4}\right)$, with the implicit coefficient in this expression to be later determined in the proof of Lemma \ref{lem:AdditionalVar}. Next, for every $t\in[k]$ we define the random variable $\bR_t\in [n]\cup \{0\}$ as 
\[\bR_t=
\begin{cases}
	\bx_t, &\text{if   }\;\bX_t
	=1\\
	0, &\text{otherwise}
\end{cases},
\]

Now we have the following lemma which lower bounds the probability of the event $\bR_t \neq 0$ conditioned on the values of $\bR_1, \ldots, \bR_{t-1}$ for every $t \in [k]$.

\begin{lemma}\label{lem:problowerbound} 
	Let $t \in[k]$ and $r_1,\ldots,r_{t-1} \in [n]\cup \{0\}$. If $\ix\left(\zeta\right)-\ix(\xi^{ \{r_1,\ldots,r_{t-1}\}\setminus \{0\}}) \ge\gamma/2$, then $$\Prx_{\balpha\sim[n]^k}[\bR_t\neq 0\mid \bR_1=r_1,\ldots,\bR_{t-1}=r_{t-1}]\geq\rho,$$ where $\rho=\Omega\left(\frac{\gamma^{19}}{|B|^4}\right)$.
\end{lemma}

\begin{proof} We apply Lemma~\ref{lem:AdditionalVar} with the detailing $\xi^{ \{r_1,\ldots,r_{t-1}\}\setminus \{0\}}$, and $\gamma/2$ instead of $\gamma$. By our assumptions we have $\ix\left(\zeta\right)-\ix(\xi^{\{r_1,\ldots,r_{t-1}\}\setminus \{0\}}) \ge \gamma/2$, so there exists a set $S\subseteq[n]$ of size at least $\Omega\left(\frac{\gamma^{19}n}{|B|^4}\right)$, such that for any $x_t\in S$, we have 
	\begin{align}
		\ix(\xi^{ (\{r_1,\ldots,r_{t-1}\}\setminus \{0\})\cup\{x_t\}})-\ix(\xi^{ \{r_1,\ldots,r_{t-1}\}\setminus \{0\}})\ge \Omega\left(\frac{\gamma^{19}}{|B|^4}\right). \label{eq:blahblah}
	\end{align}
	Note that if $\bx_t\in S$, then since Equation~(\ref{eq:blahblah}) holds, we must have that $\bX_t=1$ and $\bR_t=\bx_{t}\neq 0$. As the probability that $\bx_t\in S$ is at least $\Omega\left(\frac{\gamma^{19}}{|B|^4}\right)$ the lemma follows.
\end{proof}

Now we will use the following lemma from \cite{AFL23}.

\begin{lemma}[\cite{AFL23}]\label{lem:goalineq}
	Let $\mathcal{G}\subset \mathbb{R}^*$ be a set of goal sequences, satisfying that if $u$ is a prefix of $v$ and $u\in\mathcal{G}$ then $v\in\mathcal{G}$. Additionally, let $\bR_1,\ldots,\bR_M$ be a set of random variables and $p_1,\ldots,p_M$ be values in $[0,1]$, such that for every $1\leq t\leq M$ and $v=(r_1,\ldots,r_{t-1})\in\mathbb{R}^{t-1}\setminus\mathcal{G}$ (that can happen with positive probability), we have $\Prx\left[\bR_i\ne 0 \mid \bR_1=r_1,\ldots,\bR_{t-1}=r_{t-1}\right] \ge p_t$. For every $1 \le t \le M$, let $\bX_t \in \{0,1\}$ be an indicator for $\bR_t \ne 0$ and $\bX = \sum_{t=1}^M \bX_t$. Under these premises, for every $0 < \delta < 1$, the following holds:
	\[\Prx\left[((\bR_1,\ldots,\bR_M)\notin\mathcal{G}) \wedge \left(\bX < (1-\delta) \sum_{t=1}^M p_t\right)\right]
	< \left(\frac{e^{-\delta}}{(1 - \delta)^{1 - \delta}} \right)^{\sum_{t=1}^M p_t}\]
\end{lemma}

\begin{proofof}{Lemma~\ref{lem:conTranslate}} We define the set of goals as follows
	\[ \calG\eqdef \left\{ (r_1,\ldots ,r_k)\in ([n]\cup \{0\})^k :  \ix(\zeta)-\ix(\xi^{ \{r_1,\ldots,r_{k}\}\setminus \{0\}})<\gamma/2 \right\}\]
 Set $M=k$ for some $k$ to be defined later. Let $p_1=\ldots=p_k=\rho$, where $\rho=\frac{\gamma^{19}}{C\cdot|B|^4}$ for some large enough absolute constant $C>1$. By following  Lemma~\ref{lem:goalineq} and Lemma~\ref{lem:problowerbound}, we have that for every $\delta\in(0,1)$,
	\[\Prx\left[((\bR_1,\ldots,\bR_k)\notin\mathcal{G}) \wedge \left(\bX < (1-\delta) \sum_{t=1}^k \rho\right)\right]
	< \left(\frac{e^{-\delta}}{(1 - \delta)^{1 - \delta}} \right)^{\sum_{t=1}^k \rho}.\]
	As $\left(\frac{e^{-\delta}}{(1 - \delta)^{1 - \delta}} \right) \leq e^{-\frac{\delta^2}{2}}$ for $\delta \in (0,1)$, we have:
	\[\Prx\left[((\bR_1,\ldots,\bR_k)\notin\mathcal{G}) \wedge \left(\bX < (1-\delta) \rho \cdot k\right)\right]
	< e^{-\frac{\delta^2 \rho k}{2}}, \] which implies that
	\[\Prx\left[((\bR_1,\ldots,\bR_k) \in \mathcal{G}) \lor \left(\bX \geq (1-\delta) \rho \cdot k\right)\right]
	\geq 1 - e^{-\frac{\delta^2 \rho k}{2}} .\]
	
	Next, we show that the event $\bX\ge (1-\delta)\rho \cdot k$ is empty. We set $k\eqdef \lceil{8}/{\rho^2}\rceil=O\left(\frac{|B|^8}{\gamma^{38}}\right)$ and $\delta=1/2$. Then, $(1-\delta)\rho\cdot k\ge 1/\rho+1$, and by the definition of the random variables $\bX_t$ we have that 
	\[\ix(\xi^{ \{ \bx_1,\ldots,\bx_{k}\}\setminus \{0\}})\ge (1-\delta)\rho\cdot k \cdot \rho>1,\]
	which is a contradiction to the fact that the index of a detailing is bounded by $1$. 
	
	This implies that for the setting of $k= \lceil{8}/{\rho^2}\rceil$ and $\delta=1/2$, $\Prx_{\balpha\sim[n]^k}[(\bR_1,\ldots,\bR_k)\in \calG]\ge 1/2$. Therefore, by the definition of $\calG$, we have that with probability at least $1/2$ for a uniformly random $k$-tuple $\balpha=(\bx_1,\ldots,\bx_k)$ it holds that $\ix(\zeta)-\ix(\xi^{ \{\bx_1,\ldots,\bx_{k}\}\setminus \{0\}})<\gamma/2$ as required.
\end{proofof}
\subsubsection{Proof of the progress lemma} \label{ssec:Progress}

To prove Lemma~\ref{lem:AdditionalVar} we begin with the following.
\begin{lemma}\label{lem:existindexvector}
	Fix $\gamma\in(0,1)$  and suppose that there is a detailing $\xi$ of $\mu$ with respect to $A$ and a refinement $\zeta$ of $\xi$ over $A\times B$ such that $\mathrm{Ind}(\zeta)-\mathrm{Ind}(\xi)\ge \gamma$. Then the following hold:
	\begin{enumerate}
		\item There exists a set $A' \subseteq A$ such that $\xi|_2(A') \geq \gamma/2$, and for every $a\in A'$ there exists $b_a\in B$ (depending on $a$) such that $\frac{\zeta|_{2,3}(a,b_a)}{\xi|_2(a)}=\zeta|_3^{2:a}(b_a)\ge \frac{\gamma}{4|B|}$ for which 
		$$ \Ex_{i\sim [n]} \left[ \Prx_{\bx \sim \zeta|_1^{2,3:(a,b_a)} }[\bx_i=1]^2 - \Prx_{\bx \sim \xi|_1^{2:a}}[\bx_i=1]^2 \right] \geq \gamma/4.
		$$
		\item For every fixed $a\in A'$ and $b\in B$, $\zeta|_3^{2:a}(b)\le 1-\gamma/8$. 
	\end{enumerate}
\end{lemma}
\begin{proof} We start by proving the first item. By the premise that $\mathrm{Ind}(\zeta)-\mathrm{Ind}(\xi)\ge \gamma$ we have:
	\begin{align*}
		\gamma&\le\Ex_{i\sim [n]}\left[\Ex_{(\ba,\bb)\sim\zeta|_{2,3}}\left[\Prx_{\bx \sim \zeta|_1^{2,3:(\ba,\bb)} }\left[\bx_i=1\right]^2\right]-\Ex_{\ba \sim\xi|_2}\left[\Prx_{\bx \sim \xi|_1^{2:\ba}}[\bx_i=1]^2\right]\right]\\
		&=\Ex_{i\sim [n]}\left[\Ex_{\ba\sim \xi|_2}\left[\Ex_{\bb\sim\zeta|_3^{2:\ba}}\left[\Prx_{\bx \sim  \zeta|_1^{2,3:(\ba,\bb)} }\left[\bx_i=1\right]^2\right]-\Prx_{\bx\sim \xi|_1^{2:\ba} }\left[\bx_i=1\right]^2\right]\right]\\
		&=\Ex_{\ba\sim \xi|_2}\left[\Ex_{\bb\sim\zeta|_3^{2:\ba}}\left[\Ex_{i\sim [n]}\left[\Prx_{\bx \sim  \zeta|_1^{2,3:(\ba,\bb)} }\left[\bx_i=1\right]^2-\Prx_{\bx\sim \xi|_1^{2:\ba} }\left[\bx_i=1\right]^2\right]\right]\right].
	\end{align*}
	Therefore, by using the reverse Markov's inequality (Lemma~\ref{lem:reverseMarkov}) we have that there exists a set $A'\subseteq A$ with $\Prx_{\xi|_2}[A']\geq\gamma/2$ such that for every $a\in A'$ it holds that 
	\begin{align}  \gamma/2&\le\Ex_{\bb\sim \zeta|_3^{2:a}}\left[\Ex_{i\sim [n]}\left[\Prx_{\bx \sim \zeta|_1^{2,3:(a,\bb)} }\left[\bx_i=1\right]^2-\Prx_{\bx\sim \xi|_1^{2:a} }\left[\bx_i=1\right]^2\right]\right]\label{Blah}.
	\end{align}
	By using the reverse Markov's inequality again, for every $a\in A'$ there exists a set $B_a$ with $\Prx_{\zeta|_3^{2:a}}[B_a]\geq\gamma/4$ so that for every $b\in B_a$,
	\begin{align}  \gamma/4&\le\Ex_{i\sim [n]}\left[\Prx_{\bx \sim \zeta|_1^{2,3:(a,b)} }\left[\bx_i=1\right]^2-\Prx_{\bx\sim \xi|_1^{2:a} }\left[\bx_i=1\right]^2\right]. \label{blahblah}
	\end{align}
	Finally, by an averaging argument (noting that $|B_a|\leq|B|$), for every $a\in A'$ there exists $b_a\in B_a$ with $\zeta|_3^{2:a}(b_a)\ge \gamma/4|B|$ which satisfies Equation (\ref{blahblah}), completing the proof of the first item.
	
	For the second item, note that from Equation~(\ref{Blah}) for any fixed $a\in A'$ we have that \begin{align}
		\gamma/2\le\Ex_{i\sim [n]}\left[\Ex_{\bb\sim\zeta|_3^{2:a}}\left[\Prx_{\bx \sim \zeta|_1^{2,3:(a,\bb)} }\left[\bx_i=1\right]^2\right]-\Prx_{\bx\sim \xi|_1^{2:a} }\left[\bx_i=1\right]^2\right]=\mathrm{Ind}(\zeta|^{2:a})-\mathrm{Ind}(\xi|^{2:a}).\label{eqn:747-748}
	\end{align}
	Suppose that there is $b^*\in B$ such that $\zeta|_3^{2:a}(b^*)>1-\gamma/8$. Then,
	\begin{align}
		\ix(\zeta|^{2:a})&=\Ex_{i\sim [n]}\left[\Ex_{\bb\sim\zeta|_3^{2:a}}\left[\Prx_{\bx \sim \zeta|_1^{2,3:(a,\bb)} }[\bx_i=1]^2\right]\right]=\Ex_{\bb\sim\zeta|_3^{2:a}}\left[\mathrm{Ind}(\zeta|^{2,3:(a,\bb)})\right]\notag\\
		&=\zeta|_3^{2:a}(b^*)\cdot\ix(\zeta|^{2,3:(a,b^*)})+\sum_{b\neq b^*\in B}\zeta|_3^{2:a}(b)\cdot\ix(\zeta|^{2,3:(a,b)})\notag\\
		&\le \ix(\zeta|^{2,3:(a,b^*)})+\frac{\gamma}{8}. \label{eq:hehe1}
	\end{align}
	On the other hand, 
	\begin{align*}
		\mathrm{Ind}(\xi|^{2:a})&=\Ex_{i\sim [n]}\left[\Ex_{\bb\sim\zeta|_3^{2:a}}\left[\Prx_{\bx \sim \zeta|_1^{2,3:(a,\bb)} }[\bx_{i}=1]\right]^2\right]\ge \Ex_{i\sim [n]}\left[\left(\zeta|_3^{2:a}(b^*)\cdot\Prx_{\bx \sim \zeta|_1^{2,3:(a,b^*)} }[\bx_{i}=1]\right)^2\right]\\
		&\ge (1-\gamma/8) \ix(\zeta|^{2,3:(a,b^*)}).
	\end{align*}
	Plugging in Equation~(\ref{eq:hehe1}) we get that 
	\[\ix(\zeta|^{2:a})\le \frac{\ix(\xi|^{2:a})}{(1-\gamma/8)}+\gamma/8\le\ix(\xi|^{2:a})(1+\gamma/8)+\gamma/8\le \ix(\xi|^{2:a})+\frac{\gamma}{4},\]
	which is a contradiction to Equation~(\ref{eqn:747-748}). 
\end{proof}

\begin{definition}[Complement outside $b$] Given a detailing $\zeta$ of $\mu$ over $A\times B$, 
	for every $a\in A$ and $b\in B$ we define $\kappa_{a,b}$, the \emph{complement outside $b$}, as the distribution
	\[\zeta|_1^{2:a,3:B\setminus\{b\}}=\frac{1}{1-\zeta|_3^{2:a}(b)}\sum_{b'\neq b\in B}\zeta|_3^{2:a}(b')\cdot \zeta|_1^{2,3:(a,b')}.\]
\end{definition}
\begin{lemma} \label{lem:BigSetS_a}Fix $(a,b)\in A\times B$ and $\gamma\in(0,1)$. If $$ \Ex_{i\sim [n]} \left[ \Prx_{\bx \sim \zeta|_1^{2,3:(a,b)} }[\bx_i=1]^2 - \Prx_{\bx \sim \xi|_1^{2:a}}[\bx_i=1]^2 \right] \geq \gamma,
	$$ then there exists a set $S_a\subseteq [n]$ of size at least $\gamma n/2$ such that for every $i\in S_a$ it holds that
	\[\Prx_{\bx \sim \zeta|_1^{2,3:(a,b)} }[\bx_i=1]-\Prx_{\bx \sim \kappa_{a,b}}[\bx_i=1]\ge\gamma/4.\]
\end{lemma}
\begin{proof} By the reverse Markov inequality there exists a set $S_a\subseteq[n]$ of size at least $\gamma n/2$ such that for each $i\in S_a$ it holds that $\Prx_{\bx \sim \zeta|_1^{2,3:(a,b)} }[\bx_i=1]^2 - \Prx_{\bx \sim \xi|_1^{2:a}}[\bx_i=1]^2\ge \gamma/2$. Since $c^2-d^2\le 2|c-d|$ for $c,d\in[0,1]$, for all $i\in S_a$ we have 
\begin{align}
    \Prx_{\bx \sim \zeta|_1^{2,3:(a,b)} }[\bx_i=1] - \Prx_{\bx \sim \xi|_1^{2:a}}[\bx_i=1]\ge \gamma/4.\label{eq:Difference!}
\end{align}
	
	Note that since $\Prx_{\bx \sim \xi|_1^{2:a}}[\bx_i=1]=\zeta|_3^{2:a}(b)\cdot\Prx_{\bx \sim \zeta|_1^{2,3:(a,b)} }[\bx_i=1]+(1-\zeta|_3^{2:a}(b))\Prx_{\bx \sim \kappa_{a,b}}[\bx_i=1]$
	and  $\Prx_{\bx \sim \zeta|_1^{2,3:(a,b)} }[\bx_i=1] > \Prx_{\bx \sim \xi|_1^{2:a}}[\bx_i=1]$, we must have that $\Prx_{\bx \sim \kappa_{a,b}}[\bx_i=1]<\Prx_{\bx \sim \xi|_1^{2:a}}[\bx_i=1]$. Combining that with Equation~(\ref{eq:Difference!}) implies the lemma.
\end{proof}

\begin{lemma}\label{lem:largevariance} For any $a\in A$, if there exists $b_a\in B$ and $S_a\subseteq [n]$
such that for all $i\in S_a$ it holds that 	\[\Prx_{\bx \sim \zeta|_1^{2,3:(a,b_a)} }[\bx_i=1]-\Prx_{\bx \sim \kappa_{a,b_a}}[\bx_i=1]\ge\gamma,\] then it holds  that $\Var_{\bx\sim \xi|^{2:a}_1}\left[\sum_{i\in S_a}\bx_i\right]\ge \gamma^2\delta_1\delta_2|S_a|^2$, where $\delta_1=\zeta|_3^{2:a}(b_a)$ and $\delta_2=1-\delta_1$.
\end{lemma}
\begin{proof} Let $\bX\eqdef\sum_{i\in S_a}\bx_i$ and let $\bZ\in\{1,2\}$ be a random variable such that $\bZ=1$ with probability $\delta_1$ and $\bZ=2$ with probability $\delta_2$. 
	We can think of the distribution of $\bx\sim \xi|_1^{2:a}$ as first drawing $\bZ$ and then drawing $\bx$  from $\zeta|_{1}^{2,3:(a,b_a)}$ or $\kappa_{a,b_a}$ based on whether $\bZ=1$ or $\bZ=2$, respectively.
	By using Cauchy-Schwartz inequality (Lemma~\ref{lem:CS}),
	\[\Ex_{\bx \sim \xi|_1^{2:a}}[\bX^2]=\Ex_{\bZ}\left[\Ex_{\bx}[\bX^2\mid \bZ]\right]\ge \Ex_{\bZ}\left[\Ex_{\bx}[\bX\mid \bZ]^2\right].\]
	By the premise of the lemma we have that 
	\[\Ex_{\bx}[\bX\mid\bZ=1]-\Ex_{\bx}[\bX\mid \bZ=2]\ge \gamma|S_{a}|.\]
	Combining the above (and using $\delta_1+\delta_2=1$),
	\begin{align*}
		\Varx_{\bx \sim \xi|_1^{2:a}}[\bX]&=\Ex[\bX^2]-\Ex[\bX]^2\ge \Ex_\bZ\left[\Ex[\bX\mid \bZ]^2\right]-\left(\Ex_{\bZ}\left[\Ex[\bX\mid \bZ]\right]\right)^2\\
		&=\delta_1\Ex[\bX\mid \bZ=1]^2+\delta_2\Ex[\bX\mid\bZ=2]^2-\delta_1^2\Ex[\bX\mid \bZ=1]^2-\delta_2^2\Ex[\bX\mid \bZ=2]^2\\
		&\qquad-2\delta_1 \delta_2\Ex[\bX\mid \bZ=1]\Ex[\bX\mid \bZ=2]\\
		&=\delta_1(1-\delta_1)\Ex[\bX\mid \bZ=1]^2+\delta_2(1-\delta_2)\Ex[\bX\mid \bZ=2]^2-2\delta_1 \delta_2\Ex[\bX\mid \bZ=1]\cdot 
  \Ex[\bX\mid \bZ=2]\\
		&=\delta_1 \delta_2\left(\Ex[\bX\mid \bZ=1]-\Ex[\bX\mid \bZ=2]\right)^2\\
		&\ge\gamma^2\delta_1 \delta_2 |S_a|^2
	\end{align*}
	This completes the proof of the lemma.
\end{proof}

Now  we prove that given that the conclusion of Lemma~\ref{lem:largevariance} holds, that is, $\Var_{\bx \sim \xi|_1^{2:a}}[\sum_{i\in S_a}\bx_{i}]$ is large, there exists a subset of indices $S'_{a} \subseteq S_a$ such that for every index $i \in S'_{a}$, there is a large number of indices $j \in S_a$ with large covariance between $\bx_i$ and $\bx_j$.

\begin{lemma}\label{lem:largeCov}
	Let $S_a \subseteq [n]$ be such that $\Var_{\bx \sim \xi|^{2:a}_1}\left[\sum_{i\in S_a}\bx_i\right]\ge \beta |S_a|^2$ for some $\beta>0$. Then there is a subset $S'_{a} \subseteq S_a$ with $|S'_{a}|\geq \beta|{S_a}|/4$, such that for every $i \in S'_{a}$ there exists a set $A_{i,a}\subseteq S_a$ with at least $\beta|S_a|/2$ many indices, where for all $j\in A_{i,a}$ we have $\Cov_{\bx \sim \xi|^{2:a}_1}\left(\bx_i,\bx_{j}\right)\geq \beta/4$.
\end{lemma}
\begin{proof}
	For every $i\in S_a$ we define 
	\[A_{i,a}\eqdef\{j\in S_a: \Cov(\bx_i,\bx_j)> \beta/4\},\] and let \[S'_a\eqdef\{i\in S_a : |A_{i,a}|\ge \beta |S_a|/2\}.\] Then, for all $i,j\in S_a$, using the fact that $\Cov(\bx_i,\bx_j)\le 1$ for all $i,j\in[n]$ we have:
	\begin{align*}
		\beta |S_a|^2&\le \Varx_{\bx \sim \xi|_1^{2:a}}\left[\sum_{i\in S_a}\bx_i\right]=\sum_{i,j\in S_a}\Cov(\bx_i,\bx_j)\\
		&=\sum_{i\in S'_a}\left(\sum_{j\in A_{i,a}}\Cov(\bx_i,\bx_j)+\sum_{j\notin A_{i,a}}\Cov(\bx_i,\bx_j)\right)+\sum_{i\notin S'_a}\left(\sum_{j\in A_{i,a}}\Cov(\bx_i,\bx_j)+\sum_{j\notin A_{i,a}}\Cov(\bx_i,\bx_j)\right)\\
		&\le |S'_a||S_a|+ |S_a\setminus S'_a|(\beta|S_a|/2 + \beta|S_a|/4)\\
		&\le |S'_a||S_a| + |S_a|^2\cdot \frac{3 \beta}{4} ,
	\end{align*}
	This implies that $\frac{\beta}{4}|S_a|\le |S'_a|$, which finishes the proof.
\end{proof}

\begin{lemma} \label{lem:Indexincrease}
	Fix $\gamma\in(0,1)$  and suppose that there is a detailing $\xi$ of $\mu$ with respect to $A$ and a refinement $\zeta$ of $\xi$ over $A\times B$ such that $\mathrm{Ind}(\zeta)-\mathrm{Ind}(\xi)\ge \gamma$. Then there exists a set $A'$ with $\Prx_{\xi|_2}[A']\geq\gamma/2$, and a set $S'_a\subseteq [n]$ of size at least $\frac{\gamma^5}{2^{10}|B|}n$ for every $a\in A'$, so that for all $i^*\in S'_a$ we have
	\[\mathrm{Ind}((\xi|^{2:a})^{\{i^*\}})-\mathrm{Ind}(\xi|^{2:a})\ge \frac{\gamma^{13}}{2^{28}|B|^3}\]
\end{lemma}
\begin{proof}
	Before constructing $A'$ and $S'_a$, let us prove an index lower bound for every $a\in A$ and $i^*\in [n]$. Let $\bZ_i^a$ be a random variable defined by the following process. We draw $\bj\in \zo$ such that $\bj=1$ with probability $\Prx_{\bx \sim \xi|_1^{2:a}}[\bx_{i^*}=1]$ and let $\bZ_i^a=\Prx_{\bx \sim \xi|_1^{2:a}}[\bx_i=1\mid  \bx_{i^*}=\bj]$. Note that 
	$\Ex_{\bj}[\bZ_i^a]=\Prx_{\bx \sim \xi|_1^{2:a}}[\bx_i=1]$  and therefore,
	\begin{align*}
		\Var_{\bj}[\bZ^a_{i}] =& \E_{\bj}\left[\left(\bZ^a_{i} - \E[\bZ^a_{i}]\right)^2\right] = \sum_{j \in\{0,1\}} \Prx_{\bx \sim \xi|_1^{2:a}}[\bx_{i^*}=j]\left(\Prx_{\bx \sim  \xi|_1^{2:a}}[\bx_i=1\mid  \bx_{i^*}=j]- \Prx_{\bx \sim \xi|_1^{2:a}}[\bx_i=1]\right)^2 \\
		\geq& \Prx_{\bx \sim \xi|_1^{2:a}}[\bx_{i^*}=1]\left(\Prx_{\bx \sim  \xi|_1^{2:a}}[\bx_i=1\mid  \bx_{i^*}=1]- \Prx_{\bx \sim \xi|_1^{2:a}}[\bx_i=1]\right)^2 \\
		\geq& \Prx_{\bx \sim \xi|_1^{2:a}}[\bx_{i^*}=1]^2\left(\Prx_{\bx \sim  \xi|_1^{2:a}}[\bx_i=1\mid  \bx_{i^*}=1]- \Prx_{\bx \sim \xi|_1^{2:a}}[\bx_i=1]\right)^2 \\
		=& \left(\Prx_{\bx \sim  \xi|_1^{2:a}}[\bx_i=1\wedge \bx_{i^*}=1]- \Prx_{\bx \sim \xi|_1^{2:a}}[\bx_i=1]\Prx_{\bx \sim \xi|_1^{2:a}}[\bx_{i^*}=1]\right)^2 \\
		=& \left(\Ex_{\bx \sim \xi|_1^{2:a}}[\bx_i\bx_{i^*}]-\Ex_{\bx \sim \xi|_1^{2:a}}[\bx_i]\Ex_{\bx \sim \xi|_1^{2:a}}[\bx_{i^*}]\right)^2\\
		=& \left(\Cov_{\bx \sim\xi|_1^{2:a}}(\bx_{i^*},\bx_i)\right)^2.
	\end{align*}
	Thus, by the definition of the index of a distribution,
	\begin{align*}
		\mathrm{Ind}((\xi|^{2:a})^{\{i^*\}})-\mathrm{Ind}(\xi|^{2:a})&=\Ex_{i\sim [n]}\left[\sum_{j \in\{0,1\}}\Prx_{\bx \sim  \xi|_1^{2:a}}[\bx_{i^*}=j]\Prx_{\bx \sim  \xi|_1^{2:a}}[\bx_i=1\mid \bx_{i^*}=j]^2-\Prx_{\bx \sim  \xi|_1^{2:a}}[\bx_i=1]^2\right]\\
		&=\Ex_{i\sim [n]}\left[\Ex_{\bj}[(\bZ_i^a)^2]-\Ex_{\bj}[\bZ_i^a]^2\right]=\Ex_{i\sim [n]}[\Var_{\bj}[\bZ_i^a]].
	\end{align*}
	
	Now we go back to constructing out sets. We use Lemma~\ref{lem:existindexvector} to obtain $A'$, noting also $b_a$ for every $a\in A'$. Then we use Lemma~\ref{lem:BigSetS_a} for every $a\in A'$ along with $b_a$ with parameter $\gamma/4$, noting the resulting set $S_a$, whose size is at least $\gamma n/8$. Next we use Lemma \ref{lem:largevariance} over this with parameter $\gamma/16$, also recalling that in the notation of that lemma $\delta_1\geq\gamma/4|B|$ and $\delta_2\geq\gamma/8$ (as guaranteed by Lemma \ref{lem:existindexvector}) to obtain the bound $\Var_{\bx\sim \xi|^{2:a}_1}\left[\sum_{i\in S_a}\bx_i\right]\ge \gamma^2\delta_1\delta_2|S_a|^2\geq \frac{\gamma^4}{32|B|}|S_a|^2$.
	
	Now we use Lemma~\ref{lem:largeCov}, with $\beta=\frac{\gamma^4}{32|B|}$ to obtain the set $S'_a$ and the sets $A_{i,a}$ for all $i\in S'_a$. Note that the obtained bounds here are $|S'_a|\geq\frac{\gamma^4}{32|B|}|S_a|\geq\frac{\gamma^5}{2^{10}|B|}n$ and $|A_{i,a}|\geq\frac{\gamma^5}{2^{10}|B|}n$, where for $i\in S'_a$ and $j\in A_{i,a}$ we have $\Cov_{\bx \sim \xi|^{2:a}_1}\left(\bx_i,\bx_{j}\right)\geq\frac{\gamma^4}{2^9|B|}$. Thus, for every $a\in A'$ and $i^*\in S'_a$ we have
	\begin{align*}
	\ix((\xi|^{2:a})^{\{i^*\}})-\ix(\xi|^{2:a})&=\frac{1}{n}\sum_{i\in [n]}\Var_{\bj}[\bZ^a_i]\ge \frac{1}{n}\sum_{i\in [n]}\left(\Cov_{\bx \sim \xi|^a_1}(\bx_{i^*},\bx_{i})\right)^2\\
	&\ge\frac{1}{n}\sum_{i\in A_{i^*,a}}\left(\frac{\gamma^4}{2^9|B|}\right)^2 =\frac{\gamma^{13}}{2^{28}|B|^3},
	\end{align*}
	completing our proof.
 \end{proof}

\begin{proofof}{Lemma~\ref{lem:AdditionalVar}}
	We start by invoking Lemma~\ref{lem:Indexincrease} to obtain the sets $A'$ and $S'_a$ for every $a\in A$ that are guaranteed by it. Now note that
	\begin{align*}
	\Ex_{\bi\sim [n]}\left[\ix(\xi^{\{\bi\}})-\ix(\xi)\right]&=\Ex_{\bi\sim [n]}\left[\Ex_{\ba\sim\xi|_2}\left[\ix(\xi^{\{\bi\}})-\ix(\xi)\right]\right]\\
	&=\Ex_{\ba\sim\xi|_2}\left[\Ex_{\bi\sim [n]}\left[\ix(\xi^{\{\bi\}})-\ix(\xi)\right]\right]\\
	&\ge\Prx_{\xi|_2}[A']\Ex_{\ba\sim\xi|^{2:A'}_2}\left[\frac{|S'_a|}{n}\Ex_{\bi\sim S'_{\ba}}\left[\ix(\xi^{\{\bi\}})-\ix(\xi)\right]\right]\\
	&\ge\frac{\gamma}{2}\Ex_{\ba\sim\xi|^{2:A'}_2}\left[\frac{\gamma^5}{2^{10}|B|}\Ex_{\bi\sim S'_{\ba}}\left[\frac{\gamma^{13}}{2^{28}|B^3|}\right]\right]=\frac{\gamma^{19}}{2^{39}|B|^4}.
	\end{align*}
	To conclude, we use the reverse Markov inequality to obtain that there exists a set of at least $\frac{\gamma^{19}}{2^{40}|B|^4}n$ indices $i$ for which $\ix(\xi^{\{i\}})-\ix(\xi)\geq\frac{\gamma^{19}}{2^{40}|B|^4}$ holds.
\end{proofof}

\section{Predicting the acceptance probability of a canonical tester}\label{sec:predict}

In this section we show that knowing the statistics of a good detailing is sufficient to approximate the acceptance probability of a property testing algorithm without actually running it.

\subsection{Simulation of a canonical tester}

We describe here a randomized process, $\Simulate$ (see Figure \ref{fig:simulation}), that given the statistic of a detailing $\xi$ of a distribution $\mu$ provides a simulated result of a canonical query pattern, without making any further queries to $\mu$ (eventually we will only use queries to find $\xi$ and its statistic).

\begin{figure}[ht!]
\begin{framed}
	\noindent Procedure $\Simulate(s,q,\eta,\Lambda)$
	\begin{flushleft}
		\noindent {\bf Input:} Integers $s,q\in \N$ corresponding to samples and queries, weight distribution $\eta$ and type distribution $\Lambda$ of a detailing $\xi$.\\
		{\bf Output:} A matrix $\bM\in \zo^{s\times q}$. 
\begin{enumerate}
	\item For each $j\in[q]$, draw $\bt_j\sim\Lambda$ independently.
	\item For each $i\in[s]$, draw $\ba_i\sim \eta$ independently.
	\item For every $(i,j)\in[s]\times[q]$, draw $\bM_{i,j}\sim \Ber(\bt_j(\ba_i))$ independently.
	\item {\bf Return} $\bM$
\end{enumerate}
	\end{flushleft}\vskip -0.14in
\end{framed}\vspace{-0.2cm}
\caption{Description of the $\Simulate$ procedure.} \label{fig:simulation}
\end{figure}

From now on we denote by $\calD_{\mathsf{sim}}(\eta,\Lambda)$ the distribution over the output of $\Simulate(s,q,\eta,\Lambda)$ (note that this output is a randomized $s\times q$ matrix). We now show that given a robust enough detailing of $\mu$, this simulated distribution is indeed close to the distribution $\calD_{\mathsf{test}}$ generated by an actual canonical $s\times q$ query pattern.

\begin{lemma} \label{lem:approx_Can}Fix $s,q\in \N$ and $\epsilon\in(0,1)$. If $\xi$ is $\left(\frac{\eps}{3(s+1)},q\right)$-good and $n\geq\frac{6q^2(s+1)}{\epsilon}$, then $\dtv(\calD_{\mathsf{sim}}(\eta,\Lambda),\calD_{\mathsf{test}})\le\eps$.
\end{lemma}

\begin{proof} Let $\epsilon'=\frac{\epsilon}{3\cdot(s+1)}$ and consider the following coupling $\bT^*$ (see Observation~\ref{obs:TV_is_EMD})
	 between $\calD_{\mathsf{sim}}(\eta,\Lambda)$ and ${\calD}_{\mathsf{test}}$:

We first draw $\bj_1,\ldots,\bj_q$ uniformly and independently from $[n]$ (with repetitions), and let $\bt_{1},\ldots,\bt_{q}$ be their respective types in $\xi$. Note that $\bt_{1},\ldots,\bt_{q}$ are distributed the same as independent draws from $\Lambda$. Then, we draw $\ba_1,\ldots,\ba_s$ independently from $\xi|_2$.  Also, we let $\bj'_1,\ldots,\bj'_q$ be equal to $\bj_1,\ldots,\bj_q$ if they contain no repetition, and otherwise let $(\bj'_1,\ldots,\bj'_q)$ be uniformly drawn \emph{without repetitions} from $[n]^q$ (independently of $\bj_1,\ldots,\bj_q$).
	
For every $i\in[s]$, let $\bT_i$  denote an optimal transfer distribution between $\prod_{j\in [q]}\Ber(\bt_j(\ba_i))$ and $\xi|^{2:\ba_i}_{\{\bj'_1,\ldots,\bj'_q\}}$. That is, $\Pr_{(\bx,\tilde \bx)\sim \bT_i}[\bx\neq\tilde \bx]=\dtv(\prod_{j\in [q]}\Ber(\bt_j(\ba_i)),\xi|^{2:\ba_i}_{\{\bj'_1,\ldots,\bj'_q\}})$. To obtain $\bM$ and $\widetilde \bM$, for every $i\in [s]$ independently, we draw $((\bM_{i,1},\ldots,\bM_{i,q}),(\widetilde \bM_{i,1},\ldots,\widetilde \bM_{i,q}))$ from $\bT_i$.

	In addition, we define the following random events:
\begin{itemize}
		\item $\calbE_0$ is the event that ($\bj_1,\ldots,\bj_q)= (\bj'_1,\ldots,\bj'_q)$. 
	\item $\calbE_1$ is the event that $\ba_1,\ldots,\ba_s\in J$ where $J$ is the good set as per Definition~\ref{def:eps_good_partition}.
	\item $\calbE_2$ is the event that the $q$-tuple $(\bj'_1,\ldots,\bj'_q)$ is  $\epsilon'$-independent with respect to all $\xi_{\ba_1},\ldots,\xi_{\ba_s}$.
	\item $\calbE_3$ is the event that $(\bM_{i,1},\ldots,\bM_{i,q})=(\widetilde{\bM}_{i,1},\ldots,\widetilde{\bM}_{i,q})$  for all $i\in[s]$.
\end{itemize}
We now note that
\begin{align*}
\dtv(\calD_{\mathsf{sim}},\calD_{\mathsf{test}})&\le \Pr_{\bT^*}[\bM\neq\widetilde \bM]=\Pr[\neg\calbE_3]\\
&\le \Pr[\neg\calbE_0]+\Pr[\neg\calbE_1]+\Pr[\neg\calbE_2\wedge\calbE_1]+\Pr[\neg \calbE_3\wedge\calbE_0\wedge\calbE_1\wedge\calbE_2]\\
&\le \Pr[\neg\calbE_0]+\Pr[\neg\calbE_1]+\Pr[\neg\calbE_2\mid\calbE_1]+\Pr[\neg \calbE_3\mid\calbE_0\wedge\calbE_1\wedge\calbE_2]
\end{align*}
By the fact that the detailing is $(\epsilon',q)$-good, using a union bound we get $\Prx[\neg\calbE_1]\le s
\cdot\epsilon'$. Similarly, the probability of a $q$-tuple $(\bj'_1,\ldots,\bj'_q)$ to be $\epsilon'$-independent with respect to $\xi|^{2:\ba_i}$ is at least $1-\epsilon'$, and thus the probability that the $q$-tuple is $\epsilon'$-independent with respect to all $\xi|^{2:\ba_1},\ldots,\xi|^{2:\ba_s}$ is greater than $1-s\cdot\epsilon'$ (equivalently $\Prx[\neg\calbE_2|\calbE_1]\le s\cdot \epsilon'$). To bound the probability of $\neg\calbE_0$, note that if $n\ge2q^2/\epsilon'$, then by the birthday paradox $\Prx[\neg\calbE_0]\le \frac{q(q-1)}{2n}\le \epsilon'$. It is left to bound $\Pr[\neg\calbE_3|\calbE_0\wedge\calbE_1\wedge\calbE_2]$.
\begin{align*}
\Pr[\neg\calbE_3|\calbE_0\wedge\calbE_1\wedge\calbE_2]&\le \sum_{i\in [s]}\Prx_{(\bM_{i,*},\widetilde{\bM}_{i,*})\sim\bT_i}\left[\bM_{i,*}\neq\widetilde{\bM}_{i,*}\mid \calbE_0\wedge\calbE_1\wedge\calbE_2\right]\\
&\le \sum_{i\in [s]} \dtv\left(\prod_{j\in [q]}\Ber(\bt_j(\ba_i)),\xi|^{2:\ba_i}_{\{\bj'_1,\ldots,\bj'_q\}}\right)\le s\cdot \epsilon',
\end{align*}
where the last inequality follows from the $q$-tuple being $\epsilon'$-independent with respect  to all $\xi_{\ba_i}$ with $i \in [s]$.
 Overall, we get 
\[\dtv(\calD_{\mathsf{sim}}(\eta,\Lambda),\calD_{\mathsf{test}})\le s\cdot\epsilon'+2s\cdot \eps'+\epsilon'=(3s+1)\epsilon'\le \epsilon.\]
\end{proof}

However, our eventual estimation algorithm will not be able to accurately find $\eta$ and $\Lambda$, but only approximations thereof. For this reason we will need ``continuity lemmas'' for these quantities, showing that the output distribution $\calD_{\mathsf{sim}}$ of $\Simulate$ will not degrade by much.

The next lemma handles the setting where we only have an approximation $\tilde{\Lambda}$ of the type distribution $\Lambda$, that is close enough to it in the $\eta$-weighted $\ell_1$ Earth Mover distance.

\begin{lemma} \label{lem:approx_types}
For $q,s\in \N$ and $\epsilon\in(0,1)$, if $\dem^\eta(\Lambda,\tilde{\Lambda})\le \frac{\eps}{ s\cdot q}$ then  $\dtv(\calD_{\mathsf{sim}}(\eta,\Lambda),{\calD}_{\mathsf{sim}}(\eta,\tilde{\Lambda}))\le \eps$. 
\end{lemma}

\begin{proof} By the premise of the lemma we have that $\dem^{\eta}(\Lambda,\tilde{\Lambda})\le \frac{\eps}{ s \cdot q}$, so let $\bT$ be a the transfer distribution over $([0,1]^{A})^2$ exhibiting the distance. We consider the following coupling $\bT^*$ between $\calD_{\mathsf{sim}}(\eta,\Lambda)$ and ${\calD}_{\mathsf{sim}}(\eta,\tilde{\Lambda})$. Sample $(\bt_1,\tilde \bt_1),\ldots ,(\bt_q,\tilde \bt_q)\sim \bT$ and $\ba_1,\ldots,\ba_s\sim \eta$, all independently. 
To set $\bM$ and $\widetilde \bM$, perform the following for every $i\in [s]$ and $j\in [q]$ independently:
\begin{enumerate}
    \item[(i)] If $\bt_j(\ba_i)\leq \widetilde \bt_j(\ba_i)$, then we do the following:
    \begin{enumerate}
        \item[(a)] with probability $\bt_j(\ba_i)$ set $(\bM_{i,j},\widetilde \bM_{i,j})=(1,1)$,
        
        \item[(b)] with probability $\widetilde \bt_j(\ba_i)-\bt_j(\ba_i)$ set $(\bM_{i,j},\widetilde \bM_{i,j})=(0,1)$,
        
        \item[(c)] with probability $1-\widetilde\bt_j(\ba_i)$ set  $(\bM_{i,j},\widetilde \bM_{i,j})=(0,0)$.
    \end{enumerate}

    \item[(ii)] If $\widetilde \bt_j(\ba_i)\leq  \bt_j(\ba_i)$, then we do the following:
    \begin{enumerate}
        \item[(a)] with probability $\widetilde \bt_j(\ba_i)$ set $(\bM_{i,j},\widetilde \bM_{i,j})=(1,1)$,
        
        \item[(b)] with probability $\bt_j(\ba_i)-\widetilde \bt_j(\ba_i)$ set $(\bM_{i,j},\widetilde \bM_{i,j})=(1,0)$,
        
        \item[(c)]with probability $1-\bt_j(\ba_i)$ set  $(\bM_{i,j},\widetilde \bM_{i,j})=(0,0)$.
    \end{enumerate}
\end{enumerate}

Using the above coupling, we have
\begin{align*}
\dtv(\calD_{\mathsf{sim}}(\eta,\Lambda),{\calD}_{\mathsf{sim}}(\eta,\tilde{\Lambda}))&\le \Prx_{(\bM,\widetilde \bM)\sim \bT^*}[\bM\neq \widetilde \bM]\leq \sum_{(i,j)\in [s]\times [q]}\Prx_{(\bM,\widetilde \bM)\sim \bT^*}[\bM_{i,j}\neq \widetilde \bM_{i,j}]\\
&\le \sum_{(i,j)\in[s]\times [q]}\Ex_{(\bt_j,\tilde \bt_j)\sim \bT}\Ex_{\ba_i\sim \eta}[|\bt_j(\ba_i)-\tilde{\bt}_j(\ba_i)|]\\
&=\sum_{(i,j)\in[s]\times [q]}\dem^\eta(\Lambda,\tilde{\Lambda})\le \epsilon.
\end{align*}
\end{proof}

Next, we handle the setting where we only have a distribution $\tilde\eta$ that is close to $\eta$ in the variation distance.

\begin{lemma}\label{lem:approx_eta}For any $q,s\in \N$, $\epsilon\in(0,1)$, 
	if $\dtv(\eta,\tilde{\eta})\le \frac{\eps}{s}$, then $\dtv(\calD_{\mathsf{sim}}(\eta,\Lambda),\calD_{\mathsf{sim}}(\tilde{\eta},\Lambda))\le \eps$.
\end{lemma}

\begin{proof} Let $\bT$ be an optimal coupling (transfer distribution) between $\eta$ and $\tilde\eta$. Namely, $\dtv(\eta,\tilde{\eta})=\Pr_{(\ba,\tilde{\ba})\sim\bT}[\ba\neq\tilde{\ba}]$. We construct a coupling $\bT^*$ between $\calD_{\mathsf{sim}}(\eta,\Lambda)$ and $\calD_{\mathsf{sim}}(\tilde{\eta},\Lambda)$. First, draw $(\ba_1,\tilde{\ba}_1),\ldots,(\ba_s,\tilde{\ba}_s)\sim \bT$ and $\bt_1,\ldots,\bt_{q}\sim \Lambda$, all independently. We define $\bM$ and $\widetilde{\bM}$ as follows: For every $(i,j)\in[s]\times[q]$, if $\ba_i=\tilde{\ba}_i$, then draw a bit $\bb\sim\Ber\left(\bt_j(\ba_i)\right)$ and set $\bM_{i,j}=\widetilde{\bM}_{i,j}=\bb$. If $\ba_i\neq\tilde{\ba}_i$, then independently set $\bM_{i,j}\sim\Ber(\bt_j(\ba_i))$ and $\widetilde{\bM}_{i,j}\sim\Ber(\bt_j(\tilde\ba_i))$.
Then,
\begin{align*}
\dtv(\calD_{\mathsf{sim}}(\eta,\Lambda),\calD_{\mathsf{sim}}(\tilde{\eta},\Lambda))\le \Prx_{(\bM,\widetilde \bM)\sim \bT^*}[\bM\neq\widetilde{\bM}]\le \sum_{i\in [s]}\Pr_{(\ba_i,\tilde{\ba}_i)\sim\bT}[\ba_i\neq\tilde{\ba}_i]\le \epsilon.
\end{align*}
\end{proof}

The following lemma bundles together the approximation results in this section.

\begin{lemma} \label{thm:Simulation}Fix $s,q\in \N$ and $\epsilon\in(0,1)$ and suppose that $\dem^{\eta}(\Lambda,\tilde\Lambda)\le \eps/3sq$ and $\dtv(\eta,\tilde{\eta})\le \epsilon/3s$.  If $n\ge 18q^2(s+1)/\eps$ and $\xi$ is $\left(\frac{\epsilon}{9(s+1)},q\right)$-good then $\dtv(\calD_{\mathsf{sim}}(\tilde\eta,\tilde{\Lambda}),\calD_{\mathsf{test}})\le\eps$.
\end{lemma}

\begin{proof} First note that by the triangle inequality
\begin{align*}
	\dtv\left(\calD_{\mathsf{sim}}(\tilde{\eta},\tilde{\Lambda}),\calD_{\mathsf{test}}\right)&\le \dtv\left(\calD_{\mathsf{sim}}(\tilde{\eta},\tilde{\Lambda}),\calD_{\mathsf{sim}}(\eta,\tilde{\Lambda})\right)+\dtv\left(\calD_{\mathsf{sim}}({\eta},\tilde{\Lambda}),\calD_{\mathsf{sim}}(\eta,\Lambda)\right)\\
	&\qquad\qquad + \dtv\left(\calD_{\mathsf{sim}}(\eta,\Lambda),\calD_{\mathsf{test}}\right).
\end{align*}
Applying Lemma~\ref{lem:approx_eta}, Lemma~\ref{lem:approx_types} and Lemma~\ref{lem:approx_Can}, all with $\epsilon/3$ instead of $\epsilon$, gives that every term in the above sum is bounded by $\epsilon/3$. 
\end{proof}

\subsection{Acceptance probability computation}

With Lemma~\ref{thm:Simulation} in mind, we construct an ``acceptance predictor'' procedure, \textsf{Accept-Probability} (see Figure \ref{fig:AcceptProb}), which is intended to \emph{calculate} the estimated acceptance probability of a test. It works by calculating (instead of actually running) the output distribution of $\Simulate$. Recall that we denote the probability that the canonical tester with proximity parameter $\epsilon$ accepts $\mu$ with $\mathsf{acc}_{\epsilon}(\mu)$.

\begin{figure}[ht!]
	\begin{framed}
		\noindent Procedure $\textsf{Accept-Probability}(s,q,\eta,\Lambda,{\epsilon'})$
		\begin{flushleft}
			\noindent {\bf Input:} Integers $s,q\in \N$ corresponding to samples and queries, weight distribution $\eta$, type distribution $\Lambda$ of a detailing $\xi$ of $\mu$ with respect to $A$ and test proximity parameter $\epsilon'\in(0,1)$.\\
			{\bf Output:} An estimated acceptance probability $\widetilde{\textsf{acc}}_{\epsilon'}(\eta,\Lambda)\in[0,1]$. 
			\begin{enumerate}
				\item Let $\alpha_{\eps'}:\zo^{s\times q}\to [0,1]$ be the acceptance probability function of the canonical tester with proximity parameter $\epsilon'$, as per Definition~\ref{def:cantest}.
				\item For each $M\in\zo^{s\times q}$, $a\in A^s$ and $t\in ([0,1]^A)^q$ in the support of $\Lambda$ compute \label{tilde-D}
				\begin{align*}
				\widetilde\calD(M,a,t)=\prod_{i\in [s]}\eta(a_i)\cdot \prod_{j\in [q]}\Lambda(t_j)\cdot \prod_{i,j\in [s]\times [q]}(\indi_{M_{i,j}=1}\cdot t_j(a_i)+\indi_{M_{i,j}=0}\cdot (1-t_j(a_i))
				\end{align*}
				\item Set \begin{align*}
				    \widetilde{\textsf{acc}}_{\epsilon'}(\eta,\Lambda)&=\sum_{(M,a,t) \in \zo^{s\times q}\times A^s\times [0,1]^{A\times q}}\widetilde\calD(M,a,t)\cdot \alpha_{\eps'}(M)
				    =\Ex_{\bM\sim\widetilde\calD|_1}[\alpha_{\eps'}(\bM)]
				\end{align*}
				\item \textbf{Return} $\widetilde{\textsf{acc}}_{\epsilon'}(\eta,\Lambda)$.
			\end{enumerate}
		\end{flushleft}\vskip -0.14in
	\end{framed}\vspace{-0.2cm}
	\caption{Description of the \textsf{Accept-Probability} procedure.} \label{fig:AcceptProb}
\end{figure}

\begin{lemma}\label{cor:approxprob}
Fix $\epsilon'\in(0,1)$ and let $\calP$ be an index-invariant property that admits a canonical tester with proximity parameter $\eps'$ using sample complexity $s=s(\eps')$ and query complexity $q=q(\eps')$. Fix $\epsilon\in(0,1)$, and let
$(\eta, \Lambda)$ be the parameters of an $\left(\frac{\epsilon}{9(s+1)},q\right)$-good detailing of $\mu $ with respect to $U\subseteq[n]$. If $n\ge 18q^2(s+1)/\eps$, then given $(\widetilde{\eta}, \widetilde{\Lambda})$ such that $\dtv(\eta, \widetilde{\eta}) \leq \frac{\eps}{3s}$ and $\dem^{\eta}(\Lambda, \widetilde{\Lambda}) \leq \frac{\eps}{3s q}$, the subroutine $\textsf{Accept-Probability}(s,q, \widetilde{\eta}, \widetilde{\Lambda},\epsilon')$ reports $\widetilde{\mathsf{acc}}_{\eps'}(\mu)$ such that 
$$|\widetilde{\mathsf{acc}}_{\eps'}(\eta,\Lambda)-{\mathsf{acc}}_{\eps'}(\mu)|\leq\eps.$$
\end{lemma}

\begin{proof}
	By using Lemma~\ref{thm:Simulation}, we have $\dtv(\calD_{\mathsf{sim}}(\tilde{\eta},\tilde{\Lambda}),\calD_{\mathsf{test}})\le \eps$. We next analyze the calculated quantity $\widetilde{\textsf{acc}}_{\epsilon'}(\eta,\Lambda)$. Let $\widetilde\calD$ denote the distribution over $\{0,1\}^{s \times q}\times A^s \times \left([0,1]^A\right)^q$ such that $\widetilde\calD(M,a,t)$ is the probability that $\Simulate$  draws $(\bt_1,\ldots,\bt_q)=t$, $(\ba_1,\ldots,\ba_s)=a$, and $\bM=M$. Therefore, $\widetilde\calD|_1=\calD_{\mathsf{sim}}(\tilde{\eta},\tilde{\Lambda})$ and hence $\widetilde{\textsf{acc}}_{\epsilon'}(\eta,\Lambda)=\Ex_{\bM\sim\calD_{\mathsf{sim}}(\tilde{\eta},\tilde{\Lambda})}[\alpha_{\epsilon'}(\bM)]$.
	
	This all means that
	\[|\widetilde{\textsf{acc}}_{\epsilon'}(\eta,\Lambda)-{\textsf{acc}}_{\epsilon'}(\mu)|=\left|\Ex_{\bM\sim\calD_{\mathsf{sim}}(\tilde{\eta},\tilde{\Lambda})}[\alpha_{\epsilon'}(\bM)]-\Ex_{\bM\sim\calD_{\mathsf{test}}}[\alpha_{\epsilon'}(\bM)]\right|\le \dtv(\calD_{\mathsf{sim}}(\tilde{\eta},\tilde{\Lambda}),\calD_{\mathsf{test}})\le \eps\]
	concluding the proof.
\end{proof}

\section{Finding a weakly robust detailing and estimating its parameters}\label{sec:find}

This section is devoted to finding a variable set defining a (weakly) robust partition, and then to estimating its weight and type distributions. The final estimation algorithm will make all its queries deploying the algorithms developed here.

\subsection{Estimating the index of a detailing}

In this section, we describe and analyze an algorithm for estimating the index of a detailing (Figure~\ref{algo:estind}).

\begin{figure}[ht!]
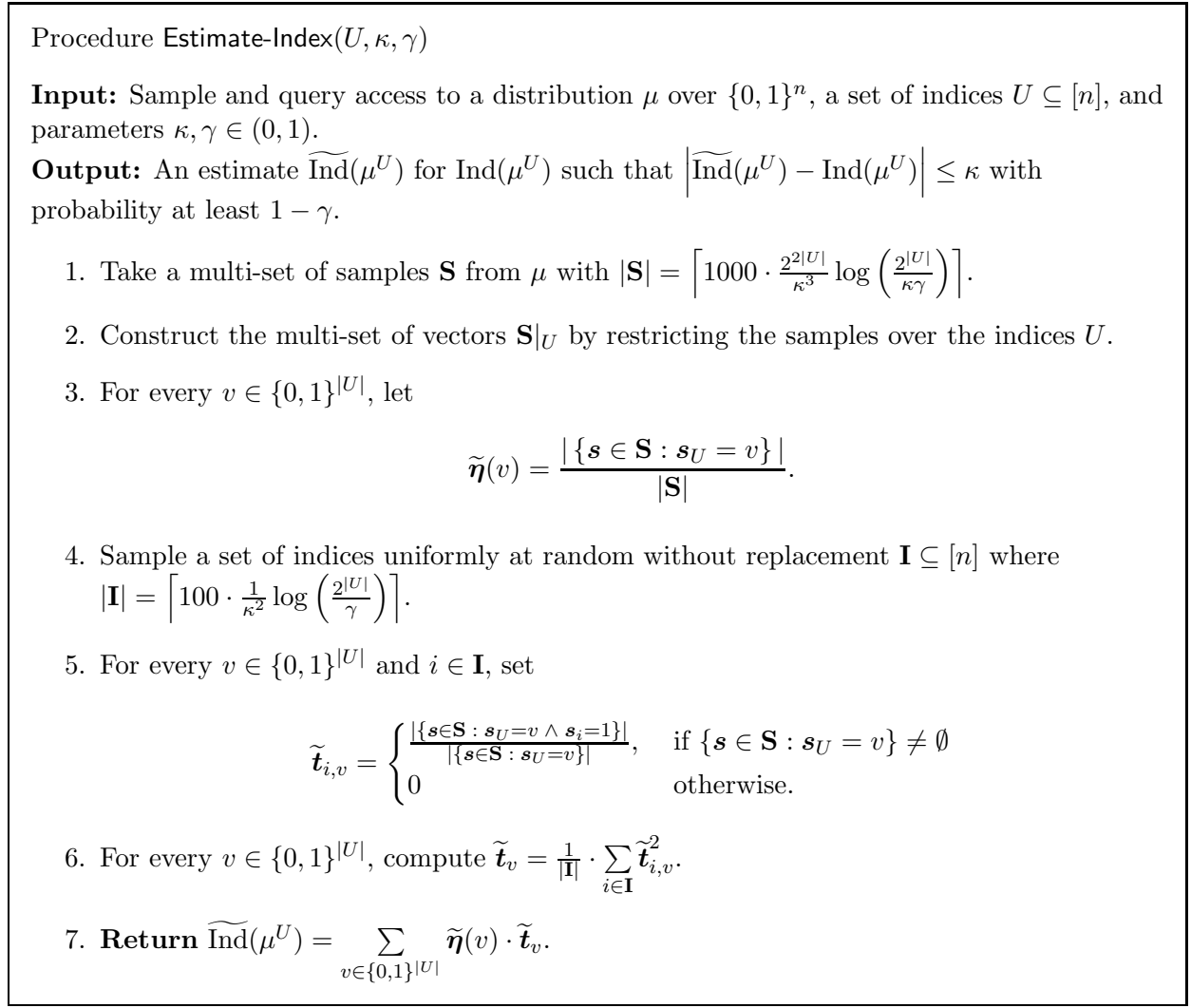

	\begin{framed}
		\noindent Procedure $\textsf{Estimate-Index}(U,\kappa,\gamma)$
			\begin{flushleft}
			\noindent {\bf Input:} Sample and query access to a distribution $\mu$ over $\{0,1\}^n$, a set of indices $U \subseteq [n]$, and parameters $\kappa, \gamma \in (0,1)$.\\
			{\bf Output:} An estimate $\widetilde{\ix}(\mu^U)$ for $\ix(\mu^U)$ such that $\left|{\widetilde{\ix}(\mu^U) -\ix(\mu^U)}\right|\leq \kappa$ with probability at least $1 -\gamma$. 
			\begin{enumerate}
 
				\item Take a multi-set of samples $\bS$ from $\mu$ with $|{\bS}|=\left\lceil 1000\cdot\frac{2^{2|{U}|}}{\kappa^3} \log \left(\frac{2^{|U|}}{\kappa\gamma}\right)\right\rceil$.\label{step:drawsamples}

				\item Construct the multi-set of vectors $\bS|_{U}$ by restricting the samples over the indices $U$.
				
				\item For every $v\in \zo^{|U|}$, let $$\widetilde{\boldeta}(v)=\frac{|\left\{{\bs} \in \bS : {\bs}_{ {U}}= v\right\}|}{|\bS|}.$$ 

				\item  Sample a set of indices uniformly at random without replacement $\bI \subseteq [n]$ where $|{\bI}| = \left\lceil100\cdot\frac{ 1}{\kappa^2} \log \left(\frac{2^{|U|}}{\gamma}\right)\right\rceil$.

				\item For every $v \in \{0,1\}^{|U|}$ and $i \in \bI$, set 
				
				\[\widetilde{\bt}_{i,v}=\begin{cases}
			\frac{|\{{\bs} \in \bS\; :\; {\bs}_{ U}= v
				\;\wedge\; \bs_i=1\}|}{|\{{\bs} \in \bS \;: \;{\bs}_{ U}= v\}|},\; &\text{if } \{{\bs} \in \bS : {\bs}_{ U}= v\}\neq \emptyset\\
				0\;&\text{otherwise}.
				\end{cases}				
 \]

				\item For every $v \in \{0,1\}^{|U|}$, compute $\widetilde{\bt}_v=\frac{1}{|{\bI}|} \cdot \sum\limits_{i \in \bI}\widetilde{\bt}_{i,v}^2$.
				
				\item {\bf Return} $\widetilde{\ix}(\mu^U)=\sum\limits_{v \in \{0,1\}^{|U|}} \widetilde{\boldeta}(v) \cdot \widetilde{\bt}_v$.
			\end{enumerate}
		\end{flushleft}\vskip -0.14in
	\end{framed}\vspace{-0.2cm}
	\caption{Description of the $\textsf{Estimate-Index}$ procedure.} \label{algo:estind}
\end{figure}

We will prove that with high probability, the estimate of the index of $\mu^U$ returned by the algorithm \textsf{Estimate-Index} is indeed close to $\ix(\mu^U)$.

\begin{lemma}\label{lem:estindex}
	Consider \textsf{Estimate-Index}$(U,\kappa,\gamma)$ as described in Figure~\ref{algo:estind}. Given a subset of indices $U \subseteq [n]$ and parameters $\kappa,\gamma\in(0,1)$ as input, the procedure makes at most $O\left(\frac{2^{2|U|}}{\kappa^5}\log^2\left(\frac{2^{|U|}}{\kappa\gamma}\right)\right)$ queries, and outputs $\widetilde{\ix}(\mu^U)$ such that $|{  \widetilde{\ix}(\mu^U)- \ix(\mu^U)}| \leq {\kappa}$ holds with probability at least $1 - \gamma$.
\end{lemma}

Fix  $v \in \{0,1\}^{|U|}$ and let $\eta(v)\eqdef \mu^U|_2(v)=\Prx_{{\bx} \sim \mu}[\bx _U=v]$. In addition, for $i \in [n]$ let 
\[t_{i,v}=\begin{cases}\Prx_{{\bw}\sim\mu^U|_1^{2:v}}[{\bw}_i=1], &\text{if }\mu^U|_2(v)\neq 0\\
	0, &\text{otherwise}
\end{cases}.
\]
For $v \in \{0,1\}^{|U|}$, let $t_v=\frac{1}{n}\sum\limits_{i=1}^n t_{i,v}^2$. Observe that the index of the detailing $\mu^U$ can be expressed as

\begin{equation*}\label{def:ind-alt}
	\ix(\mu^U)=\sum\limits_{v \in \{0,1\}^{|{U}|}}\eta(v) \cdot t_v.
\end{equation*}

Let $r= 2^{|{U}|}$, and define 
$$J \eqdef \{v \in \{0,1\}^{|{U}|}: \eta(v)\ge\kappa/5r\}.$$

\begin{definition}[Definition of the event $ \calE^*$] The event $\calE^*$ is defined as the intersection of the following two events.
	\begin{description}
		\item[$\calE_1$] (Approximating $\eta(v)$'s): For every $v \in \{0,1\}^{|U|}$, if $v \in J$ then $|{\widetilde{\boldeta}(v) - \eta(v)}|\leq \frac{\kappa}{10r}$, and if $v\notin J$ then  $\widetilde{\boldeta}(v)\leq  \frac{3\kappa}{10 r}$. 
		\item[$\calE_2$] (Approximating $t_v$'s): For every $v \in J$, $|{\widetilde{\bt}_v-t_v  }|\leq \frac{6\kappa}{10}$.
	\end{description}
\end{definition}

\begin{lemma}\label{lem:eventgood} Event $\calE^*$ holds with probability at least $1-\gamma$.
\end{lemma}

For the proof of the above lemma we need the following two lemmas (in fact the next lemma is also used \emph{inside} the proof of the following one).

\begin{lemma}[Approximating $\eta(v)$'s]\label{cl:estimateeta}
	Consider $v \in \{0,1\}^{|U|}$.  The following hold with probability at least $1- \frac{\kappa\gamma}{30r}$:
	\begin{enumerate}
		\item[(i)] If $\eta(v) \geq \frac{\kappa}{5 r}$ then $|{\widetilde{\boldeta}(v) - \eta(v)}| \leq \frac{\kappa}{10r}$.
		\item[(ii)] If $\eta(v) < \frac{\kappa}{5 r}$ then  $\widetilde{\boldeta}(v)\leq  \frac{3\kappa}{10 r}$.
	\end{enumerate}
\end{lemma}

\begin{proof} 
 For $\bs\in \bS$, we let $\bchi_{s}$ be the indicator random variable for $\bs_U=v$ and note that $\Ex[\bchi_{s}]=\eta(v)$. Note that $\widetilde{\boldeta}(v)=\sum_{\bs \in \bS}\bchi_{\bs}$. Applying an additive Chernoff bound, we have that 
	\[\Prx \left[|\widetilde\boldeta(v)-\eta(v)|\ge \frac{\kappa}{10r}\right]\le 2\exp\left(\frac{-2\kappa^2}{100r^2}\cdot |\bS|\right)=2\exp\left(\frac{-2\kappa^2}{100r^2}\cdot \frac{{1000}r^2}{\kappa^3}\log \left(\frac{r}{\kappa\gamma}\right)\right)<\frac{\kappa\gamma}{30r}.\]
	Note that this covers both cases $v\in J$ and $v\notin J$.
\end{proof}

\begin{lemma}\label{cl:aplhavigap} 
	Consider $v \in J$. Then $|{\widetilde{\bt}_{v}- t_{v} }|\leq \frac{6\kappa}{10}$ holds with probability at least $1- \frac{2\gamma}{3r}$.
\end{lemma}

Let us define $\bt'_v$ as $\frac{1}{n}\sum\limits_{i \in [n]}\widetilde{\bt}^2_{i,v}$, where $v \in \{0,1\}^U$. We prove Lemma~\ref{cl:aplhavigap} with the help of the two following lemmas.

\begin{lemma}\label{lem:lemlem1}
    For any $v \in \{0,1\}^U$, $|\widetilde{\bt}_v-\bt_v'|\leq \frac{\kappa}{10}$ holds with probability at least $1-\gamma/3r$.
\end{lemma}

\begin{proof}
    	Note that by taking expectation over the set of indices $\bI$, we have  $\Ex_{\bI}\left[\widetilde{\bt}_v\right]=\frac{1}{n}\sum_{i=1}^{n}\widetilde{\bt}^2_{i,v}=\bt'_v$. Applying Hoeffding's inequality for sampling without replacement (Lemma~\ref{lem:hoeffdingineq_without_replacement}), we have 
 \begin{eqnarray*}
  \Prx \left[|\widetilde{\bt}_v-\bt_v'|>\frac{\kappa}{10}\right]&=& \Prx\left[\left|\frac{1}{|\bI|}\sum_{i\in\bI}\widetilde{\bt}_{i,v}^2-\frac{1}{n}\sum_{i=1}^n\widetilde{\bt}_{i,v}^2\right|>\frac{\kappa}{10}\right]\\
     &\le& 2\exp\left(-2\cdot\frac{\kappa^2}{100}\cdot |\bI|\right) \\
     &\leq& 2\exp\left(-2\cdot\frac{\kappa^2}{100}\cdot \frac{100}{\kappa^2} \log \left(\frac{r}{\gamma}\right)\right)\\
     &<&\gamma/3r.
 \end{eqnarray*}
\end{proof}

\begin{lemma}\label{lem:bound-|t_'-t_v|}
    For any $v \in J$, $|\bt_v'-t_v|\leq \frac{5\kappa}{10}$ holds with probability at least $1-\gamma/3r$.
\end{lemma}

To prove the above lemma, we need the following definition of the notion of a variable being bad or good with respect to $v \in \{0,1\}^U$, and a claim about the probability of a variable being bad.
\begin{definition}
A variable $i \in [n]$ is said to be \emph{bad} for $v \in \{0,1\}^{U}$ if $|\widetilde{\bt}_{i,v}^2- t_{i,v}^2|> \frac{\kappa}{5r}$. Otherwise, $i$ is said to be \emph{good} for $v$.  
\end{definition}
\begin{lemma} \label{lem:t_i_v-apx}
For any fixed $i \in [n]$ and $v \in J$, $i$ is bad for $v$ with probability at most $\kappa\gamma/15r$.%
\end{lemma}
\begin{proof} We first argue that 
	$$\Prx \left[|{\widetilde{\bt}_{i,v}-t_{i,v}}|\leq \frac{\kappa}{10}  \right]\geq 1 - \frac{\kappa\gamma}{15 r}.$$
	This will imply that
	$$\Prx \left[|{\widetilde{\bt}^2_{i,v}-t^2_{i,v}}|\leq \frac{\kappa}{5}  \right]\geq 1 - \frac{\kappa\gamma}{15 r},$$
	 due to the fact that $|{a^2-b^2}|\leq 2 \cdot |{a-b}|$ when $a,b \in (0,1)$. 
	
	Let $\bS=\{\bs^j\}_{j\in[|\bS|]}$ be the multi-set of samples drawn in Step (\ref{step:drawsamples}) of the procedure. For every $j\in[|\bS|]$, consider an  indicator random variable $\bchi_{j,v,i}$ such that $\bchi_{j,v,i}=1$ if and only if ${\bs^j}_{U}= v$ and $\bs^j_i=1$. Let $\bchi_{v,i}=\sum_{j \in [|\bS|]} \bchi_{j,v,i}$.
	From the definition of $\widetilde{\bt}_{i,v}$ from Figure~\ref{algo:estind} and the definition of $\bchi_{v,i}$, we have
	$$
	\widetilde{\bt}_{i,v} = \frac{|\{{\bs} \in \bS : {\bs}_{ U}= v
		\;\wedge\; \bs_i=1\}|}{|\{{\bs} \in \bS : {\bs}_{ U}= v\}|}=\frac{\bchi_{v,i}}{|\{{\bs} \in \bS : {\bs}|_{ U}= v\}|}.
	$$
	From the definition of $\widetilde{\boldeta}(v)$ in Figure~\ref{algo:estind}, this means that
	\begin{equation}\label{eqn:approxlphav}
		\widetilde{\bt}_{i,v} =\frac{\bchi_{v,i}}{\widetilde{\boldeta}(v) \cdot |{\bS}|}.
	\end{equation}
	In addition, note that for any $z\in \N$:
	\begin{eqnarray*}
		\Ex\left[\widetilde{\bt}_{i,v}\mid \widetilde{\boldeta}(v)\cdot |\bS|=z  \right]=\E\left[\frac{\bchi_{v,i}}{\widetilde{\boldeta}(v)\cdot |\bS| }\;\bigg|\; \widetilde{\boldeta}(v)\cdot |\bS| =z\right]&=& \frac{z\cdot \Prx_{{\bw}\sim\mu^U|_1^{\bw_U={\bv}}}[\bw_i=1]}{z}\\
		&=& \Prx_{{\bw}\sim\mu^U|_1^{\bw_U={\bv}}}[\bw_i=1]=t_{i,v}.
	\end{eqnarray*}

For $\bs \in \bS$, let $\bY_{\bs,v,i}$ be the random variable defined as $\bY_{\bs,v,i}=\frac{\bchi_{\bs,v,i}}{\widetilde{\boldeta}(v) \cdot |\bS|}$. Let $\bY_{v,i}=\sum_{\bs \in \bS} \bY_{s,v,i}$. Since $\bchi_{\bs,v,i} \in \{0,1\}$, $\bY_{\bs,v,i} \in \left\{0,\frac{1}{\widetilde{\boldeta}(v) \cdot |\bS|}\right\}$. We would like to apply the Hoeffding bound (Lemma~\ref{lem:hoeffdingineq}) to bound $\bY_{v,i}$. But $\widetilde{\boldeta}(v) \cdot |\bS|$ is a random variable, so instead we apply the Hoeffding bound conditioned on $\widetilde{\boldeta}(v) \cdot |\bS|=z$, where $z \in \N$. Note that, when $\widetilde{\boldeta}(v) \cdot |\bS|=z$, the number of $\bs$ such that $\bY_{\bs,v,i}\neq 0$ is at most $z$. Hence,
\begin{eqnarray*}
\Prx\left[|\widetilde{\bt}_{i,v}- t_{i,v}|> \frac{\kappa}{10} \;\bigg|\;  \widetilde{\boldeta}(v) \cdot |\bS|=z\right] &\leq& 2\exp{\left(-\frac{2(\kappa/10)^2}{z \cdot (1/z)^2}\right)} =
2\exp{\left(-\frac{2\kappa^2 z}{100}\right)}.
\end{eqnarray*}   

As $v \in J$, it holds that $\eta(v) \geq \frac{\kappa}{5r}$. From Lemma~\ref{cl:estimateeta}, with probability at least $1- \kappa\gamma/30r$, it holds that $\abs{\widetilde{\boldeta}(v)-\eta(v)} \leq \kappa/10r$, and in particular  $\widetilde{\boldeta}(v)\cdot|\bS| \geq  \frac{\kappa}{10r} \cdot |\bS|$.

Hence, 
\begin{eqnarray*}
&& \Prx\left[|\widetilde{\bt}_{i,v}- t_{i,v}|> \frac{\kappa}{10} \right]\\
&\le& \Prx\left[\widetilde{\boldeta}(v)|\bS| < \frac{\kappa}{10r} \cdot |\bS|\right] +  \Prx\left[\widetilde{\boldeta}(v)|\bS| \geq  \frac{\kappa}{10} \cdot |\bS|\right] \cdot \Prx\left[|\widetilde{\bt}_{i,v}- t_{i,v}|> \frac{\kappa}{10} \;\bigg|\; \widetilde{\boldeta}(v)|\bS| \geq \frac{\kappa}{10r} \cdot |\bS|\right] \\
&\leq&
\frac{\kappa\gamma}{30r} + 1 \cdot 2\exp{\left(-2 \cdot \frac{ \kappa^2 \cdot (\kappa/10r) \cdot |\bS| }{100 }\right)} \\
&\leq&  \frac{\kappa\gamma}{30r} +  2\exp{\left(-2 \cdot \frac{ \kappa^3 }{1000r}  \cdot \frac{1000r^2}{\kappa^3}\log \left(\frac{r}{\kappa \gamma}\right)\right)} \leq \frac{ \kappa\gamma }{15 r}.
\end{eqnarray*} 
 \end{proof}

\color{black}

  \begin{proofof}{Lemma~\ref{lem:bound-|t_'-t_v|}}
      For $i \in [n]$, let $\bX_i$ denote the random variable defined as 
      \[\bX_i = \begin{cases}
        1 &  i~\mbox{is bad for}~v \\
        0 & otherwise
\end{cases}
\]

Let $\bX=\frac{1}{n}  \sum\limits_{ i \in [n]} \bX_i$, which denotes the fraction of variables  that are bad for $v$. By Lemma~\ref{lem:t_i_v-apx}, $\E[\bX_i] \leq \frac{\kappa \gamma}{15r}$. Hence, $\E[\bX]\leq \frac{\kappa \gamma}{15r}$. Using Markov Inequality, $\Prx \left[ \bX \geq  \kappa/5\right] \leq \frac{ \gamma}{3r}$. So, with probability $1-\gamma/3r$, there are at most $\kappa n/5$ bad variables for $v$.

Note that,  $$\abs{\bt_v'-t_v}=\abs{\frac{1}{n}\sum_{i=1}^{n}\widetilde{\bt}^2_{i,v} - \frac{1}{n}\sum_{i=1}^{n}{t}^2_{i,v}} \leq \frac{1}{n}{\sum_{i=1}^{n}\abs{\widetilde{\bt}^2_{i,v} - {t}^2_{i,v}}}.$$
 For every variable $i \in[n]$ that is not bad for $v$, $\abs{\widetilde{\bt}^2_{i,v} - {t}^2_{i,v}} \leq \frac{\kappa}{5}$. For every bad variable $i \in [n]$, $\abs{\widetilde{\bt}^2_{i,v} - {t}^2_{i,v}}$ can be trivially upper bounded by $1$. As there are there are at most $\kappa n/5$ bad variables for $v$ with probability $1-\gamma/3r$, we have $\abs{\bt_v'-t_v} \leq 2\kappa/5< 6\kappa/10$ with the same probability.
\end{proofof}

\begin{proofof}{Lemma~\ref{cl:aplhavigap}} Considering Lemma~\ref{lem:lemlem1} and Lemma~\ref{lem:bound-|t_'-t_v|}, a use of the triangle inequality and a union bound complete the proof.\end{proofof}

\begin{proofof}{Lemma~\ref{lem:eventgood}}
    Event $\calE_1$ happens with probability at least $1-\frac{\gamma}{3}$ by a union bound of Lemma~\ref{cl:estimateeta} over all $v\in\zo^{|U|}$. Event $\calE_2$ happens with probability at least $1-\frac{2\gamma}{3}$ by a union bound of Lemma~\ref{cl:aplhavigap} over all $v\in J$. A final union bound implies that event $\calE^*$ occurs with probability at least $1-\gamma$, as required.
\end{proofof}

Now we are ready to prove Lemma~\ref{lem:estindex}.

\begin{proofof}{Lemma~\ref{lem:estindex}}
	Conditioned on $\calE^*$ (which by Lemma~\ref{lem:eventgood} holds with probability at least $1-\gamma$) we have that, for every $v \in J$, $|\widetilde{\boldeta}(v)- \eta(v)|\le\frac{\kappa}{10r}$ and $|\widetilde{\bt}_v- t_{v}|\le \frac{6\kappa}{10}$ hold. Hence,
	\begin{eqnarray*}	\sum\limits_{v \in J} |{\widetilde{\boldeta}(v) \widetilde{\bt}_v - \eta(v) t_v}| 
 &\leq &
 \sum\limits_{v \in J} \widetilde{\boldeta}(v)\abs{ \widetilde{\bt}_v - t_v} + \sum\limits_{v \in J} t_v \abs{ \widetilde{\boldeta}(v) - \eta(v)}\\
  &\leq& \sum\limits_{v \in J} \widetilde{\boldeta}(v) \cdot \frac{6\kappa}{10} + \sum\limits_{v \in J} t_v \cdot \frac{\kappa}{10r}
	\end{eqnarray*}
 Note that $\sum\limits_{v \in J} \widetilde{\boldeta}(v) \leq 1$ and $t_v \leq 1$ for each $v \in J$. So, 
 \begin{equation}\label{eqn:etavalphavbound1}	
 	\sum\limits_{v \in  J}|{\widetilde{\boldeta}(v) \widetilde{\bt}_v-\eta(v) t_v}|\leq \frac{7\kappa}{10}.
	\end{equation}
	
	For $v\notin J$, we have that $$|\widetilde{\boldeta}(v)\widetilde{\bt}_v-\eta_{v}t_{v}|\le \max\{\widetilde{\boldeta}(v)\widetilde{\bt}_v,\eta(v)t_v\} \leq  \max\{\widetilde{\boldeta}(v),\eta(v)\} \leq \frac{3\kappa}{10r}.$$
 Therefore,
	\begin{equation}\label{eqn:boundoutsidej}
		\sum\limits_{v \in \{0,1\}^{|U|}\setminus J}|{\widetilde{\boldeta}(v) \widetilde{\bt}_v-\eta(v) t_v}|\leq \frac{3\kappa}{10}.
	\end{equation}
	
	From  Equation~(\ref{eqn:etavalphavbound1}) and Equation~(\ref{eqn:boundoutsidej}) we have:
	\begin{eqnarray*}
		\left|{\sum\limits_{v \in \{0,1\}^{|{U}|}} \widetilde{\boldeta}(v) \widetilde{\bt}_v- \ix(\xi^U)}\right| &\leq& \left|{\sum\limits_{v \in J} \widetilde{\boldeta}(v) \cdot \widetilde{\bt}_v- \eta(v) t_v} \right|+ \left|{\sum\limits_{v \in \{0,1\}^{|{U}|} \setminus J} \widetilde{\boldeta}(v) \cdot \widetilde{\bt}_v- \eta(v) t_v}\right| \\ &\leq& \frac{7\kappa}{10} + \frac{3\kappa}{10} ={\kappa}. 
	\end{eqnarray*}
 For the query complexity, note that the total number of queries that the procedure makes to the samples is at most $|\bS|\cdot (|\bI|+|U|)=O\left(\frac{2^{2|U|}}{\kappa^5}\log^2\left(\frac{2^{|U|}}{\kappa\gamma}\right)\right)$.
	This completes the proof.
\end{proofof}

\subsection{Verifying and searching for a weakly robust detailing}
In this section we will show how to obtain a weakly robust detailing with respect to a set of variables. We first describe a procedure that will be used as a single step towards this goal.

\begin{figure}[ht!]
	\begin{framed}
		\noindent Procedure $\textsf{Test-Weakly-Robust-Detailing}(\delta,k,\gamma,U)$
		\begin{flushleft}
			\noindent {\bf Input:} Sample and query access to a distribution $\mu$ over $\{0,1\}^n$, $U \subseteq[n]$, parameters $k\in[n]$ and $\gamma \in (0,1)$.\\
			{\bf Output:} Either \textbf{Accept} or a set $U'\subseteq [n]\setminus U$ with $|{U'}| \le k$. 
			\begin{enumerate}
				\item Set $\widetilde{\ix}(\mu^U) =\textsf{Estimate-Index}(U, \delta/10,\gamma/3)$.
				\item Let $\mathcal{L} = \{U' \subseteq [n]\setminus U : \   |U'|\le k\}$.
				\item Sample $r=\left\lceil\frac{\log(3/\gamma)}{\delta}\right\rceil$ sets $\bU'_1,\ldots,\bU'_r \in \calL$ uniformly at random.\label{Alg_Step}
				\item For $i=1$ to $r$ do:
				\begin{enumerate}
					\item Set $\widetilde{\ix}(\mu^{U\cup \bU'_i}) =\textsf{Estimate-Index}(U\cup\bU'_i,\delta/10,\gamma/3r)$.
					\item If $\widetilde{\ix}(\mu^{U\cup \bU'_i}) -\widetilde{\ix}(\mu^U) > 7\delta/10, $ \textbf{Return }$\bU'_i$.
				\end{enumerate}
				\item \textbf{Return Accept}.
			\end{enumerate}
		\end{flushleft}\vskip -0.14in
	\end{framed}\vspace{-0.2cm}
	\caption{A description of the \textsf{Test-Weakly-Robust-Detailing} procedure.} \label{algo:testrobust}
\end{figure}

\begin{lemma}\label{lem:testweaklyrobust} Fix $\gamma,\delta\in (0,1)$, $U\subseteq[n]$, $k\in[n]$.
	Consider \textsf{Test-Weakly-Robust-Detailing}$(\delta,k,\gamma,U)$ as described in Figure~\ref{algo:testrobust}. With probability at least $1-\gamma$, the algorithm either accepts  $U$ or reports some $U' \subseteq [n]\setminus U$ such that $|{U'}| \leq k$, satisfying the following:
	\begin{enumerate}
		\item If  $\mu^U$ is not $(\delta,k)$-weakly robust, then \textsf{Test-Weakly-Robust-Detailing} does not accept and outputs some $U'$;\label{prove1}
		\item If \textsf{Test-Weakly-Robust-Detailing} outputs some $U'$, then $ \ix(\mu^{U\cup U'}) > \ix(\mu^U) + 4\delta/10$.\label{prove2}
	\end{enumerate}
 The total number of queries that the algorithm makes is at most $O\left(\frac{2^{2(|U|+k)}}{\delta^6}\cdot \log^3\left(\frac{2^{|U|+k}}{\delta\gamma}\right)\right)$.
\end{lemma}

\begin{proof}
	We start by proving (\ref{prove1}). By Lemma~\ref{lem:estindex}, we have that with probability at least $1-\gamma/3$, $|\widetilde{\ix}(\mu^U)-\ix(\mu^U)|\le \delta/10$. In addition, since $\mu^U$ is not $(\delta,k)$-weakly robust, we have that for at least $\delta|\calL|$ of the sets $U'\in\calL$, it holds that $\ix(\mu^{U\cup U'})-\ix(\mu^U)>\delta$. Therefore, with probability at least $1-\gamma/3$, one such set $\bU^*$ is sampled in Step~\ref{Alg_Step} of the algorithm. By applying Lemma~\ref{lem:estindex} and a union bound, we have that with probability at least $1-\gamma/3$, for all $i\in[r]$ it holds that $|\widetilde{\ix}(\mu^{U\cup \bU'_i})-\ix(\mu^{U\cup \bU'_i})|\le \delta/10$. Therefore, in particular,
	\[\widetilde{\ix}(\mu^{U\cup \bU^*})-\widetilde\ix(\mu^{U})\ge {\ix}(\mu^{U\cup \bU^*})-\ix(\mu^{U})-\frac{2\delta}{10}\ge \frac{8\delta}{10},\]
	and $\bU^*$ (or some other set) will be returned by the algorithm. A union bound over the above three events completes the proof of this item.
	
	To prove (\ref{prove2}), note that when some $\bU'$ is returned, we have that $\widetilde{\ix}(\mu^{U\cup\bU'})-\widetilde\ix(\mu^{U})> 7\delta/10$. By Lemma~\ref{lem:estindex}, similarly to item (\ref{prove1}), with probability at least $1-2\gamma/3$ we have $|\widetilde{\ix}(\mu^{U\cup\bU'})-\ix(\mu^{U\cup \bU'})|\le \delta/10$ and $|\widetilde{\ix}(\mu^{U})-\ix(\mu^{U})|\le \delta/10$.
	Conditioned on this, if ${\ix}(\mu^{U\cup \bU'})-\ix(\mu^{U})\le 4\delta/10$, then we have
	\[\widetilde{\ix}(\mu^{U\cup \bU'})-\widetilde\ix(\mu^{U})\le \frac{4\delta}{10}+\frac{2\delta}{10}<\frac{7\delta}{10},\]
	which is a contradiction to $\bU'$ being reported by the algorithm.

    To finish the proof we bound the total number of queries that the algorithm makes. Note that the algorithm uses one call to $\textsf{Estimate-Index}$ with $U$ and the rest of the $r$ calls to $\textsf{Estimate-Index}$ with a set $U\cup U'$ which can be of size at most $|U|+k$. Therefore, the total number of queries the algorithm makes is at most $O\left(\frac{2^{ 2(|U|+k)}}{\delta^6}\cdot \log^3\left(\frac{2^{|U|+k}}{\delta\gamma}\right)\right)$.
\end{proof}

Now we leverage \textsf{Test-Weakly-Robust-Detailing} to design a procedure, \textsf{Find-Weakly-Robust-Detailing} (see Figure~\ref{algo:findrobust}), which finds a weakly robust detailing $\mu^U$ with respect to some $U\subseteq[n]$. This procedure along with the parameter estimation procedure in the next section serve as the main information-gathering mechanism in our estimation algorithm.

\begin{figure}[ht!]
	\begin{framed}
		\noindent Procedure $\textsf{Find-Weakly-Robust-Detailing}(\delta,k,\gamma)$
		\begin{flushleft}
			\noindent {\bf Input:} Sample and query access to a distribution $\mu$ over $\{0,1\}^n$, $k\in [n]$ and parameters $\delta,\gamma\in(0,1)$.\\
			{\bf Output:} $U\subseteq [n]$ such that with probability at least $1-\gamma$, $\mu^U$ is $(\delta,k)$-weakly robust.
			\begin{enumerate}
				\item Set $U=\emptyset$ and $\ell=1$.
				\item While $\ell \leq \left\lceil\frac{10}{4\delta}\right\rceil$ do:\label{step:2!}
				\begin{enumerate}
					\item Set $\bZ=\textsf{Test-Weakly-Robust-Detailing}(\delta,k,2\gamma\delta/10,U)$
					\item If  $\bZ=\textbf{Accept}$, then \textbf{Return} $U$.\label{step:return}
					\item Otherwise,
					\begin{itemize}
						\item $U'\leftarrow \bZ$. 
						\item Set $U=U\cup U'$ and $\ell=\ell+1$.
					\end{itemize}
				\end{enumerate}
			\item \textbf{Return Fail}\label{step:fail}
			\end{enumerate}
		\end{flushleft}\vskip -0.14in
	\end{framed}\vspace{-0.2cm}
	\caption{A description of the \textsf{Find-Weakly-Robust-Detailing} procedure.} \label{algo:findrobust}
\end{figure}

\begin{lemma}\label{lem:findweakly} Consider $\textsf{Find-Weakly-Robust-Detailing}(\delta,k,\gamma)$ as described in Figure~\ref{algo:findrobust}. With probability at least $1-\gamma$, the procedure outputs $U \subset [n]$ such that $|U|=O(k/\delta)$ and $\mu^U$ is $(\delta,k)$-weakly robust. The total number of queries that the procedure makes is at most $O\left(\frac{2^{5k(1/\delta +1)}}{\delta^7}\log^3\left(\frac{2^{k(1/\delta+1)}}{\gamma\delta}\right)\right)$.
\end{lemma}

\begin{proof}
	Note that since $\ell\le \lceil10/4\delta\rceil$, by a union bound we have that with probability at least $1-\gamma$ all the calls for \textsf{Test-Weakly-Robust-Detailing}($\delta,k,2\gamma\delta/10,U$) satisfy the conclusions of  Lemma~\ref{lem:testweaklyrobust}. Conditioned on this event, whenever in Step (\ref{step:2!}) the detailing $\mu^U$ was not $(\delta,k)$-weakly robust we obtained a set $U'$ so that $\ix(\mu^{U'})-\ix(\mu^U)>4\delta/10$ (and in particular did not return $U$).
	
	In addition, conditioned on the above event, since the index of a detailing is always between $0$ and $1$ it is not possible to pass through all $\lceil\frac{10}{4\delta}\rceil$ iterations of Step (\ref{step:2!}) and reach the failure mode of Step (\ref{step:fail}). Hence, eventually the algorithm will return a set $U$ in Step (\ref{step:return}), such that $\mu^U$ is a $(\delta,k)$-weakly robust detailing.

    To bound the number of queries note that since the algorithm starts with $U=\emptyset$, and by the guarantees of \textsf{Test-Weakly-Robust-Detailing}, at any of the $\ell$ iterations of the algorithm the size of the input $U$ to  \textsf{Test-Weakly-Robust-Detailing} is at most $k(\frac{10}{4\delta}+1)=O(k/\delta)$. Therefore the total number of queries is at most $O\left(\frac{2^{5k(1/\delta +1)}}{\delta^7}\log^3\left(\frac{2^{k(1/\delta+1)}}{\gamma\delta}\right)\right)$.
\end{proof}
\subsection{Estimation of detailing parameters}
\newcommand{\estparam}{\textsf{Estimate-Parameters}}

In this section we will design a procedure, \estparam ~(see Figure \ref{algo:estparamnew}), for estimating the parameters of a detailing $\mu^U$ with respect to some $U\subseteq[n]$ admitting type distribution $\Lambda$ and weight distribution $\eta=\mu^U|_2$. Our main goal for this section is to prove the following.

\begin{lemma}\label{theo:estparamalgnewcorrectness}
Let $\mu$ be a distribution over $\{0,1\}^n$, $U\subseteq [n]$ and $\kappa, \gamma \in (0,1)$ be parameters. Then the procedure $\estparam(\mu,U,\kappa,\gamma)$ (see Figure~\ref{algo:estparamnew}) outputs a pair of distributions $(\widetilde{\boldeta},\widetilde\bLambda)$ such that with probability at least $1- \gamma$, $\dtv(\widetilde{\boldeta},\eta)\le \kappa$ and $\dem^{\eta}(\widetilde\bLambda, \Lambda) \leq \kappa$, where $\eta=\mu^U$ is the weight distribution of $\mu^U$ and $\Lambda$ is the type distribution of $\mu^U$. The algorithm makes at most $O\left(\frac{2^{2|U|}}{\kappa^{2^{|U|+1}+5}}\cdot\log^2\left(\frac{2^{|U|}}{\gamma\kappa^{2^{|U|}}}\right)\right)$ queries to the input.
\end{lemma}

\begin{figure}[ht!]
	\begin{framed}
		\noindent Procedure $\estparam(\mu,U,\kappa,\gamma)$
		\begin{flushleft}
			\noindent {\bf Input:} Sample and query access to a distribution $\mu$ over $\{0,1\}^n$, $U \subseteq [n]$ and  parameters $\kappa, \gamma \in (0,1)$.\\
			{\bf Output:} A pair of distributions $(\widetilde\boldeta,\widetilde\bLambda)$, where $\widetilde{\boldeta}$ is a distribution on $A=\zo^{|U|}$ and $\widetilde{\bLambda}$ is a distribution on $[0,1]^A$.
			\begin{enumerate}
            \item Set $\rho=\frac{1}{4\cdot\lceil1/\kappa\rceil}$ and $r=2^{|U|}$. Also, let $\calR=\{0,\rho/10,\ldots,1\}^{2^{|U|}}$.
            \item Take a multi-set of samples $\bS$ from $\mu$ with $|{\bS}|= \left\lceil 50 \cdot \frac{r^2}{\rho^3}\log \frac{r}{\rho \gamma}\right\rceil=O\left( \frac{ 2^{2|U|} }{\kappa^3} \log \frac{2^{|U|}}{\kappa \gamma}\right)$.
            \item Construct the multi-set of vectors $\bS|_{U}$ by restricting the samples over the indices $U$.

            \item  For every $v\in \zo^{|U|}$ let \label{Step--4} $$\widetilde{\boldeta}({v})=\frac{|\left\{{\bs} \in \bS : {\bs}|_{ {U}}= v\right\}|}{|\bS|}.$$

            \item  Sample a set of indices uniformly at random without replacement $\bI \subseteq [n]$ where $|{\bI}| =\left\lceil 20\cdot \frac{|\calR|^2}{\rho^2}\log \frac{|\calR|}{\gamma}\right \rceil$, noting that $|\calR| = \left(\frac{10}{\rho}+1\right)^{2^{|U|}}$.
            \item For every $v \in \{0,1\}^{|U|}$ and $i \in \bI$, set 

        \begin{itemize}
            \item[(i)] 
                $\widehat{\balpha}_{v,i}=\frac{|\{{\bs} \in \bS :\; {\bs}|_{ U}= v
                \;\wedge\; \bs_i=1\}|}{|\{{\bs} \in \bS :\; {\bs}|_{ U}= v\}|}\; \text{if } \{{\bs} \in \bS : {\bs}|_{ U}= v\}\neq \emptyset,\; \text{and }\widehat{\balpha}_{v,i}=0\;\text{otherwise}. $

            \item[(ii)] Round off  $\widehat{\balpha}_{v,i}$ to the nearest value in $\{0,\rho/10,\ldots,1\}$, and denote this value by $\widetilde{\balpha}_{v,i}$.
        \end{itemize}
        \item For each $i \in \bI$, assign it the type vector $\widetilde{\balpha}^{(i)}=\langle \widetilde{\balpha}_{v,i} \rangle_{v \in \{0,1\}^{|U|}} \in \calR$.

        \item Construct a distribution $\widetilde\bLambda$ over $ \calR$ such that for each $a \in \calR$, we have:
        $$\widetilde\bLambda(a)=\frac{|\{i \in \bI:  \widetilde{\balpha}^{(i)}=a\}| }{|\bI|}.$$
        \item \textbf{Return} $(\widetilde{\boldeta},\widetilde\bLambda)$.
			\end{enumerate}
		\end{flushleft}\vskip -0.14in
	\end{framed}\vspace{-0.2cm}
	\caption{A description of the $\estparam$ procedure.} \label{algo:estparamnew}
\end{figure}

We start by setting the stage with some definitions and preliminary results. Consider some $\rho \in (0,1)$ that satisfies $\rho=1/z$ for $z\in\N$. We use $\calR$ to denote the set $\{0,\rho/10,\ldots,1\}^{2^{|U|}}$ and note that  $|{\calR}|= (10/\rho+1)^{2^{|U|}} $.

\begin{definition} Fix $U\subseteq[n]$. We say that a vector $w\in[0,1]^{2^{|U|}}$ is a \emph{$\rho$-approximation} of the type $t\in[0,1]^{2^{|U|}}$ of $\mu^U$ if $|w(v)-t(v)|\le \rho$ for any $v\in \zo^{2^{|U|}}$ for which $\mu^U|_2(v)\ge \rho/2^{|U|}$.
\end{definition} 

\begin{observation}\label{obs:udist} Let $\eta=\mu^U|_2$ and fix $i\in [n]$.
	Let $t_i\in [0,1]^{2^{U}}$ be the type of a variable $i$ with respect to $\mu^U$, and let  $w \in [0,1]^{2^{U}}$ be a $\rho$-approximation of $t_i$. Then, $d^{\eta}_{\ell_1}(t_i,w) \leq 2 \rho$.    
\end{observation}
\begin{proof}By definition of the $\eta$-weighted $\ell_1$ distance,
	\begin{align*}
d^{\eta}_{\ell_1}(t_i,w)&= \Ex_{\bv\sim \eta}[|t_i(\bv)-w(\bv)|]\\
&=\sum_{v\in\zo^U:\;\eta(v)\ge \rho/2^{|U|}}\eta(v)\cdot |t_i(v)-w(v)|+\sum_{v\in\zo^U:\;\eta(v)<\rho/2^{|U|} }\eta(v)\cdot |t_i(v)-w(v)|\\
&\le \rho+2^{|U|}\cdot \frac{\rho}{2^{|U|}}=2\rho.
	\end{align*}
\end{proof}

\begin{definition}\label{def:type-approx-on-sample}
    Let $U \subseteq [n]$, $S$ be a multi-set over $\{0,1\}^n$, and $\rho \in (0,1)$ be a parameter so that $1/\rho \in\mathbb N$. For each $v \in \{0,1\}^{|U|}$ and $i \in [n]$, let $\widehat\alpha_{v}^{(i)}$ to be the fraction of $x$'s in $S$ such that $x_U=v$ and $x_i=1$, and let $\widetilde{\alpha}_{v}^{(i)}$ be the element in $\{0,\rho,2\rho,\ldots,1\}$ closest to $\widehat\alpha_{v}^{(i)}$. Let $\widehat\alpha^{(i)}=\langle \widehat\alpha_{v}^{(i)} \rangle_{v \in \{0,1\}^{|U|}}$ and $\widetilde{\alpha}^{(i)}=\langle \widetilde{\alpha}_{v}^{(i)} \rangle_{v \in \{0,1\}^{|U|}}$. We call $\widehat\alpha^{(i)}$ the \emph{type of $i$ according to $S$} and $\widetilde\alpha^{(i)}$ its \emph{$\rho$-rounding}.
    
    We also define the distribution $\hat\Lambda$ as the distribution over $\calR=\{0,\rho,2\rho,\ldots,1\}^{2^{|U|}}$ such that the probability of $a \in \calR$ is the fraction of $i \in [n]$ such that $\widetilde{\alpha}^{(i)}=a$. 
\end{definition}

\begin{definition}[$\rho$-Good Sample Set]\label{def:goodSampleSet}
A multi-set $S$ over $\zo^n$ is said to be a \emph{$\rho$-good sample set} with respect to $\mu^U$, if for at least $(1-\rho)n$ of the variables $i \in [n]$, $\widetilde{\alpha}^{(i)}$ is a $\rho$-approximation of the type $t_i$ of the detailing $\mu^U$.
\end{definition}

\begin{lemma}\label{lem:distdtusig}
Fix $U \subseteq [n]$, $\rho \in (0,1)$ and let $\mu^U$ be the detailing of $\mu$ with respect to $U$ admitting type distribution $\Lambda$.
If $S\subseteq\{0,1\}^n$ is a $\rho$-good sample set with respect to $\mu^U$, then letting $\eta=\mu^U|_2$ we have $\dem^{\eta}(\Lambda,\hat{\Lambda})\le 3\rho$.
\end{lemma}
\begin{proof}
    As $S$ is a $\rho$-good sample set, we know that for at least $(1-\rho)n$ of the variables $i \in [n]$, $\widetilde{\alpha}^{(i)}$ is a $\rho$-approximation of the type $t_i$. Let $W\subseteq[n]$ be the set of all such variables with $|W| \geq (1-\rho)n$.  By Observation~\ref{obs:udist},
$$W\subseteq\left\{i \in [n]: d^{\eta}_{\ell_1}(\widetilde{\alpha}^{(i)},t_{i}) \leq 2\rho\right\}.$$

Consider the function $f: \left([0,1]^{2^{|U|}}\right)^2 \rightarrow [0,1]$ defined as follows.
\[f(a,b)= \begin{cases} 
       \frac{|\{i \in W \;:\;  t_{i}=a\}| }{n} & \widetilde{\alpha}^{(i)}=b \\
       0 & otherwise 
   \end{cases}.
\]

From the definition of $f$ and Observation~\ref{obs:udist}, note that $f(a,b) \neq 0$ implies that $d^{\eta}_{\ell_1}(a,b) \leq 2 \rho$. Consider a transfer function $T: \left([0,1]^{2^{|U|}}\right)^2 \rightarrow [0,1]$ from $\hat\Lambda$ to $\Lambda$ satisfying $T(a,b)\geq f(a,b)$ for all $a, b \in [0,1]^{2^{|U|}}$.
Note that such a $T$ exists, since clearly $\sum_{a,b \in [0,1]^{2^{|U|}}} f(a,b)\leq 1$. Also note that $\sum_{a,b \in [0,1]^{2^{|U|}}} f(a,b)=1-|W|/n\geq 1-\rho$, and hence $\sum_{a,b \in [0,1]^{2^{|U|}}:\;f(a,b)=0}  T(a,b) \leq \rho$.
Thus, we have 
\begin{align*}
    \dem^{\eta}(\hat{\Lambda},\Lambda)&\le \Ex_{(\ba,\bb)\sim T}[d^{\eta}_{\ell_1}(\ba,\bb)]\\
    &=\sum_{a,b\in [0,1]^{2^{|U|}}:\;f(a,b)\neq 0}T(a,b)\cdot d^{\eta}_{\ell_1}(a,b)+\sum_{a,b\in [0,1]^{2^{|U|}}:\;f(a,b)= 0}T(a,b)\cdot d^{\eta}_{\ell_1}(a,b)\\
    &\le 2\rho\cdot \sum_{a,b\in [0,1]^{2^{|U|}}:\;f(a,b)\neq 0}T(a,b) + \sum_{a,b\in [0,1]^{2^{|U|}}:\;f(a,b)= 0}T(a,b)\cdot 1\le 3\rho,
\end{align*}
and the lemma follows.
\end{proof}

\begin{definition}\label{def:approx-type-dist}
Let $U \subseteq [n]$, $S$ be a multi-set over $\zo^n$, and $\rho \in (0,1)$ where $1/\rho\in\mathbb N$ be a parameter. Let $\widehat\alpha^{(i)}$ be the type of $i$ according to $S$ and $\widetilde{\alpha}^{(i)}$ be its $\rho$-rounding over $S$, as per Definition~\ref{def:type-approx-on-sample}.
For each $a \in \calR$, let $f_a$ be the fraction of $i \in [n]$ such that $\widetilde{\alpha}^{(i)}=a$. Let $I \subseteq [n]$ be a multi-set of variables. For each $a \in \calR$, let $\widetilde{f}_a$ be the fraction of $i \in I$ such that $\widetilde{\alpha}^{(i)}=a$.

The approximate type distribution ${\widetilde\Lambda}$ with respect to $I$ and $S$ is a distribution over $\calR$ such that the probability of $a \in \calR$ is the fraction of $i \in I$ such that $\widetilde{\alpha}^{(i)}=a$.
\end{definition}

\begin{definition}[$\rho$-Good Variable Set] Let $S$ be a multi-set over $\{0,1\}^n$. Using the notation of Definition~\ref{def:approx-type-dist}, a multi-set $I$ over $[n]$ is said to be \emph{$\rho$-Good-Variable-Set} with respect to $S$ if $|{\widetilde{f}_a-f_a}| \leq \rho/|\calR|$ for each $a \in \calR$.
\end{definition}

\begin{lemma}\label{lem:distdtdi}
Let $S$ be a multi-set over $\{0,1\}^n$ and $I \subseteq [n]$ be a $\rho$-good variable set with respect to $S$. Then, $\dtv(\hat\Lambda,{\widetilde\Lambda}) \leq \rho$.
\end{lemma}
\begin{proof} By definition of the distributions $\hat\Lambda$ and ${\widetilde\Lambda}$,
\begin{eqnarray*}
\dtv(\hat\Lambda,{\widetilde\Lambda})\le \sum_{a \in \calR} |\hat{\Lambda}(a)- {\widetilde\Lambda}(a)| = \sum_{a \in \calR} |{\widetilde{f}_a-f_a}| 
\leq|\calR| \cdot \frac{\rho}{|\calR|} \leq \rho.
\end{eqnarray*}\end{proof}
\begin{lemma}\label{lem:distdiusig} Let $\eta=\mu^U|_2$ for some $U\subseteq[n]$.
Let $S \subseteq \{0,1\}^n$ be a $\rho$-{good-sample set} and $I \subset [n]$ be a $\rho$-good variable set with respect to $S$.
 Then, $\dem^{\eta}(\Lambda,{\widetilde\Lambda})\le 4\rho$.
\end{lemma}

\begin{proof}
From Lemma~\ref{lem:distdtusig}, we know that $\dem^{\eta}(\Lambda, \hat{\Lambda}) \leq 3 \rho$. Also, from Lemma~\ref{lem:distdtdi}, we have that 
$\dtv(\hat\Lambda,{\widetilde\Lambda})\leq\rho$, and hence $\dem^{\eta}(\hat\Lambda,{\widetilde\Lambda}) \leq \rho$. Using the triangle inequality we obtain $\dem^{\eta}(\Lambda, {\widetilde\Lambda}) \leq 4 \rho$, completing the proof of the lemma.
\end{proof}   

\subsubsection{Proof of Lemma of \ref{theo:estparamalgnewcorrectness}}
\begin{definition}[Definition of the event $ \calE^{*}$] The event $\calE^{*}$ is defined as the intersection of the following three events:
\begin{description}
\item[$\calE_{\widetilde{\boldeta}}$:] The output $\widetilde{\boldeta}$ produced by Algorithm~\ref{algo:estparamnew} satisfies $\dtv(\widetilde{\boldeta},\eta)\le \kappa$.
\item[$\calE_\bS$:] The multi-set of samples $\bS$ taken in Step $2$ of Algorithm~\ref{algo:estparamnew} is a $\kappa/4$-good sample set.
\item[$\calE_\bI$:] The set of indices $\bI$ taken in Step $5$ of Algorithm~\ref{algo:estparamnew} is a $\kappa/4$-good variable set.
\end{description}
\end{definition}

We will prove Lemma~\ref{theo:estparamalgnewcorrectness} by a series of lemmas stated below. The lemmas themselves will be proved in Section~\ref{subs:goodness-proofs}.

\begin{lemma}\label{lem:event-eta-good} Event $\calE_{\widetilde{\boldeta}}$ holds with probability at least $1-\gamma/3$.
\end{lemma}

\begin{lemma}\label{lem:T-good}
Event $\calE_\bS$ holds with probability at least $1 - \gamma/3$.   
\end{lemma}

\begin{lemma}\label{lem:I-good}
Event $\calE_\bI$ holds with probability at least $1- \gamma/3$.  \end{lemma}

\begin{proofof}{Lemma~\ref{theo:estparamalgnewcorrectness}}
Assuming that Lemma~\ref{lem:event-eta-good}, Lemma~\ref{lem:T-good} and Lemma~\ref{lem:I-good} hold, the event $\calE^*$ holds with probability at least $1-\gamma$ by a union bound.

When $\calE^{*}$ holds, we know that $\dtv(\widetilde{\boldeta},\eta)\le \kappa$,  the multi-set of samples $\bS$ taken in Step $2$ is a good sample set as well as that the set of indices $\bI$ is a good variable set. Following Lemma~\ref{lem:distdiusig}, we have that $\dem^{\eta}(\Lambda,{\widetilde\bLambda})\le 4\rho\le \kappa $. 

Note that since $|\calR|=(10/\rho+1)^{2^{|U|}}$ and $\rho=O(\kappa)$, the total number of queries the algorithm makes is at most $|\bS|\cdot(|\bI|+|U|)=O\left(\frac{2^{2|U|}}{\kappa^{2^{|U|+1}+5}}\cdot\log^2\left(\frac{2^{|U|}}{\gamma\kappa^{2^{|U|}}}\right)\right)$. \end{proofof}

\subsubsection{Proofs of Lemmas~\ref{lem:event-eta-good}, \ref{lem:T-good} and \ref{lem:I-good}}\label{subs:goodness-proofs}

We state the following lemmas whose proofs are identical to the proofs of Lemma~\ref{cl:estimateeta} and Lemma~\ref{cl:aplhavigap} respectively and are thus omitted.

\begin{lemma}\label{cl:estimateeta1new} Fix $U\subseteq[n]$ and $\rho,\gamma\in(0,1)$ and let $r=2^{|U|}$. Consider the detailing $\mu^U$ and let $\eta=\mu^U|_2$ and $J=\{v\in\zo^{{|U|}}:\eta(v)\ge \rho/r\}$. For $v \in \{0,1\}^{|U|}$,
we define the following event $\calE_{\widetilde{\boldeta}(v)}$: The value $\widetilde{\boldeta}(v)$ defined in Step \ref{Step--4} of $\estparam$ satisfies the following.
\begin{enumerate}
    \item[(i)] If $v \in J$, $|{\widetilde{\boldeta}(v) - \eta(v)}| \leq \frac{\rho}{2r}$.
    \item[(ii)] If $v \notin J$ ,  $\widetilde{\boldeta}(v)\leq  \frac{3\rho}{2r}$.
\end{enumerate}
Then the probability that $\calE_{\widetilde{\boldeta}(v)}$ holds is at least $1-\frac{\rho\gamma}{6r}$.
\end{lemma}

\begin{lemma}\label{lem:approxbetav}Fix $U\subseteq[n]$ and $\rho,\gamma\in(0,1)$ and let $r=2^{|U|}$. Consider the detailing $\mu^U$ and let $\eta=\mu^U|_2$ and $J=\{v\in\zo^{{|U|}}:\eta(v)\ge \rho/r\}$. For each $v \in J$ and $i \in [n]$, we let $\calE_{\hat{\balpha}_{v,i}}$ be the event that the value $\widehat{\balpha}_{v,i}$ satisfies
$|{{\widehat\balpha}_{v,i}-t_i(v)}|\leq \frac{\rho}{4}$. Then the probability of $\calE_{\hat{\balpha}_{v,i}}$ is at least $1 - \frac{\rho \gamma}{3r}.$
\end{lemma}

From Lemma~\ref{lem:approxbetav} we obtain the following lemma.
\begin{lemma}\label{coro:betaiaprpox}
Fix $i \in [n]$. Then $\widetilde{\balpha}^{(i)}$ is a $\rho$-approximation of $t_i$ with probability at least $1-\rho\gamma/3$.
\end{lemma}
\begin{proof} Let $\calE$ be the following event:
    $$\calE:~|{\widetilde{\balpha}_{v}^{(i)}-t_i(v)}| \leq \rho~\mbox{for each}~v \in J.$$

If the event $\calE$ holds, then $\widetilde{\balpha}^{(i)}$ is a $\rho$-approximation of $t_i$. Hence, it remains to show that $\Prx[\calE] \geq 1-\gamma \rho/3$.

Consider any particular $v \in J$, i.e, $\eta(v) \geq \frac{\rho}{r}$. By Lemma~\ref{lem:approxbetav}, $|{\widehat\balpha_{v}^{(i)}-t_i(v)}| \leq \frac{\rho}{4}$ with probability at least $1-\rho\gamma/3r$. From the construction of $\widetilde{\balpha}_{v}^{(i)}$, note that $|{\widetilde{\balpha}_{v}^{(i)}-\widehat{\balpha}_v^{(i)}}| \leq {\rho}/20$. By the triangle inequality, $|{\widetilde{\balpha}_{v}^{(i)}-t_i(v)}| \leq 3\rho/10 <\rho$.
By the union bound over all $v \in J$, event $\calE$ holds with probability at least $1-\rho \gamma/3$. 
\end{proof}

Now we are ready to prove Lemma~\ref{lem:event-eta-good}, Lemma~\ref{lem:T-good} and Lemma~\ref{lem:I-good}. 

\begin{proofof}{Lemma~\ref{lem:event-eta-good}} By Lemma~\ref{cl:estimateeta1new} and a union bound we have that the following holds for all $v\in\zo^{|U|}$ with probability at least $1-\rho\gamma/4>1-\gamma/3$.  If $v \in J$, $|{\widetilde{\boldeta}(v) - \eta(v)}| \leq \frac{\rho}{2r}$, and otherwise  $\widetilde{\boldeta}(v)\leq  \frac{3\rho}{2r}$. Conditioned on that, 
\begin{align*}
    \dtv(\widetilde{\boldeta},\eta)=\frac{1}{2}\left(\sum_{v\in J}|\widetilde{\boldeta}(v)-\eta(v)|+\sum_{v\notin J}|\widetilde{\boldeta}(v)-\eta(v)|\right)\le \frac{1}{2}\left(\sum_{v\in J}\frac{\rho}{2r}+\sum_{v\notin J}\frac{3\rho}{2r}\right)\le \rho.
\end{align*}
As $\rho=\frac{1}{4\lceil1/\kappa\rceil}$ the lemma follows.
\end{proofof}

\begin{proofof}{Lemma~\ref{lem:T-good}}
Let $\bX_i$ be a random variable such that 
\[\bX_i = \begin{cases}
        1 &  \widetilde{\balpha}^{(i)} \  \mbox{is not a $\rho$-approximation of $t_i$} \\
        0 & \text{otherwise}
\end{cases}.
\]
Let $\bX=\frac{1}{n}  \sum\limits_{ i \in [n]} \bX_i$.  From Lemma~\ref{coro:betaiaprpox}, $\E[\bX_i]\leq \frac{\rho \gamma}{3}$, and hence $\E[\bX]\leq \frac{\rho\gamma}{3}.$ By Markov inequality, we have
$$\Prx \left[ \bX \geq \rho \right]\leq \frac{\E[\bX]}{\rho} \leq \frac{\gamma}{3}.$$
Thus, with probability at least $1-\gamma/3$, for at least $(1-\rho)n$ variables $i \in [n]$, $\widetilde{\balpha}^{(i)}$ is a $\rho$-approximation of $t_i$. Hence, $\bS$ is a $\rho$-good sample set with probability at at least $1-\gamma/3$. 
\end{proofof}

\begin{proofof}{Lemma~\ref{lem:I-good}}
Recall that  $\hat\Lambda$ be the type distribution over $\calR$ with respect to $\bS$ such that, for $a \in \calR$,
$$\hat\Lambda(a)=\frac{|\{i \in [n]:\widetilde{\balpha}^{(i)}=a\}|}{n}.$$

Also, consider the distribution $\widetilde\bLambda$ reported by the algorithm. Note that $\widetilde\bLambda$ is the type distribution over $\calR$  with respect to $\bI$ and $\bS$, i.e., for $a \in \calR$,
$$\widetilde\bLambda( a)=\frac{|\{i \in \bI:\widetilde{\balpha}^{(i)}=a \}|}{|\bI|}.$$

Consider a fixed $a \in \calR$, and note that $\E_{\bI}[\widetilde\bLambda(a)]=\hat{\Lambda} (a)$.  By using the Hoeffding bound (Lemma~\ref{lem:hoeffdingineq_without_replacement}), 
$$\Prx \left[ |\widetilde\bLambda(a)- \hat{\Lambda}(a)| \geq \frac{\rho}{|{\calR}|}\right] \leq \frac{\gamma}{3|{\calR}|}.$$

In the above, we have used $|{\bI}|=\left \lceil 20 \cdot \frac{|\calR|^2}{\rho^2}\log \frac{|\calR|}{\gamma}\right \rceil.$ Applying the union bound over all $a \in \calR$, we can say that for all $a \in \calR$,  $|\widetilde\bLambda(a)- \hat{\Lambda}(a)| \leq \frac{\rho}{|{\calR}|}$  with probability at least $1-\gamma/3$. Since $\rho=\frac{1}{4\lceil1/\kappa\rceil}$, this implies that the event $\calE_\bI$ holds with probability at least $1- \gamma/3$. This concludes the proof of the lemma.
\end{proofof}

\section{The estimation algorithm}\label{sec:estimation}
\newcommand{\EstimateDistance}{\textsf{Tolerant-Tester}}

In this section we prove our main result, that index-invariant properties that admit tests whose number of queries is independent of $n$ also admit such distance-estimation procedures.

\begin{theorem} \label{thm:main-estimation} Let $\calP$ be an index-invariant property of  distributions supported on $\zo^n$. If $\calP$ is $\eps'$-testable with $s=s(\eps')\in\N$ samples and $q=q(\eps')\in \N$ queries for every $\eps'\in(0,1)$, and $n\ge 2^{\poly(2^q,s,1/\eps')}$, then for every $\eps\in(0,1)$ there exists an algorithm that given access to an unknown distribution $\mu$ over $\zo^n$, performs at most $\exp\left(\exp\left(2^{q(\Omega(\eps))}\cdot \poly\left(s(\Omega(\eps)),q(\Omega(\eps)),{1/\eps}\right)\right)\right)$ queries, and outputs a value $\widetilde{\bd}$ such that with probability at least $2/3$ it holds that $|\widetilde{\bd}-\dem(\mu,\calP)|\le \eps$.
\end{theorem}

Theorem~\ref{thm:main-estimation} follows directly from the following lemma about the existence of an $(\eps_1,\eps_2)$-tolerant test for any index-invariant property $\calP$ which admits an $\eps$-test for $\eps=(\eps_2-\eps_1)/12$ with $s$ samples and $q$ queries (see Claim 2 in~\cite{PRR06}).  

\begin{lemma}\label{lem:tolerant-test} Suppose that an index-invariant property of distributions $\calP$ has an $\eps$-test with $s(\eps)$ samples and $q(\eps)$ queries for every $\eps\in(0,1)$. Then, for every $0<\eps_1<\eps_2<1$, $\EstimateDistance$ (see Figure~\ref{fig:Estimation}) is an $(\eps_1,\eps_2)$-tolerant tester for $\calP$ that makes at most 
\[\exp\left(\exp\left(2^q\cdot \poly\left(s((\eps_2-\eps_1)/12),q((\eps_2-\eps_1)/12),1/({\eps_2-\eps_1})\right)\right)\right)\]
queries {to the samples obtained from} the unknown input distribution.
\end{lemma}

\begin{figure}[ht!]
	\begin{framed}
		\noindent Procedure $\EstimateDistance(\mu,\calP,\eps_1,\eps_2,s,q)$
		\begin{flushleft}
			\noindent {\bf Input:} Sample and query access to $\mu$, index-invariant property $\calP$, and parameters $\eps_1,\eps_2\in(0,1)$ such that $\eps_1<\eps_2$, number of samples $s$ and queries $q$.\\
			{\bf Output:} \textbf{Accept} if $\mu$ is $\eps_1$-close to $\calP$, and \textbf{Reject} if $\mu$ is $\eps_2$-far from $\calP$.
			\begin{enumerate}
                \item Let $\delta=\frac{(\eps_2-\eps_1)^2}{10^5\cdot q^3(s+1)^6}$, $k=q-1$, $k'=\Theta(2^{8k/\delta}/\delta^{38})$ and $\rho= \left(\left\lceil\frac{36sq}{\eps_2-\eps_1}\right\rceil\right)^{-1}$.\label{step:1a}
                
				\item\label{step:1} Call \textsf{Find-Weakly-Robust-Detailing}($\delta$,$k'$,$1/6$) on $\mu$ and obtain the detailing $\mu^U$ with respect to a set of variables $U\subseteq [n]$.
    
				\item Call \textsf{Estimate-Parameters}$(\mu,U,\rho,1/6)$ to obtain $(\tilde\boldeta,\tilde{\bLambda})$ where $\tilde\boldeta$ is the weight distribution over $2^{U}$ and $\tilde{\bLambda}$ is the type distribution over $[0,1]^{2^U}$. \label{step:1b}
				\item For each $B$ of size $|B|\le 2^{k/\delta}$, any $\frac{2\rho}{2^{U}\cdot|B|}$-quantized detailing $\eta'$ of $\tilde{\boldeta}$ with respect to $B$, and any $\frac{\rho}{(1/\rho+1)^{2^{|U|}\cdot |B|}}$-quantized type distribution $\tilde\Upsilon$ on $\{0,\rho,2\rho,\ldots,1\}^{2^U\times B}$, do the following: \label{estim:loop}
				\begin{enumerate}
					\item Call \textsf{Accept-Probability}$(s,q,\tilde\Upsilon,\eta', \frac{\eps_2-\eps_1}{12})$.  If the output is at least $1/2$, go to the next step. Otherwise, go to the next option for $B
     $, $\eta'$ and $\tilde\Upsilon$. \label{step1}
					\item Compute $\dem^{\eta'}(\tilde{\bLambda}_{\langle B\rangle}, \tilde\Upsilon)$, where $\tilde{\bLambda}_{\langle B\rangle}$ is the flat extension of $\tilde\bLambda$ with respect to $B$. If the computed distance is at most $\frac{\eps_1+\eps_2}{2}$, output \textbf{Accept} and terminate. %
				\end{enumerate}
				\item Output \textbf{Reject}.
			\end{enumerate}
		\end{flushleft}\vskip -0.14in
	\end{framed}\vspace{-0.2cm}
	\caption{Description of the $\EstimateDistance$ procedure.} \label{fig:Estimation}
\end{figure}

The proof of Lemma~\ref{lem:tolerant-test} follows from the following two lemmas, proved in Section~\ref{sec:completeness} and Section~\ref{sec:soundness}, respectively.

\begin{lemma}\label{lem:completenes}
Fix $0<\eps_1<\eps_2<1$, and set $s=s((\eps_2-\eps_1)/12)$ and $q=q((\eps_2-\eps_1)/12)$. For any $n\ge18q^2(s+1)/(\eps_2-\eps_1)$, if a distribution $\mu$ over $\zo^n$ is $\eps_1$-close to $\calP$, then $\EstimateDistance$ accepts it with probability at least $2/3$.
\end{lemma}

\begin{lemma}\label{lem:soundness}
Fix $0<\eps_1<\eps_2<1$, and set $s=s((\eps_2-\eps_1)/12)$ and $q=q((\eps_2-\eps_1)/12)$. For any $n\ge2^{\poly(2^q,s,1/(\eps_2-\eps_1))}$, if a distribution $\mu$ over $\zo^n$ is $\eps_2$-far from $\calP$, then $\EstimateDistance$ rejects it with probability at least $2/3$.
\end{lemma}

\begin{proofof}{Lemma~\ref{lem:tolerant-test}} Assuming Lemma~\ref{lem:completenes} and Lemma~\ref{lem:soundness} hold, the correctness of $\EstimateDistance$ is immediate. To bound the query complexity, note that the algorithm performs all of its queries in steps (\ref{step:1}) and (\ref{step:1b}). Thus, by the choice of parameters, Lemma~\ref{lem:findweakly} and Lemma~\ref{theo:estparamalgnewcorrectness}, the size of $U$ is at most $2^q\cdot\poly(s,q,1/(\eps_2-\eps_1))$, and the query complexity of the algorithm is at most $\exp\left(\exp\left(2^q\cdot \poly\left(s,q,1/({\eps_2-\eps_1})\right)\right)\right)$.
\end{proofof}

\subsection{Completeness (Proof of Lemma~\ref{lem:completenes})}\label{sec:completeness}
In this section, we prove that if the unknown distribution $\mu$ is $\eps_1$-close to the index-invariant property $\calP$, then it will be accepted by \textsf{Tolerant-Tester} (Figure~\ref{fig:Estimation}) with probability at least $2/3$.

\begin{lemma} \label{lem:RobustCloseness}Fix $\delta\in(0,1)$, $\ell\in \N$, let $\mu$ be a distribution over $\zo^n$ and let $\xi$ be a $(\delta,\ell)$-robust detailing of $\mu$ with respect to some set $A$. In addition, let $\xi'$ be a refinement of $\xi$ with respect to a set $B$ of size $\ell$, and let $\xi_{\langle\eta\rangle}\eqdef \xi_{\langle\xi'|_{2,3}\rangle}$ be a flat refinement of $\xi$ with respect to $\eta\eqdef\xi'|_{2,3}$. Then, letting $\Lambda'$ and $\Lambda_{\langle B\rangle}$ be the type distributions of $\xi'$ and $\xi_{\langle\eta\rangle}$ respectively, we have that  $\dem^{\eta}(\Lambda',\Lambda_{\langle B \rangle})\le \sqrt{\delta}$.
\end{lemma}
\begin{proof}
	By definition of the index and the fact that $\xi$ is $(\delta,\ell)$-robust, we can say the following:
	\begin{align}
		\delta\ge\ix(\xi')-\ix(\xi)&=\Ex_{i\sim [n]}\left[\Ex_{(\ba,\bb)\sim \eta}\left[\Prx_{\bx \sim  \xi'|_1^{2,3:(\ba,\bb)}}[\bx_{i}=1]^2\right]-\Ex_{\ba\sim\eta|_1}\left[\Prx_{\bx \sim  \xi|_1^{2:\ba}}[\bx_{i}=1]^2\right]\right]\notag\\
		&=\Ex_{i\sim[n]}\left[\Ex_{\ba\sim\eta|_1}\left[\Ex_{\bb\sim\eta|^{1:\ba}_2}\left[\Prx_{\bx\sim\xi'|^{2,3:(\ba,\bb)}_1}[\bx_i=1]^2\right]-\Ex_{\bb\sim\eta|^{1:\ba}_2}\left[\Prx_{\bx\sim\xi'|^{2,3:(\ba,\bb)}_1}[\bx_i=1]\right]^2\right]\right]\notag\\
		&=\Ex_{i \sim[n]}\left[\Ex_{(\ba,\bb)\sim \eta}\left[\left(\Prx_{\bx\sim\xi'|^{2,3:(\ba,\bb)}_1 }[\bx_i=1]-\Prx_{\bx\sim\xi|^{2:\ba}_1 }[\bx_i=1]\right)^2\right]\right]\label{eq:bound_index}.
	\end{align}
	Now, letting $\Lambda$ be the type distribution of $\xi$ and using the fact that $\xi_{\langle\eta\rangle}$ is a flat refinement of $\xi$, by Observation~\ref{obs:weightedEMD} (using the $\eta$-weighted $\ell_1$ distance as the metric):
	\begin{align*}
		&\dem^{\eta}(\Lambda',\Lambda_{\langle\eta\rangle})\le \Ex_{i\sim [n]}\Ex_{(\ba,\bb)\sim \eta}\left[\left|\Prx_{\bx \sim  \xi'|_1^{2,3:(\ba,\bb)}}[\bx_{i}=1]-\Prx_{\bx \sim  \xi_{\langle \eta \rangle}|_1^{2,3:(\ba,\bb)}}[\bx_i=1]\right|\right] \notag\\
		&\le\Ex_{i\sim [n]}\Ex_{(\ba,\bb)\sim \eta}\left[\left|\Prx_{\bx \sim  \xi'|_1^{2,3:(\ba,\bb)}}[\bx_{i}=1]-\Prx_{\bx \sim  \xi|_1^{2:\ba}}[\bx_i=1]\right|-\left|\Prx_{\bx \sim  \xi|_1^{2:\ba}}[\bx_{i}=1]-\Prx_{\bx \sim  \xi_{\langle\eta\rangle} |_1^{2,3:(\ba,\bb)}}[\bx_i=1]\right|\right]\notag\\
		&=\Ex_{i\sim [n]}\Ex_{(\ba,\bb)\sim \eta}\left[\left|\Prx_{\bx \sim  \xi'|_1^{2,3:(\ba,\bb)}}[\bx_{i}=1]-\Prx_{\bx \sim  \xi|_1^{2:\ba}}[\bx_i=1]\right|\right]\notag\\
		&\le \sqrt{\Ex_{i\sim [n]}\Ex_{(\ba,\bb)\sim \eta}\left[\left(\Prx_{\bx \sim  \xi'|_1^{2,3:(\ba,\bb)}}[\bx_{i}=1]-\Prx_{\bx \sim  \xi|_1^{2:\ba}}[\bx_i=1]\right)^2\right]}\le \sqrt{\delta},
	\end{align*}
	where in the last line we used the Cauchy-Schwartz inequality (Lemma~\ref{lem:CS}), followed by Equation~(\ref{eq:bound_index}).	
\end{proof}

In the following we will deal with a detailing $\Xi$ of a transfer distribution between two distributions $\mu$ and $\tau$ over $\zo^n$ with respect to a set $A$, rather than with a detailing of a ``single distribution'' over $\zo^n$. So typically such $\Xi$ can be viewed as a distribution over $\zo^n\times\zo^n\times A$, and in particular $\Xi|_{1,3}$ would be a detailing of $\mu$ and $\Xi|_{2,3}$ would be a detailing of $\tau$. The following lemma connects the EMD between $\mu$ and $\tau$ to a distance between the type distributions of the two detailings resulting from such $\Xi$.

\begin{lemma} \label{lem:DistStars} Let $T$ be a transfer function realizing the EMD between $\mu$ and $\tau$, and let $\Xi$ be a detailing of $T$ with respect to $A$. Additionally, let $\Lambda^*$ be the type distribution of the detailing $\Xi|_{1,3}$ over $\mu$, let $\Upsilon^*$ be the type distribution of the detailing $\Xi|_{2,3}$ of $\tau$, and set $\eta=\Xi|_3$ to be the weight distribution (common to all detailings mentioned here). Then, $\dem^{\eta}(\Lambda^*,\Upsilon^*)\le \dem(\mu,\tau)$.
\end{lemma}
\begin{proof} We first note that by definition of the EMD:
	\begin{align}
		\dem(\mu,\tau) &= \Ex\limits_{(\bx,\by) \sim T} \left[d_{H}(\bx,\by)\right] =   \Ex\limits_{\ba \sim \Xi |_3} \left[\Ex_{(\bx,\by)\sim \Xi |^{3:\ba}} \left[d_{H}(\bx,\by)\right]\right]\notag\\
		&= \Ex\limits_{\ba \sim \Xi |_3} \left[\Ex_{(\bx,\by)\sim \Xi |^{3:\ba}} \left[\Ex_{i\sim [n]}[ \indi_{\{\bx_i\neq\by_i\}}]\right]\right]\notag\\
		&=\Ex_{i\sim [n]}\left[ \Ex\limits_{\ba \sim \Xi |_3} \left[ \Ex_{(\bx,\by)\sim \Xi |^{3:\ba}} \left[\indi_{\{\bx_i\neq\by_i\}}\right]\right]\right]\notag\\ 
		&= \Ex_{i \sim [n]} \left[\Ex_{\ba\sim \Xi |_3}\left[\Prx_{(\bx,\by) \sim  \Xi|^{3:\ba}}[\bx_{i}=0,\by_i=1]+\Prx_{(\bx,\by) \sim  \Xi|^{3:\ba}}[\bx_{i}=1,\by_i=0]\right]\right] .\label{eq:EMDtoTypes}
	\end{align}
	On the other hand, by Observation~\ref{obs:weightedEMD}, we can say that:
	\begin{align*}
		\dem^{\eta}(\Lambda^*, \Upsilon^*) &\le \Ex_{i \sim[n]} \left[\Ex_{\ba\sim \Xi |_3}\left| \Prx_{(\bx,\by)\sim \Xi |^{3:\ba}}[\bx_i=1]- \Prx_{(\bx,\by)\sim \Xi |^{3:\ba}}[\by_i=1]\right|\right]\\
		&=\Ex_{i \sim [n]} \left[\Ex_{\ba\sim \Xi |_3} \Big|\Prx_{(\bx,\by) \sim  \Xi|^{3:\ba}}[\bx_{i}=1,\by_i=0]-\Prx_{(\bx,\by)\sim \Xi |^{3:\ba}}[\bx_{i}=0,\by_i=1]\Big|\right]\\
		&\le \Ex_{i \sim [n]} \left[\Ex_{\ba\sim \Xi |_3} \left[\Prx_{(\bx,\by) \sim  \Xi|^{3:\ba}}[\bx_{i}=1,\by_i=0]+\Prx_{(\bx,\by)\sim \Xi |^{3:\ba}}[\bx_{i}=0,\by_i=1]\right]\right].
	\end{align*}
	Combined with Equation~(\ref{eq:EMDtoTypes}) the lemma follows.
\end{proof}

As a heads-up, the above lemma will be eventually used for a detailing $\Xi$ with respect to a cross-product set $A\times B$, because it will be constructed from a join (see Definition \ref{def:join}) involving two detailings of two respective individual distributions. We are now ready to prove the completeness of $\EstimateDistance$.

\begin{proofof}{Lemma~\ref{lem:completenes}}
	Suppose that there exists a distribution  $\tau\in \calP$ for which $\dem(\mu,\tau)\le \eps_1$, and let $T$ be a transfer function exhibiting the distance. We will show that in such case, the algorithm \textsf{Tolerant-Tester} accepts. Let $k=q-1$, $\delta=\frac{(\eps_2-\eps_1)^2}{10^5\cdot (s+1)^6q^3}$ and $k'=\Theta\left(\frac{2^{8k/\delta}}{\delta^{38}}\right)$. Let $\mu^U$ be the detailing returned by the call to \textsf{Find-Weakly-Robust-Detailing}($\delta$,$k'$,$1/6$), let $\eta=\mu^U|_2$ be its weight distribution and let $\Lambda$ be its type distribution.
	
	Note that by the guarantees of \textsf{Find-Weakly-Robust-Detailing} (Lemma~\ref{lem:findweakly}) we have that with probability at least $5/6$, the returned set $U$ defines a $(\delta,k')$-weakly robust detailing. In addition, from the guarantees of  $\textsf{Estimate-Parameters}$ (Lemma~\ref{theo:estparamalgnewcorrectness}) we have that $\dtv(\tilde{\boldeta},\eta)\le\rho$ and $\dem^{\eta}(\tilde{\bLambda},\Lambda)\le \rho$ with probability at least $5/6$. We henceforth condition on the intersection of the above events (which happens with probability at least $2/3$), and prove that in this case the algorithm will indeed accept.
	
	Set $\Xi'=T\bowtie\mu^U$ (see Definition~\ref{def:join}, where in the join operation unify the first coordinate of $T$ with the
first coordinate of $\mu^U$), considering it as a distribution on $\zo^n\times\zo^n\times\zo^{|U|}$, and note that this is a detailing of $T$ with respect to $A\eqdef \zo^{|U|}$. Also set $\zeta=\Xi'|_{2,3}$, and note that this is a detailing of $\tau$ with respect to $A$. However, $\zeta$ is a general detailing of $\tau$, not necessarily one defined by variables.
	
	Next, use Lemma~\ref{lem:weakly_robust_refinement} to find a refinement $\zeta^V$ which is $(\delta,k)$-weakly robust, where $V\subset [n]$ is a variable set of size at most $k/\delta$. By Lemma~\ref{lem:weakly_robust_is_good}, the detailing $\zeta^V$ is also $\left(\frac{\eps_2-\eps_1}{100(s+1)},q\right)$-good.
	
	We now set $\Xi=\Xi'\bowtie\zeta^V$. We consider it to be a distribution over $\zo^n\times\zo^n\times(A\times B)$ for $B\eqdef\zo^{|V|}$, noting that it is a refinement of $\Xi'$ with respect to $B$. We also note that $\Xi|_{2,3}=\zeta^V$, and that $\Xi|_{1,3}$ is a refinement of $\xi$ with respect to $B$ (and is in particular a detailing of $\mu$, but no longer one defined by variables). We now let $\eta^*=\Xi|_3$ be the weight distribution of $\Xi$ (this is a distribution over $A\times B$ which is a detailing of $\eta$ with respect to $B$), let $\Lambda^*$ be the type distribution of $\Xi|_{1,3}$, and let $\Upsilon^*$ be the type distribution of $\Xi|_{2,3}$.
	
	By Lemma~\ref{lem:quatization} and Lemma~\ref{lem:Type_Quants}, there exists a $\frac{2\rho}{|A|\cdot|B|}$-quantized detailing (in the sense of Definition~\ref{def:quantized-detailing}) $\hat\eta$ of $\xi|_2=\eta|_1$ with respect to $B$ and a $\frac{\rho}{(1/\rho+1)^{|A|\cdot|B|}}$-quantized distribution $\tilde\Upsilon$ over $\{0,\rho,\ldots,1\}^{A\times B}$, for which $\dtv(\eta',\eta^*)\le \rho$ and $\dem^{\eta^*}(\tilde{\Upsilon},\Upsilon^*)\le \rho$.
	
	Now note that $\widetilde{\boldeta}(v)=0$ whenever $\hat\eta(v)=\xi|_2(v)=0$, and let $\eta'=\widetilde{\boldeta}\rhd\hat\eta$ be the adjustment of $\hat\eta$ to $\widetilde{\boldeta}$ (see Definition~\ref{def:adjust}). While $\hat\eta$ may not be a distribution that is considered by the loop in Step \ref{estim:loop} of $\EstimateDistance$, the distribution $\eta'$ (as a $\rho$-quantized detailing of $\widetilde{\boldeta}$) is considered there. Note that by Lemma~\ref{lem:adjust-dist} $\dtv(\hat\eta,\eta')\leq\rho$, and hence by the triangle inequality $\dtv(\eta',\eta^*)\leq2\rho$. To conclude the completeness proof, we will show that the algorithm will accept $\mu$ through the consideration of the pair $(\eta',\tilde\Upsilon)$.
	
	Using Lemma~\ref{cor:approxprob} we have that $\textsf{Accept-Probability}(s,q,\tilde\Upsilon,\eta', \frac{\eps_2-\eps_1}{12})$ deviates from the true acceptance probability by at most $1/20$. Therefore, since the canonical tester accepts $\tau$ with probability at least $2/3$, \textsf{Accept-Probability} returns a value which is at least $2/3-1/20\ge 1/2$, so the pair $(\eta',\tilde\Upsilon)$ passes Step~(\ref{step1}) in $\EstimateDistance$.
	
	It remains to bound $\dem^{\eta'}(\tilde{\bLambda}_{\langle B\rangle},\tilde\Upsilon)$.
	Note that using Lemma~\ref{lem:Translate} and the fact that $\mu^U$ is $(\delta,k')$-weakly robust, we have that $\mu^U$ is also $(2\delta,2^{k/\delta})$-robust. Therefore, letting $\Lambda_{\langle B\rangle}$ (respectively $\tilde\bLambda_{\langle B\rangle}$) be the flat extension of $\Lambda$ (respectively $\tilde{\bLambda}$) over $B$, from Lemma~\ref{lem:RobustCloseness} we have that $\dem^{\eta^*}(\Lambda_{\langle B\rangle},\Lambda^*)\le \sqrt{2\delta}$. In addition, since taking a flat extension does not change the distance, $\dem^{\eta^*}(\tilde{\bLambda}_{\langle B\rangle},\Lambda_{\langle B\rangle})=\dem^{\eta}(\tilde{\bLambda},\Lambda)\le \frac{\eps_2-\eps_1}{12}$. Using Lemma~\ref{lem:small_dtv}, Lemma~\ref{lem:DistStars} and the triangle inequality:
	\begin{align*}
		\dem^{\eta'}(\tilde{\bLambda}_{\langle B\rangle},\tilde\Upsilon)&\le \dem^{\eta^*}(\tilde{\bLambda}_{\langle B\rangle},\tilde\Upsilon)+2\rho\\
		&\le \dem^{\eta^*}(\tilde\bLambda_{\langle B\rangle},\Lambda_{\langle B\rangle})+\dem^{\eta^*}(\Lambda_{\langle B\rangle},\Lambda^*)+\dem^{\eta^*}(\Lambda^*,\Upsilon^*)+\dem^{\eta^*}(\Upsilon^*,\tilde\Upsilon)+2\rho\\
		&\le\frac{\eps_2-\eps_1}{12}+\sqrt{2\delta}+\eps_1+3\rho\\
		&\le\frac{\eps_2-\eps_1}{12}+ \frac{\eps_2-\eps_1}{5^{5/2}s^3q^{3/2}}+\eps_1 + \frac{\eps_2-\eps_1}{9sq}\\
		&\le \frac{\eps_2+\eps_1}{2}.
	\end{align*}
	Therefore, we have that (conditioned on the probability $2/3$ event described at the beginning of the proof), the algorithm \textsf{Tolerant-Tester} indeed accepts through $(\eta',\tilde\Upsilon)$, as required.
\end{proofof}

\subsection{Soundness (Proof of Lemma~\ref{lem:soundness})}\label{sec:soundness}
Next, we will show that $\EstimateDistance$ is sound. That is, if $\mu$ is $\eps_2$-far from the index-invariant property $\calP$, then it will be rejected by $\EstimateDistance$ with probability at least $2/3$.

We define a procedure (used later only for an existence proof) that given a detailing $\xi$ over $A$ of a distribution $\mu$, a detailing $\eta$ of $\xi|_2$ with respect to $B$, a ``target'' type distribution $\Lambda_t$ on $[0,1]^{A\times B}$, and an implementation $H:[n]\to \left([0,1]^{A\times B}\right)^2$ of a transfer function between the ``source'' type distribution $\Lambda_s=\Lambda_{\langle B\rangle}$ of $\xi_{\langle\eta\rangle}$ and $\Lambda_t$ which extends the implementation of $\Lambda_s$ that is demonstrated by $\xi_{\langle\eta\rangle}$, produces a sample from a detailing $\Xi$ of a transfer distribution between $\mu$ and some $\tau$, so that the corresponding detailing of $\tau$ with respect to $A\times B$ has type distribution $\Lambda_t$, and the distance between $\mu$ and $\tau$ is bounded by the average of the distances between the types provided in $H$.

\begin{figure}[ht!]
	\begin{framed}
		\noindent Procedure \textsf{Change-types}$(\xi, \eta ,H)$
		\begin{flushleft}
			\noindent {\bf Input:} A detailing $\xi $ of $\mu$ with respect to  $A$, a detailing $\eta$  of $\xi|_2$ with respect to $B$, and an implementation $H:[n]\to \left([0,1]^{A\times B}\right)^2$ of a transfer distribution between the type distribution $\Lambda_s=\Lambda_{\langle B\rangle}$ of $\xi_{\langle \eta \rangle}$ and a target type distribution $\Lambda_t$, such that $H$ extends the implementation of $\Lambda_s$ demonstrated by $\xi_{\langle\eta\rangle}$. \\
			{\bf Output:} A sample from a distribution $\Xi$ over $\zo^n\times\zo^n\times (A\times B)$ which is a detailing of a transfer distribution between $\mu$ and some $\tau$.
			\begin{enumerate}
				\item Draw $(\bx,\ba,\bb)\sim \xi_{\langle\eta\rangle}$.
				\item For every $i\in [n]$ independently, set $\by_i$ conditioned  on the value of $\bx_i$. If $\bx_i=1$, then set $\by_i=0$ with probability $\max\left\{0,\frac{(H(i))_1(\ba,\bb)- (H(i))_2(\ba,\bb) }{(H(i))_1(\ba,\bb)}\right\}$, and if $\bx_i=0$, then set $\by_i=1$ with probability $\max\left\{0,\frac{(H(i))_2(\ba,\bb) - (H(i))_1(\ba,\bb)}{1-(H(i))_1(\ba,\bb)}\right\}$.\label{proc:flip_bits}
				\item return $(\bx, \by, (\ba,\bb))$.
			\end{enumerate}
		\end{flushleft}\vskip -0.14in
	\end{framed}\vspace{-0.2cm}
	\caption{A description of the \textsf{Change-types} procedure} \label{fig:type change}
\end{figure}

We will later use \textsf{Change-types} to show that if $\EstimateDistance$ does not reject with high probability then there exists some distribution $\tau$ which is close to $\mu$ and is not rejectable by a $\frac{\eps_2-\eps_1}{12}$-test for $\calP$, meaning that $\mu$ itself is not very far from $\calP$. But first we need to prove its properties, starting with the following trivial observation (which follows from the construction of $\Xi$).

\begin{observation}
	The distribution $\Xi$ of the output of \textsf{Change-types}$(\xi, \eta ,H)$ satisfies $\Xi|_{1,3}=\xi_{\langle\eta\rangle}$, and in particular the type distribution of $\Xi|_{1,3}$ equals $\Lambda_{\langle B\rangle}$.
\end{observation}

The following shows that the detailing $\Xi|_{2,3}$ of the distribution $\tau=\Xi|_2$ indeed has the required target type distribution.

\begin{lemma} \label{lem:Procedure-types} The procedure \textsf{Change-types} produces a sample from a distribution $\Xi$ such that the detailing $\Xi|_{2,3}$ of $\tau$ admits the type distribution $\Lambda_t$.
\end{lemma}
\begin{proof} We prove for every $i$ that indeed $\Pr_{(\bx,\by)\sim\Xi|^{3:(a,b)}}[\by_i=1]=(H(i))_2(a,b)$, which means that the corresponding type for this coordinate is indeed $(H(i))_2$.
	\begin{align*}
		\Prx_{(\bx,\by)\sim\Xi|^{3:(a,b)}}[\by_i=1]=\Prx_{(\bx,\by)}[\by_i=1\mid \bx_{i}=0]\cdot \Prx_{(\bx,\by)}[\bx_i=0]+\Prx_{(\bx,\by)}[\by_i=1\mid \bx_{i}=1]\cdot \Prx_{(\bx,\by)}[\bx_i=1]. %
	\end{align*}
	In the case $(H(i))_1(a,b)\geq (H(i))_2(a,b)$, this gives
	\begin{align*}
		\Prx_{(\bx,\by)\sim\Xi|^{3:(a,b)}}[\by_i=1]=0+ (H(i))_1(a,b)\cdot \left(1-\frac{ (H(i))_1(a,b)- (H(i))_2(a,b)}{ (H(i))_1(a,b)}\right)=(H(i))_2(a,b).
	\end{align*}
	In the case $(H(i))_1(a,b)< (H(i))_2(a,b)$, this gives
	\begin{align*}
		(1-(H(i))_1(a,b))\cdot \frac{ (H(i))_2(a,b)- (H(i))_1(a,b)}{1- (H(i))_1(a,b)}+(H(i))_1(a,b)=(H(i))_2(a,b).
	\end{align*}
\end{proof}

The next lemma bounds the distance between $\mu=\Xi|_1$ and $\tau=\Xi|_2$.

\begin{lemma} \label{lem:Procedure-dist} The procedure \textsf{Change-types} produces a distribution $\Xi$ such that $$\dem(\Xi|_1,\Xi|_{2})\le \Ex_{i\sim[n]}\left[\Ex_{(\ba,\bb)\sim \eta}\left[   \Big|(H(i))_2(\ba,\bb)-(H(i))_1(\ba,\bb)\Big|\right]\right].$$
\end{lemma}
\begin{proof}
	Noting that $\Xi|_3=\eta$, we have that:
	\begin{align}
		\dem(\Xi|_1,\Xi|_2)&\le\Ex_{(\bx,\by)\sim\Xi|_{1,2}}[d_H(\bx,\by)]= \Ex_{(\ba,\bb)\sim \eta}\left[\Ex_{(\bx,\by)\sim \Xi |^{3:(\ba,\bb)}}\left[d_H(\bx,\by)\right]\right] \notag\\
		&=\Ex_{i\sim[n]}\left[\Ex_{(\ba,\bb)\sim \eta}\left[\Prx_{(\bx,\by)\sim \Xi |^{3:(\ba,\bb)}}[\bx_{i}=0\wedge \by_i=1]+\Prx_{(\bx,\by)\sim \Xi  |^{3:(\ba,\bb)}}[\bx_{i}=1\wedge\by_i=0]\right]\right].\label{eq:EMDupperboundSoundness1}
	\end{align}
	For a fixed $i\in[n]$ and $(a,b)\in A\times B$ we have that 
	\begin{align*}
		\Prx_{(\bx,\by)\sim \Xi |^{ 3:(a, b)}}[\bx_{i}=0\wedge\by_i=1]&=\Prx_{(\bx,\by)\sim \Xi |^{ 3:(a, b)}}[\by_i=1\mid \bx_{i}=0]\cdot\Prx_{(\bx,\by)\sim \Xi |^{ 3:(a, b)}}[\bx_{i}=0]\\
		&=\left(1-(H(i))_1(a,b) \right)\cdot \max\left\{0,\frac{(H(i))_2(a,b) - (H(i))_1(a,b)}{1-(H(i))_1(a,b)}\right\}\\
		&=\max\left\{0,{(H(i))_2(a,b)-(H(i))_1(a,b)} \right\},
	\end{align*} 
	and similarly,
	\[ 	\Prx_{(\bx,\by)\sim \Xi |^{3: (a, b)}}[\bx_{i}=1\wedge\by_i=0] =  \max\left\{0,(H(i))_1(a,b)-(H(i))_2(a,b)\right\}.\]
	Therefore,
	\[	\Prx_{(\bx,\by)\sim \Xi |^{ 3:(a, b)}}[\bx_{i}=0\wedge\by_i=1]+	\Prx_{(\bx,\by)\sim \Xi |^{ 3:(a, b)}}[\bx_{i}=1\wedge\by_i=0]=\Big|(H(i))_2(a,b)-(H(i))_1(a,b)\Big|.\]
	Plugging into Inequality (\ref{eq:EMDupperboundSoundness1}), we have the final result:
	\begin{align*}
		\dem(\Xi|_1,\Xi|_2)&\le\Ex_{i\sim[n]}\left[\Ex_{(\ba,\bb)\sim \eta}\left[   \Big|(H(i))_2(\ba,\bb)-(H(i))_1(\ba,\bb)\Big|\right]\right].
	\end{align*}
\end{proof}

The final and rather important property of the output distribution of \textsf{Change-types} is the preservation of goodness. Namely, that if $\xi$ was an $(\eps,q)$-good detailing of the unknown input distribution $\mu$, then $\Xi|_{2,3}$ will be an $(\eps,q)$-good detailing of the ``new'' $\tau$.

\begin{lemma}\label{lem:eps_good_mapping} Fix $\epsilon\in(0,1)$ and $q\in \N$. Let $\xi$ be a detailing of $\mu$ with respect to $A$, and let $\eta$ be a distribution over $A\times B$ which is a detailing of $\xi|_2$ with respect to $B$. If $\xi$ is $(\eps,q)$-good, then the detailing $\Xi|_{2,3}$ of $\tau=\Xi|_2$, as obtained by the output distribution of \textsf{Change-types}, is also $(\epsilon,q)$-good.
\end{lemma}

\begin{proof}
	Since $\xi$ is $(\epsilon,q)$-good, there exists $J\subseteq A$ with $\Prx_{\xi|_2}[J]>1-\eps$ such that for any $a\in J$, at least $(1-\epsilon)n^q$ of the $q$-tuples in $[n]^q$ are $\epsilon$-independent. Consider the flat extension  $\xi_{\langle \eta \rangle}$ with respect to $\eta$. Since $\xi_{\langle \eta \rangle}$ is a flat extension,  for any $a\in J$ and $\alpha\in[n]^q$ which is $\epsilon$-independent with respect to $\xi|_1^{2:a}$, the tuple $\alpha$ is also $\epsilon$-independent with respect to $\xi_{\langle \eta\rangle}|_1^{2:(a,b)}=\xi|^{2:a}$ for all $b\in B$. 
	
	Next we claim that if a tuple $\alpha\in[n]^q$ is $\eps$-independent with respect to $\xi_{\langle \eta \rangle}|_1^{2:(a,b)}$, then it is also $\eps$-independent with respect to $\nu\eqdef \Xi|^{3:(a,b)}_2$. Let $(j_1,\ldots,j_q)\in [n]^q$ be an $\epsilon$-independent tuple with respect to $\xi_{\langle \eta \rangle}|_1^{2:(a,b)}$. Then, we have that 
	
	\[\dtv\left(\left(\xi_{\langle \eta \rangle}|_1^{2:(a,b)}\right)|_{\{j_1,\ldots,j_q\}}, \prod_{\ell\in[q]}\left(\xi_{\langle \eta \rangle}|_1^{2:(a,b)}\right)|_{j_\ell}\right)\le \eps.\]
	
	Let $\sigma$ denote an optimal coupling from the distribution $\left(\xi_{\langle \eta \rangle}|_1^{2:(a,b)}\right)|_{\{j_1,\ldots,j_q\}}$ to the corresponding product distribution $\prod_{\ell\in[q]}\left(\xi_{\langle \eta \rangle}|_1^{2:(a,b)}\right)|_{j_\ell}$, in the sense that
	
	\[\dtv\left(\left(\xi_{\langle \eta \rangle}|_1^{2:(a,b)}\right)|_{\{j_1,\ldots,j_q\}}, \prod_{\ell\in[q]}\left(\xi_{\langle \eta \rangle}|_1^{2:(a,b)}\right)|_{j_\ell}\right)=\Ex_{(\bx,\bx')\sim \sigma}\left[\indi_{\{\bx\neq \bx'\}}\right]\]
	We construct a coupling $\sigma'$ from the distribution $\nu|_{\{j_1,\ldots,j_q\}}$ to the corresponding product distribution $\prod_{\ell\in[q]}\nu|_{j_\ell}$ and and prove that 
	\[\Ex_{(\by,\by')\sim \sigma'}\left[\indi_{\{\by\neq \by'\}}\right]\le\Ex_{(\bx,\bx')\sim \sigma}\left[\indi_{\{\bx\neq \bx'\}}\right].\]
	
	We define the coupling $\sigma'$ as follows. We first sample $(\bx,\bx') \sim\sigma$. Then for every $\ell\in [q]$, if $\bx_{j_\ell}=\bx'_{j_\ell}$ then we apply Step~(\ref{proc:flip_bits}) of \textsf{Change-types} for $\bx_{j_\ell}$ to obtain $\by_{j_\ell}$ and set  $\by'_{j_\ell}=\by_{j_\ell}$ (this means that we ``use the same random coins'' for handling $\bx_{j_\ell}$ and $\bx'_{j_\ell}$). If $\bx_{j_\ell}\neq \bx'_{j_\ell}$ then we just apply Step~(\ref{proc:flip_bits}) of \textsf{Change-types} separately for $\bx_{j_\ell}$ to obtain $\by_{j_\ell}$ and for $\bx'_{j_\ell}$ to obtain $\by'_{j_\ell}$ (note that in this case only one of the two applications uses randomness). The last two things to note are that $\by$ has the same distribution as the output of \textsf{Change-types} over $\bx$ restricted to $\{j_1,\ldots,j_q\}$, while $\by'$ distributes as the corresponding product, and that under this process $\bx=\bx'$ implies $\by=\by'$. Hence $\Ex_{(\by,\by')\sim\sigma'}[\indi_{\{\by\neq \by'\}}]\leq \Ex_{(\bx,\bx')\sim\sigma}[\indi_{\{\bx\neq \bx'\}}]$, as required.
	
	Since the above argument holds for any $\eps$-independent tuple with respect to $\xi_{\langle \eta \rangle}|_1^{2:(a,b)}$ for any $(a,b)\in A\times B$ (showing the tuple to be $\eps$-independent also for $\Xi|_2^{3:(a,b)}$), it implies that $(\eps,q)$-goodness is indeed transferred from $\xi_{\langle\eta\rangle}$ to $\Xi|_{2,3}$.
\end{proof}

Our application to this lemma would be to show that predictability still holds for the resulting detailing of $\tau$ (with respect to the weight distribution $\eta$ and the type distribution $\Lambda_t$), just as it held for a good enough detailing $\xi$ of $\mu$.

\begin{lemma}\label{lem:procedure_pred}Fix  $\epsilon,\gamma\in (0,1)$, $s,q\in \N$, a detailing $
	\xi$ of $\mu$ with respect to $A$, and a distribution $\eta$ over $A\times B$ extending $\xi|_2$. Let $\delta=\frac{\epsilon^6}{10^5\cdot q^3\cdot s^6}$, let $\textsf{acc}_{\gamma}(\Xi|_2)$ denote the acceptance probability of the canonical tester with proximity $\gamma$ applied on $\Xi|_2$, where $\Xi$ is the output distribution of \textsf{Change-types}, and denote by $\Upsilon$ the type distribution of $\Xi|_{2,3}$. If $\xi$  is $(\delta,q-1)$-weakly robust, then the output $\widetilde{\textsf{acc}}_{\gamma}(\Xi|_2)$ of the procedure $\textsf{Accept-Probability}(s,q,\eta,\Upsilon,\gamma)$ satisfies
	\[\left|   \widetilde{\textsf{acc}}_{\gamma}(\Xi|_2)- {\textsf{acc}}_{\gamma}(\Xi|_2) \right|\le\eps/20.\]
\end{lemma}

\begin{proof}
	Since $\xi$ is $(\delta,\epsilon)$-weakly robust, by Lemma~\ref{lem:weakly_robust_is_good} it is also $\left(\epsilon', q\right)$-good for some $\epsilon'\le\frac{\eps}{1000 s}$. By Lemma~\ref{lem:eps_good_mapping}, we have that the detailing $\Xi|_{2,3}$ is also $(\epsilon',q)$-good. Therefore, we can use Lemma~\ref{cor:approxprob} to conclude that 	$\left|   \widetilde{\textsf{acc}}_{\gamma}(\Xi|_2)- {\textsf{acc}}_{\gamma}(\Xi|_2) \right|\le \epsilon/20$, as required.
\end{proof}

\begin{proofof}{Lemma~\ref{lem:soundness}}
	By the guarantees of \textsf{Find-Weakly-Robust-Detailing} (Lemma~\ref{lem:findweakly}) we have that with probability at least $5/6$ the returned set $U$ defines a $(\delta,k')$-weakly robust detailing $\mu^U$ of $\mu$. In addition, from the guarantees of  $\textsf{Estimate-Parameters}$ (Lemma~\ref{theo:estparamalgnewcorrectness}) we have that $\dtv(\tilde{\boldeta},\eta)\le\rho$ and $\dem^{\eta}(\tilde{\bLambda},\Lambda)\le \rho$ with probability at least $5/6$, where $\eta=\mu|_U$ is the weight distribution of $\mu^U$ and $\Lambda$ is its type distribution. We henceforth condition on the intersection of the above events (which happens with probability at least $2/3$).
	
	We will show that if the algorithm accepts, then there exists a distribution $\tau^*$, which is accepted by the canonical test with probability larger than $1/3$ (and hence satisfies $\dem(\tau^*,\calP)\leq\frac{\eps_2-\eps_1}{12}$), satisfying also $\dem(\mu,\tau^*)\le \frac{11\eps_2+\eps_1}{12}$. By the triangle inequality, this will imply that $\dem(\mu,\calP)\le \eps_2$.
	
	Let us denote $A=\zo^{|U|}$, as in the proof of Lemma~\ref{lem:completenes}. Note that if $\EstimateDistance$ accepts, then there exist a $\frac{2\rho}{|A|\cdot |B|}$-quantized detailing $\eta'$ of $\widetilde{\boldeta}$ with respect to $A\times B$, and a $\frac{\rho}{(1+1/\rho)^{|A|\cdot |B|}}$-quantized type distribution $\tilde{\Upsilon}$ over $\{0,\rho,\ldots,1\}^{A\times B}$, that satisfy  $\dem^{\eta'}(\tilde\bLambda_{\langle B\rangle},\tilde{\Upsilon})\le (\epsilon_1+\eps_2)/2$ and such that \textsf{Accept-Probability}$(s,q,\tilde\Upsilon,\eta', \frac{\eps_2-\eps_1}{12})$ returns a value that is at least $1/2$.
	
	Noting that in particular $\widetilde{\boldeta}(u)=0$ whenever $\mu|_U(u)=0$, let $\eta^*=\mu|_U\rhd\eta'$ be the adjustment of $\eta'$ to $\mu|_U=\mu^U|_2$, and let
	$\mu^U_{\langle\eta^*\rangle}$ be the flat refinement of $\mu^U$ with respect to $\eta^*$. By Lemma~\ref{lem:adjust-dist} we have $\dtv(\eta',\eta^*)\leq\rho$. Next we apply Lemma~\ref{lem:Type_Quants}, assuming $n\ge \frac{(1/\rho+1)^{2^{|U|}\cdot |B|}}{\rho}=2^{\poly(2^q,s,1/(\eps_2-\eps_1))}$, to obtain a distribution $\Upsilon^*$ which is $\frac{1}{n}$-quantized such that $\dtv(\Upsilon^*,\tilde{\Upsilon})\le \rho$.
	
	By Lemma~\ref{lem:small_dtv} we have $\dem^{\eta^*}({\Lambda}_{\langle B \rangle},\Upsilon^*)\leq \dem^{\eta'}({\Lambda}_{\langle B \rangle},\Upsilon^*)+2\rho$. Now let $h:[n]\to \{0,\rho,\ldots,1\}^{A\times B}$ the implementation of $\Lambda_{\langle B\rangle}$ demonstrated by $\mu^U_{\langle\eta^*\rangle}$. We next note that we can apply Lemma~\ref{lem:1/n-quantized-optimal} and obtain an implementation  $H:[n]\to  \left(\{0,\rho,\ldots,1\}^{A\times B}\right)^2$  of an optimal transfer function $\kappa$ between $\Lambda_{\langle B \rangle}$ and $\Upsilon^*$ over $\eta^*$, with $(H(i))_1=h(i)$ for all $i\in [n]$. Thus, 
	\[\dem^{\eta^*}({\Lambda}_{\langle B \rangle},\Upsilon^*)= \Ex_{i\sim[n]}\left[\Ex_{(\ba,\bb)\sim \eta^*}\left[   \big| (H(i))_2(\ba,\bb) - (H(i))_1(\ba,\bb)\big|\right]\right]\]
	
	We construct a distribution $\tau^*$ along with a detailing $\Xi^*$ of a transfer distribution $T^*$ from $\mu$ to $\tau^*$ with respect to $A\times B$ as follows. Recall that $\mu^U$ denotes the detailing of $\mu$ obtained in Step~\ref{step:1} of the algorithm. 
	
	We define $\Xi^*$ as the distribution over the output of \textsf{Change-types} with respect to $\mu^U$, $\eta^*$ and $H$. That is, a sample $(\bx,\by,(\ba,\bb))\sim \Xi^*$ is obtained by calling $\textsf{Change-types}(\mu^U,\eta^*,H)$. We will not take actual samples from $\Xi^*$, but only analyze the distance bounds that it implies.
	
	Observe that by construction the distribution $\tau^*=\Xi^*|_2$ is supported on $\zo^n$, and by Lemma~\ref{lem:Procedure-types} the detailing $\zeta=\Xi^*|_{2,3}$ of $\tau^*$ admits the target weight distribution $\eta^*$ and type distribution $\Upsilon^*$.

	Then, by Lemma~\ref{lem:Procedure-dist}, using also the facts that $\dem^{\eta'}(\tilde\Upsilon,\Upsilon^*)\leq \dtv(\tilde\Upsilon,\Upsilon^*)$ and $\dem^{\eta^*}(\Lambda_{\langle B \rangle},\widetilde\bLambda_{\langle B \rangle})=\dem^{\mu|_U}(\Lambda,\widetilde\bLambda)$ we have that 
	\begin{align*}
		\dem(\mu,\tau^*)&\le\Ex_{i\sim[n]}\left[\Ex_{(\ba,\bb)\sim \eta^*}\left[   \left|(H(i))_2(\ba,\bb)-(H(i))_1(\ba,\bb)\right|\right]\right]=\dem^{\eta^*}(\Lambda_{\langle B \rangle},\Upsilon^*)\\
  &\le \dem^{\eta^*}(\Lambda_{\langle B \rangle},\widetilde\bLambda_{\langle B \rangle}) + \dem^{\eta^*}(\widetilde\bLambda_{\langle B \rangle},\Upsilon^*) \\
   &\le \dem^{\eta^*}(\Lambda_{\langle B \rangle},\widetilde\bLambda_{\langle B \rangle}) + \dem^{\eta'}(\widetilde\bLambda_{\langle B \rangle},\Upsilon^*) +2\rho \\
  &\le \dem^{\mu|_U}(\Lambda,\widetilde\bLambda)+\dem^{\eta'}(\widetilde\bLambda_{\langle B \rangle},\tilde\Upsilon)+ \dtv(\tilde\Upsilon,\Upsilon^*)+2\rho\\
		&\le \rho+\frac{\eps_1+\eps_2}{2}+\rho+2\rho\leq\frac{11\eps_2+\eps_1}{12}
	\end{align*}
 \color{black}

	It remains to show that $\tau^*$ is accepted by the canonical tester (with proximity parameter $\frac{\eps_2-\eps_1}{12}$) with probability greater than $1/3$.
	Indeed, by our choice of parameters and using Lemma~\ref{lem:procedure_pred}, we have that 
	\[\left|   \widetilde{\textsf{acc}}_{\frac{\eps_2-\eps_1}{12}}(\tau^*)- {\textsf{acc}}_{\frac{\eps_2-\eps_1}{12}}(\tau^*) \right|\le 1/20,\] which implies that $\tau^*$ is accepted by the canonical tester with probability at least $1/2-1/20> 1/3$, and the proof is complete.
\end{proofof}

\bibliographystyle{alpha}
\bibliography{amit.bib}

\end{document}